\documentclass{article}

\PassOptionsToPackage{numbers, compress}{natbib}

\usepackage[main,final]{neurips_2026}

\usepackage[utf8]{inputenc} %
\usepackage[T1]{fontenc}    %
\usepackage{hyperref}       %
\usepackage{url}            %
\usepackage{booktabs}       %
\usepackage{amsfonts}       %
\usepackage{nicefrac}       %
\usepackage{microtype}      %
\usepackage{xcolor}         %
\usepackage{enumitem}
\usepackage{amsmath}
\usepackage{amsthm}
\usepackage{amssymb}
\usepackage{graphicx}
\usepackage{sidecap}
\usepackage{xcolor}
\usepackage{enumitem}
\usepackage{subcaption}

\usepackage{fvextra} %
\newtheorem{theorem}{Theorem}
\newtheorem{lemma}{Lemma}
\newtheorem{corollary}{Corollary}
\newtheorem{conjecture}{Conjecture}
\title{Tapes Together Strong: The Co-evolution of Computation and Cooperation}

\author{%
  Kunal Jha$^{1,2\dagger}$ \And
  Francesco Cicala$^2$ \And
  Blaise Agüera y Arcas$^2$ \AND
  Blake Aaron Richards$^{2,3}$ \And
  Natasha Jaques$^{1,4*}$ \And
  Max Kleiman-Weiner$^{1,4*}$ \And
  Eyvind Niklasson$^{2*}$
}
\begin{document}
\maketitle

\begingroup
\renewcommand{\thefootnote}{}
\footnotetext{%
  \kern-1.8em
  $^\dagger$Work done as a student researcher at Google. \quad
  $^*$Equal advising. \quad
  $^1$Dept. of Computer Science, University of Washington. \quad
  $^2$Google Paradigms of Intelligence. \quad
  $^3$School of Computer Science, McGill University; Mila. \quad
  $^4$Google DeepMind.
}
\endgroup

\begin{abstract}
How does cooperation evolve in complex agentic systems?
Prior work in evolutionary game theory studies why individuals are incentivized to cooperate by isolating social interactions from the physical costs of behavior, while artificial life models traditionally study emergent self-replication without formalizing the dilemma between acquiring resources and preserving the shared energy needed to reproduce.
In contrast, we introduce \textbf{Autopoietic Game Theory}, a computational model where social interactions, replication mechanisms, and their associated computational costs are endogenous and simultaneously co-evolving. We study these dynamics using a computational substrate of randomly initialized programs in Z80 machine code, showing empirically, and motivating with a simplified theoretical model, that embedding a social dilemma directly into the physics of computation can favor the emergence of self-replicating, cooperative strategies. When resources are scarce, our analysis shows that \textit{defection can become self-limiting even in well-mixed populations}: parasitic stealing destroys shared energy, slows execution, and can prevent reliable replication. Empirically, evolved programs suppress stealing across several Z80 environments, while spatial assortment further supports structural complexity and task performance. We further show that the framework can incorporate exogenous pressures, such as math tasks structured as sequential social dilemmas, when rewards are tied to computation budgets. These results suggest that coupling an agent's capacity for computation to its available energy transforms cooperation into a dominant scaffolding for building sustainable, self-organizing systems.
\end{abstract}

\section{Introduction}
\vspace{-0.5em}

The transition from individual agents to cooperative collectives is a central puzzle in artificial intelligence, biology, and the study of complex adaptive systems \citep{y2025intelligence,laland2017darwin,hull2004niche}. Historically, the emergence of sociality has been studied through disparate lenses: biological scaling theories \citep{dunbar2003social,dunbar2007evolution}, ecological models of resource sharing \citep{muthukrishna2018cultural,henrich2007humans,muthukrishna2023theory}, game-theoretic formalizations of reciprocal altruism \citep{axelrod1981evolution,bear2016intuition,rand2013human}, and cognitive models of multi-agent intent \citep{ullman2009,baker2017rational,zhi2020online,jha2024neural,kleiman2016coordinate,kleiman2020downloading,Kleiman-Weiner2025-jr,meulemans2025embeddeduniversalpredictiveintelligence}. Today, as machine learning scales toward autonomous multi-agent systems that can modify their decision-making logic while completing complex tasks \cite{Steinberger2026openclaw,zhang2026hyperagents,hu2025automateddesignagenticsystems}, understanding why and how self-interested agents cooperate without centralized control is a pressing engineering challenge \cite{dafoe2020openproblemscooperativeai,leibo2025pragmaticviewaipersonhood}. This problem becomes especially difficult when computation, interaction, and reproduction all draw from the same finite resource.

Consider a simple organism that acquires energy from its environment and spends it to act, interact, and reproduce. Because these behaviors share a common budget, energy spent on one action cannot be spent on another. The organism might acquire more by consuming a neighbor, but doing so also costs energy and recovers only part of what the neighbor held. Its social behavior therefore directly affects whether it retains enough energy to survive and reproduce.

Existing approaches often separate such strategic decisions from their physical execution. Evolutionary Game Theory (EGT), for example, models how populations adopt behaviors in social dilemmas but treats game payoffs and population updates as external mathematical operations \cite{smith1982evolution,Weibull1995-zu,Roca_2009}. Because EGT assumes that strategy replication occurs independently of interaction, stabilizing cooperation is notoriously difficult in well-mixed populations, where agents are paired at random, or among memoryless agents that cannot recall previous encounters. Overcoming these baseline conditions typically requires endowing agents with specific cognitive capabilities, such as memory for direct reciprocity or intentionality for costly punishment \citep{nowak2006five}. Conversely, Artificial Life (ALife) frameworks study how replicators evolve while competing for computational resources and physical space \citep{ray1992evolution,adami1994evolutionary,arcas2024computational}. Historically, however, these frameworks have emphasized zero-sum competition, ecological niche construction, or extrinsic task landscapes rather than formalizing emergent interactions as endogenous, general-sum social dilemmas \citep{coEvolutionBio,Dolson148973,flask,coreWars}.

How might cooperation evolve if agents could modify both their social strategies and the mechanisms by which they and others replicate? To answer this, we introduce \textbf{Autopoietic Game Theory} (Figure~\ref{fig:substrate}). Drawing on the biological principle of autopoiesis---the continuous reproduction and maintenance of a system's components---our framework treats strategic behavior and replication as simultaneous, co-evolving computational processes. It translates the organism analogy above into a computational setting: every instruction has an energetic cost, execution speed depends on the energy available, and computation, interaction, and reproduction all draw from the same budget \citep{Landauer1961-sj,Margolus_1998}. This allows us to test whether changes in replication alter the outcome of a social interaction and whether cooperation gives agents a competitive advantage by supporting more complex computation over time.

We instantiate this framework in a Z80 machine-code environment where self-replicating programs emerge from randomly initialized bytes and manage finite energy budgets. Rather than imposing an abstract payoff matrix, we embed a social dilemma in the mechanics of computation by introducing \texttt{STEAL} as a computational instruction. When executed, \texttt{STEAL} imperfectly transfers energy from a neighboring agent to the acting program, immediately increasing its computational capacity while reducing the total energy available in the environment. Programs must therefore allocate energy among executing code, interacting with neighbors, and reproducing. By allowing social strategy and replication to co-evolve within a single phase of interaction, our framework lets us study whether cooperation can persist among well-mixed, memoryless computational agents. Concretely, our contributions are:

\begin{itemize}[leftmargin=*]
\vspace{-0.5em}

    \item \textbf{Autopoietic Game Theory:}
    We introduce a computational framework in which social strategy, replication mechanisms, and computational resource use are endogenous processes that evolve together rather than being specified as separate update rules.
\vspace{-0.3em}

    \item \textbf{Cooperation Without Memory or Assortment:}
    We show that randomly initialized programs can evolve self-replication while suppressing destructive stealing, even when partners are chosen at random and agents cannot remember previous interactions. Mutation-free invasion assays and parameter sweeps show that this suppression is not simply an artifact of continuous mutation or inefficient stealing. Instead, cooperation is supported by the coupling between social interaction, energy, and the agents' evolving replication mechanisms.
\vspace{-0.3em}

    \item \textbf{Spatial Assortment Supports Complexity and Resilience:}
    When energy is uniformly distributed, both well-mixed and spatial populations suppress stealing, but local interactions produce more structurally complex and genetically cohesive replicators. When energy is distributed asymmetrically, spatial assortment becomes critical: local populations maintain energy and complexity, whereas well-mixed populations collapse into simpler, lower-energy programs.
\vspace{-0.3em}

    \item \textbf{Cooperation Extends to Multi-Step Computational Tasks:}
    We show that resource scarcity drives evolving programs to solve externally imposed tasks while continuing to replicate. When rewards form a sequential social dilemma, programs suppress both destructive stealing and the safer solo strategy in favor of a collaborative, multi-step solution. Cooperation persists in this more complex setting and is strengthened by local interactions.

\end{itemize}

\vspace{-0.5em}
\section{Related Work}
\vspace{-0.5em}

\textbf{Digital Self-Replication and Artificial Life.} The computational study of self-replication is anchored in \citeauthor{von1966theory}'s \cite{von1966theory} formalization of the Universal Constructor, which demonstrated that machines could use self-descriptions to build copies of themselves. Emergent digital replicators were pioneered by \textit{Tierra} \cite{ray1992evolution}, demonstrating spontaneous parasitism in a shared memory space, while \textit{Avida} \cite{adami1994evolutionary} linked computational logic to exogenously defined metabolic rewards \cite{lenski2003evolutionary} (which we discuss further in Appendix~\ref{appendix:avida_comp}). Moving beyond pure machine code, environments like \textit{Polyworld} \cite{Yaeger1997ComputationalGP} embedded neural architectures within spatial, energy-constrained environments, while formalizations of artificial life stressed the necessity of environmental pressures to drive open-ended evolution \cite{sorosIdentifying}. More recently, \citeauthor{arcas2024computational} \cite{arcas2024computational} demonstrated that self-replication spontaneously emerges as an attractor state in random instruction sets. However, these foundational models explore how ecological complexity and self-replication emerge under metabolic constraints, rather than on how replicators evolve in general-sum endogenous social dilemmas. %
We instead embed randomly initialized programs in an energy-limited computational substrate \cite{Bennett1982-hd}, where execution, interaction, and replication all draw from a single, unified budget. In this setting, we find that self-replicating programs spontaneously discover and stabilize non-destructive strategies, even when both replication and cooperation require multi-step computations, and when energy must be earned by solving collaborative downstream tasks.

\textbf{Cognitive Evolution, Multi-Agent AI, and Program Equilibria.}
In cognitive science and anthropology, sociality is often understood as an adaptation to metabolic constraints rather than a hardcoded trait. The \textit{Cultural Brain Hypothesis} \cite{muthukrishna2018cultural, henrich2015secret, muthukrishna2023theory} and \textit{Expensive Tissue Hypothesis} \cite{aiello1995expensive} propose that energy-intensive neural scaling favored cooperation through shared intentionality and distributed foraging \cite{tomasello2007shared, kaplan2000theory}. A similar view appears in developmental biology, where
competition among genetically identical cells and organs for finite
metabolic and informational resources coordinates growth and
morphogenesis \cite{Gawne2020-lh, Smiley2022-rs}. In contrast, modern AI and standard Evolutionary Game Theory \cite{smith1982evolution, Sandholm2010-je, Weibull1995-zu} typically treat replication as an external update rather than a resource-constrained action. Classical mechanisms such as assortment, reciprocity, and punishment can stabilize cooperation \cite{Nowak1992-qb, nowak2006five}, while Multi-Agent Reinforcement Learning studies sequential social dilemmas with spatial and temporally extended resource constraints \cite{leibo2017multiagentreinforcementlearningsequential, jaques2019socialinfluenceintrinsicmotivation, baker2020emergentreciprocityteamformation}.

Some evolutionary game-theoretic models examine how replication dynamics affect cooperation \citep{Ohtsuki2006-fm, Cardillo_2010, Perc_2010}, and recent work extends the spatialized
iterated Prisoner's Dilemma to players that merge and change scale as a
model of multicellularity \cite{shreesha}. These models generally restrict
social and reproductive behavior to hand-crafted strategy spaces. With executable agents, cooperation, defection, and replication can instead emerge as multi-step computations. This connects to ``program equilibria,'' in which agents coordinate by inspecting one another's logic \cite{tennenholtz2004program, critch2019parametric, sistla2025evaluating}. In open-source game theory, agents design programs that condition on others' decision-making logic before cooperating. Our model extends this setting: agents continuously edit, overwrite, and evolve one another's strategies; strategies are endogenous rather than fixed labels; and acting itself requires computation. Related environments study programs that compete by overwriting one another in zero-sum spaces \cite{coreWars}, further showing how code-level mechanisms shape strategic behavior.

Recent work on open-ended self-improvement and autonomous agent design studies systems that generate and organize increasingly complex behaviors \cite{hu2025automateddesignagenticsystems, zhang2026hyperagents, dochkina2026drophierarchyrolesselforganizing, liu2026evoxmetaevolutionautomateddiscovery}, but generally assumes large computational budgets. We instead require agents to evolve social behavior and replication from scratch while spending finite energy to compute, interact, and reproduce. Cooperative outcomes and reproduction dynamics resembling assumed population-update processes, such as the Moran process, are therefore not guaranteed. This coupling lets us study cooperation when strategic choice, computation, and self-maintenance are part of the same evolving process.

\begin{figure}
    \centering
    \includegraphics[width=\linewidth]{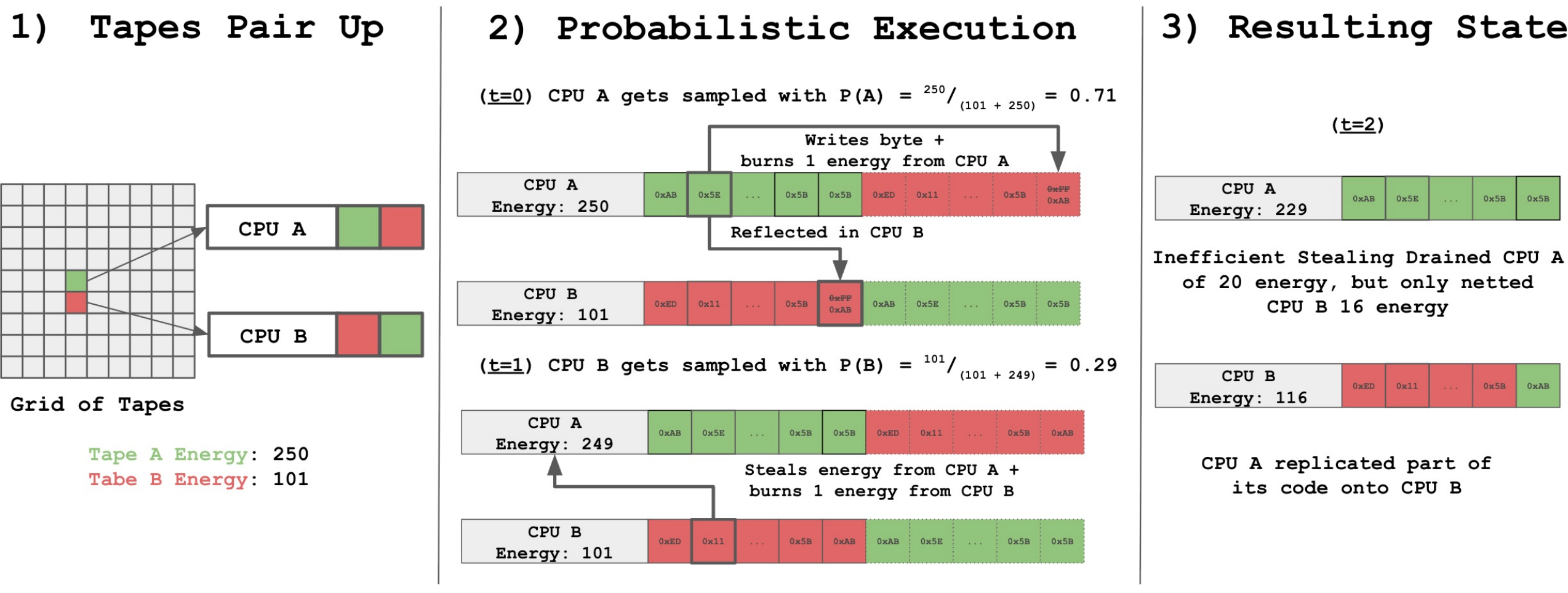}
    \caption{Visualization of the computational substrate. Two paired programs share a cyclic memory tape and have finite energy budgets that determine their individual probability of executing the next instruction. This incentivizes tapes to steal the energy of their partner, however, inefficient stealing degrades the computational budget of the entire system.}
    \label{fig:substrate}
    \vspace{-2em}
\end{figure}

\vspace{-0.5em}
\section{Autopoietic Game Theory: Embedding Social Dilemmas in Machine Code}
\label{sec:setup}
\vspace{-0.5em}

Here, we formally define the computational environment, where computation consumes limited resources available to each agent (illustrated in Figure~\ref{fig:substrate}). We simulate a finite population of 16,384 programs (which we interchangeably refer to as ``tapes'') written in a modified Z80 assembly language, initialized as 32-byte sequences of completely random bytes. Tapes from this population are paired with each other during interaction. The population lives either in a well-mixed pool (uniform random pairings) or on a 2D spatial grid (pairings restricted to the four immediate horizontal and vertical neighbors). In Appendix~\ref{appendix:implementation_details}, we provide a more detailed description of the setup and interactions.

\textbf{Epochs, Shared Memory, and Environmental Energy:} Time progresses in discrete epochs, with programs persisting in the population across epochs. During each epoch, every program receives a baseline energy injection ($\epsilon = 24$) and is paired with another program. When two 32-byte programs are paired, they are concatenated to form a single 64-byte cyclic shared memory tape. To ensure symmetric interactions, each program is assigned an independent CPU with a completely separate set of registers. Each CPU begins execution at the first byte of its own code and accesses the shared memory relative to that starting position. For example, if Program A is paired with Program B, the CPU for Program A starts at the beginning of A and perceives the cyclic tape as the continuous sequence $A \rightarrow B \rightarrow A \rightarrow B \dots$. Simultaneously, the CPU for Program B starts at the beginning of B and perceives the cyclic tape as the continuous sequence $B \rightarrow A \rightarrow B \rightarrow A \dots$.

\textbf{Asynchronous Execution and Relative Speed:} Both programs execute instructions asynchronously. Every operation costs exactly $c=1$ unit of energy; if a program attempts an invalid instruction, the emulator safely advances the program counter, but the energy cost is still consumed. %
Furthermore, a program's relative execution speed is dictated by its share of the pool: the probability that program $i$ executes the next instruction instead of program $j$ is exactly $\frac{E_i}{(E_i + E_j)}$. Each program can have at most $E_{max} = 255$ units of energy at any given time and at least $E_{min} = 0$ units of energy.

\textbf{Replication and Mutation:} Programs lack a built-in replication command; to reproduce, they must evolve code which explicitly overwrites the opponent's tape segment with their own code. Because an efficient replicator in this substrate can be handcrafted in just $L=11$ bytes, the baseline $\epsilon = 24$ provides sufficient energy for non-stealing tapes to safely replicate and accumulate surplus across epochs ($\epsilon > Lc$) assuming a replicator $L$ operations long costing $c$ units each). At the epoch's conclusion, the joint tape is cleanly cleaved in half, and the resulting programs are placed back in the soup. Before every epoch begins, all programs are subject to a background mutation rate $\mu$, representing the probability of mutating a single byte. To track evolutionary dynamics within our environment, we rely on two primary metrics. First, we measure structural complexification using \textit{higher-order entropy}---the theoretical difference between a sequence's Shannon entropy and its normalized Kolmogorov complexity (Appendix~\ref{appendix:metrics}). By approximating Kolmogorov complexity via the compressed size of the soup, this metric isolates non-trivial structural motifs from the baseline noise of independent random bytes (larger values suggest more complex replicators). Second, we evaluate genetic cohesion using \textit{edit distance} (the Hamming distance between program tapes).

\textbf{The Energy-Based Social Dilemma:} To embed a social dilemma directly into the mechanics of replication, we introduce the \texttt{STEAL} instruction. When performed, \texttt{STEAL} attempts to drain $\delta$ energy from the partner and absorbs that energy with efficiency $\alpha = 0.8$. Since $\alpha < 1$, stealing is lossy: the stealer gains only $\alpha\delta$, while the partner loses $\delta$, so each successful steal destroys $(1-\alpha)\delta$ energy from the system.

\begin{table}[b]
\vspace{-1em}
\centering
\small

\begin{tabular}{c|cc}
 & \textbf{Partner $C$} & \textbf{Partner $D$} \\
\hline
\textbf{$C$} & $\epsilon - Lc$ & $\epsilon - Lc - \delta$ \\
\textbf{$D$} & $\epsilon - Lc + \alpha\delta$ & $\epsilon - Lc - (1-\alpha)\delta$
\end{tabular}
\vspace{0.25em}
\caption{
Energy payoff matrix for the row player's tape. Each entry gives $\Delta E$ after one interaction, including baseline environmental energy $\epsilon$, replication cost $Lc$, and any energy transferred by \texttt{STEAL}.
}
\label{tab:energy-dilemma}
\end{table}

For simplicity of analysis, we use a coarse behavioral classification. Programs that do not execute \texttt{STEAL} during replication are labeled \textit{non-stealing} or \textit{cooperative} ($C$), while programs that execute \texttt{STEAL} are labeled \textit{defective} or \textit{parasitic} ($D$). However, we note in our substrate, even non-stealing programs reproduce by overwriting memory, so cooperation does not mean harmlessness in general. We interpret it as avoiding energy-destructive theft. This is still a meaningful cooperative act because energy left on the partner tape remains available for that partner's future computation.

In particular, if we assume $2\epsilon > 2Lc$, then two non-stealing tapes accumulate net system energy during an interaction: together they gain $2\epsilon - 2Lc > 0$ before any stealing occurs. We recognize, however, that this binary is a coarse-graining of a richer gradient; programs might loop over the \texttt{STEAL} instruction, vary their stealing magnitudes, or conditionally overwrite different portions of the shared tape. Charting this gradient is difficult in an open-ended environment, but we provide our initial attempts to do so in Appendix~\ref{appendix:RepDilemma}.

Table~\ref{tab:energy-dilemma} makes the dilemma explicit. Against either partner, stealing provides a local advantage by yielding more immediate energy than not stealing. However, mutual stealing degrades the shared energetic substrate, as each successful steal destroys $(1-\alpha)\delta$ energy that could otherwise support future computation. We exclude $\alpha = 1$ from this study; lossless stealing acts as a zero-sum transfer that eliminates the dilemma, shifting the evolutionary dynamics entirely toward competition \cite{adami1994evolutionary, Dolson148973}.

\vspace{-0.5em}
\section{Results}
\label{sec:results}
\vspace{-0.5em}

In this section, we examine how cooperation changes as social strategy and replication evolve together. We organize the results around five research questions that build from a well-mixed baseline toward increasingly demanding environmental and computational conditions:

\begin{enumerate}
    \item RQ1: Is stealing suppressed when agents’ social strategy and replication mechanisms must evolve from scratch, assuming uniformly distributed energy?
    \item RQ2: Is this suppression driven by high mutation rates or inefficient stealing alone?
    \item RQ3: What is the role of the classical game-theoretic notion of ``assortment'' if cooperation is viable in the well-mixed setting with uniformly distributed energy?
    \item RQ4: Do these findings change if the passive background energy agents receive is now asymmetrically distributed?
    \item RQ5: What happens if agents must now earn additional energy by maximizing an externally imposed reward, and what if cooperation is now a multi-step computation?
\end{enumerate}

When our simulation begins with a computational soup of entirely random bytes, self-replicating programs naturally emerge over time, mirroring prior work in artificial life \cite{arcas2024computational}. What distinguishes our autopoietic environment from these standard models is that this open-ended evolution is natively game-theoretic: replication emerges while simultaneously suppressing defection.

\begin{figure}[b]
\vspace{-1em}
    \centering
    \includegraphics[width=0.9\linewidth]{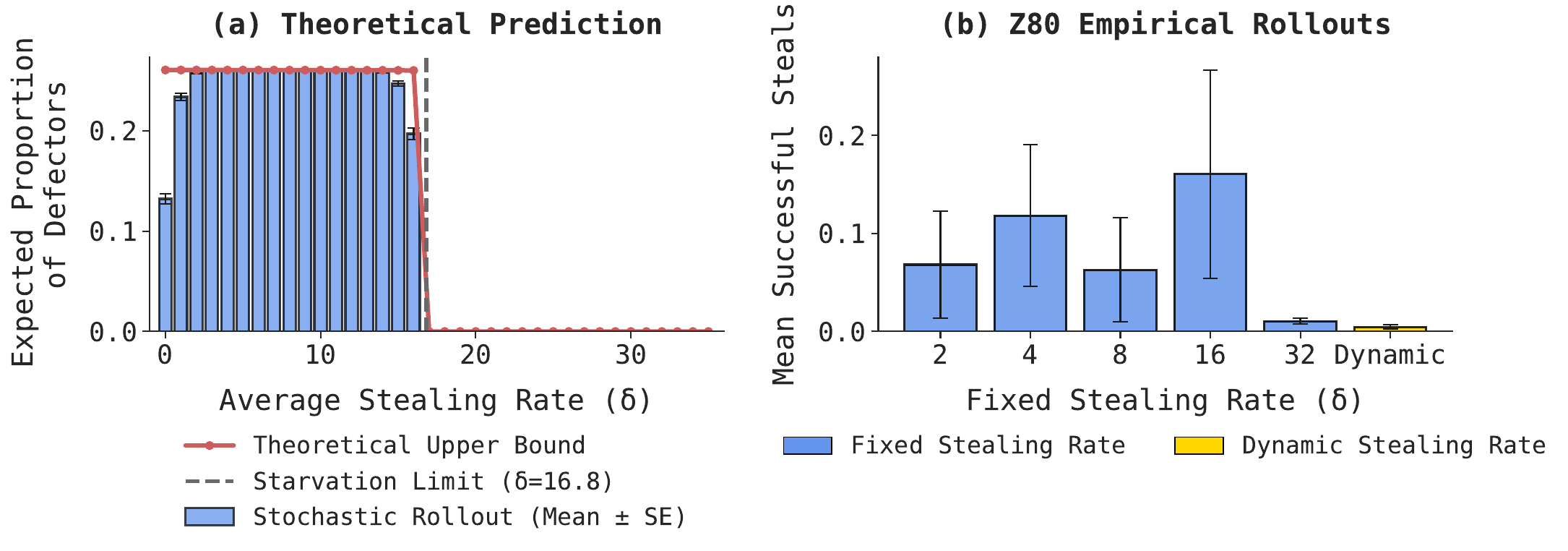}
    \caption{\textbf{The Metabolic Starvation Limit: Theory vs. Practice in Well-Mixed settings with non-zero mutation rate.} (a) Theoretical defector survival probability shows a steep drop-off at the starvation limit. (b) Empirical Z80 rollouts mirror this trajectory, remaining well below the theoretical bound. When agents can dynamically decide how much to steal by modifying custom registers, populations naturally converge to these safely suppressed levels.}
    \label{fig:rq2_starvation}
\end{figure}

\textbf{Cooperation emerges in well-mixed settings with continuous mutations (RQ1).} We evaluated a well-mixed population where every program receives a uniform baseline energy ($\epsilon = 24$) and is subject to a non-zero background mutation rate ($\mu = \frac{1}{128}$). We tested environments with varying fixed steal magnitudes ($\delta$) as well as settings where agents dynamically determine the amount to steal by modifying a dedicated register. Across all configurations, the \texttt{STEAL} instruction is heavily suppressed, natively converging to low volumes (Figure~\ref{fig:rq2_starvation}). As we show in Appendix~\ref{appendix:ProofWM}, the logic behind this suppression is fundamentally metabolic: because the \texttt{STEAL} operation absorbs energy imperfectly ($\alpha = 0.8$), mutual defection rapidly drains the shared energy pool. This energy loss severely slows down the execution speed of the paired system, exposing defectors to a much higher probability of fatal mutational errors. Conversely, cooperators successfully accumulate energy buffers across epochs. This accumulated wealth grants them a relative speed advantage in future interactions. Leveraging this speed, cooperators evolve replication logic that changes the equilibrium of the game; as demonstrated in an execution trace in Appendix~\ref{appendix:defectVCoopTrace}, a cooperative tape can utilize its speed advantage to stop a defector mid-execution and overwrite its malicious code entirely. In Appendix~\ref{appendix:SelfInterest}, we provide additional evidence suggesting these behaviors emerge from self-interested agents trying to maximize their energy rather than as a feature of our environment.
\begin{figure*}[t]
    \centering
    \begin{subfigure}{0.48\textwidth}
        \centering
        \includegraphics[width=\linewidth]{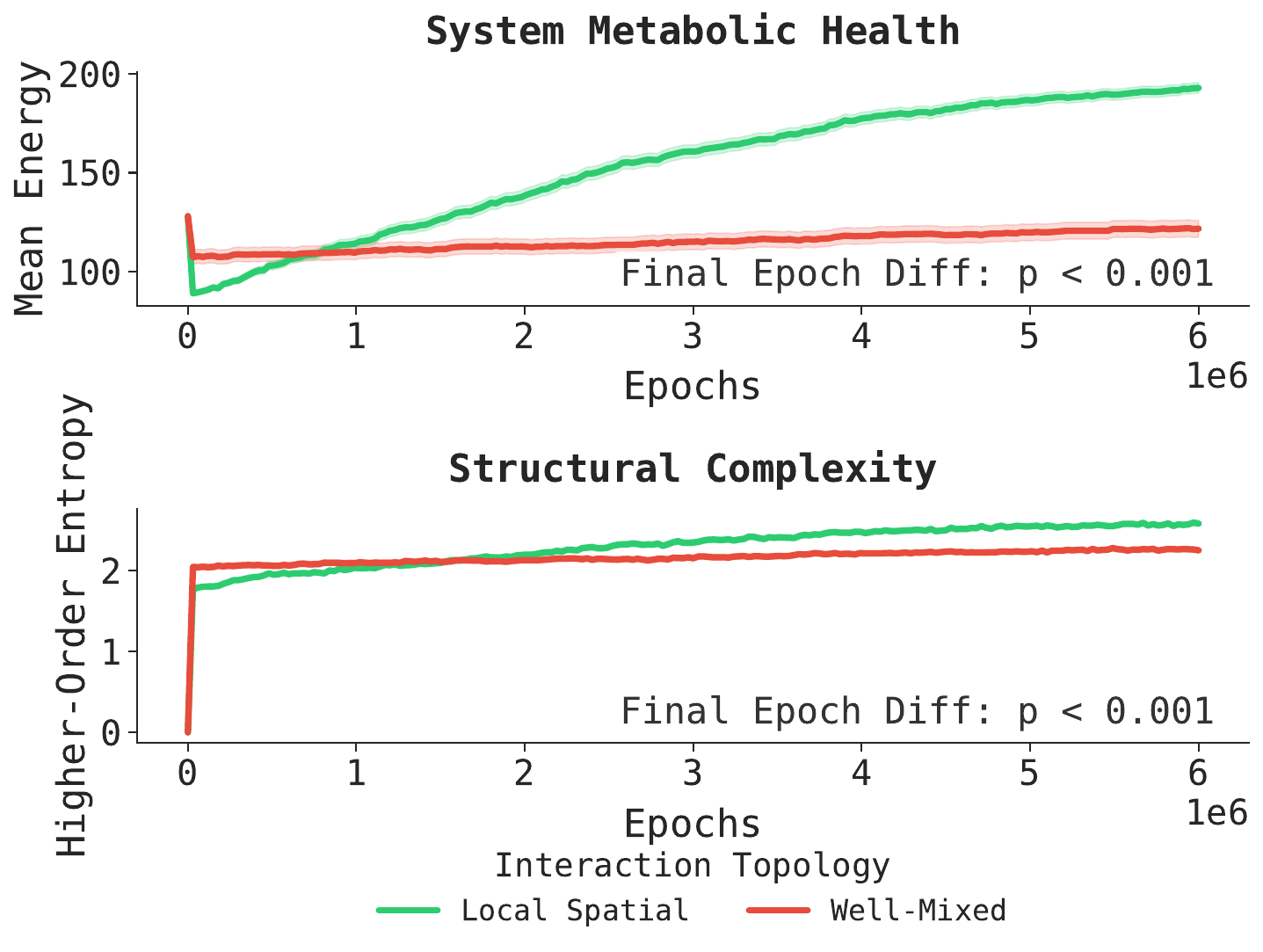}
        \caption{Overcoming Asymmetric Energy Grids. When subjected to severe, systematic environmental energy inequality, well-mixed populations suffer total metabolic and structural collapse. Conversely, local spatial topologies allow populations to maintain significantly higher terminal energy ($p<0.001$, Cohen's $d=8.98$) and structural complexity ($p<0.001$, Cohen's $d=3.54$; two-sided t-tests).}
        \label{fig:rq2_asymmetric_grid}
    \end{subfigure}
    \hfill
    \begin{subfigure}{0.48\textwidth}
        \centering
        \includegraphics[width=\linewidth]{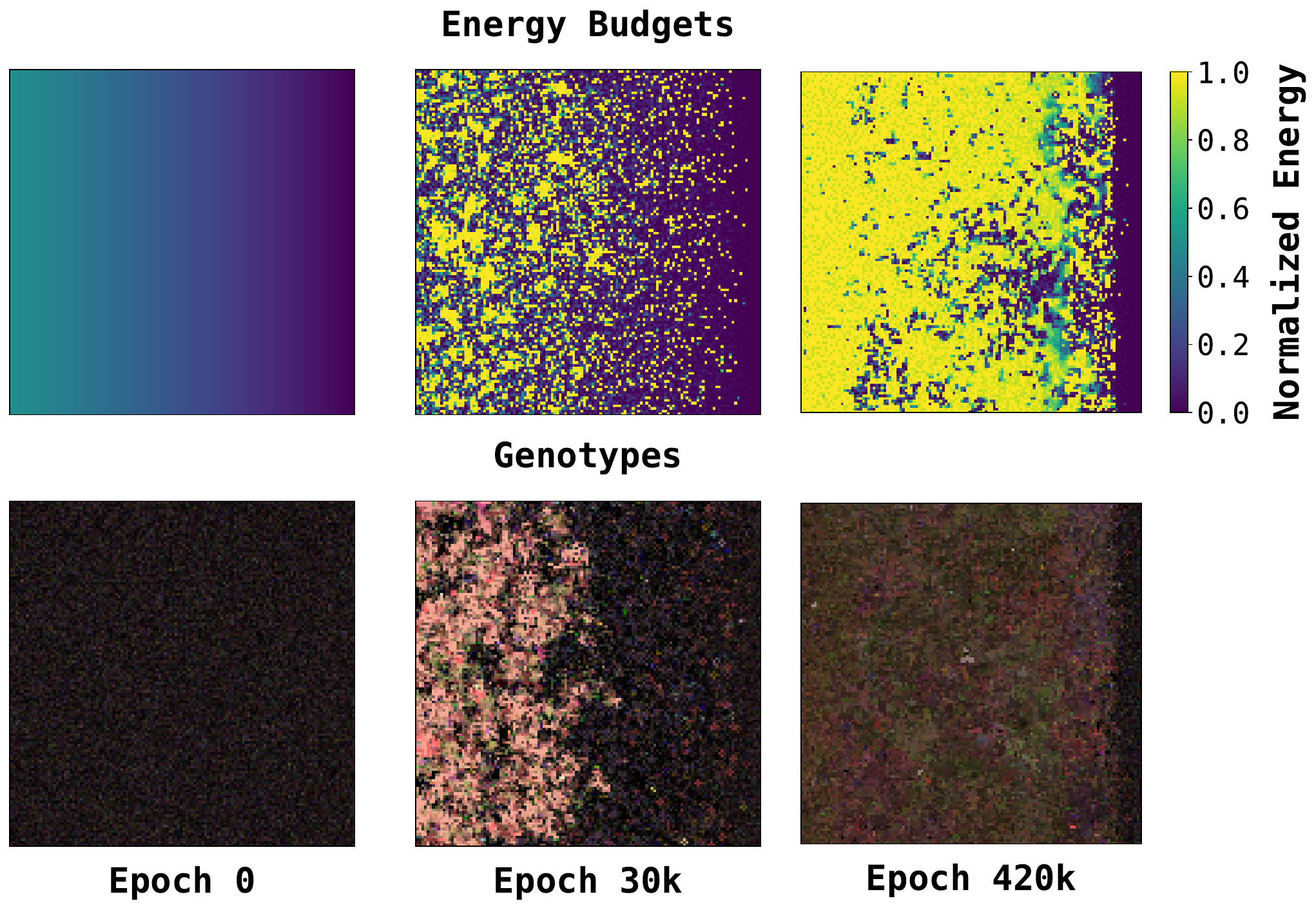}
           \caption{Visualizing evolution on an uneven energy grid. The top row shows the post-interaction energy budget of each tape, with brighter yellow indicating higher remaining energy. The bottom row visualizes program genotype: each possible operation is assigned a distinct RGB value, and each program is colored by averaging the RGB values of the operations in its code. Programs with similar colors therefore share similar operational composition. Over time, initially random, non-functioning tapes give way to efficient self-replicators that first emerge in high-energy regions and later spread into lower-energy areas.}
        \label{fig:keyframes_evolution}
    \end{subfigure}
    
    \caption{Impact of asymmetric environmental energy distribution on population health and complexity, visualized through terminal energy budgets and genotype-level program similarity.}
    \label{fig:combined_asymmetric_analysis}
    \vspace{-1em}
\end{figure*}

\textbf{The suppression of defection is not explained by high mutation rates alone (RQ2).} To test whether this suppression is driven by the coevolutionary process rather than just continuous mutational noise, we conducted a mutation-free ($\mu = 0$) petri-dish invasion assay. We seeded a naive random soup with 1\% functioning replicators and independently seeded 1\% of the soup with the \texttt{STEAL} opcode. Evaluating the endemic equilibrium reveals a clear starvation dropoff (Figure~\ref{fig:rq2_petridish_steals}): defectors are a minority of the population across the board, and the decrease in defector success when moving from a steal magnitude of $\delta=16$ to $\delta=32$ is highly statistically significant ($p<0.001$, Cohen's $d=4.17$; two-sided t-test). If a defector attempts to steal too much, it starves the shared energy pool before it can execute the $L$ writes necessary to replicate. Thus, even without mutations, high-stealing parasitic programs can trap themselves in evolutionary dead-ends.%

\begin{table}[htbp]
\vspace{-1em}
    \centering
    \begin{minipage}{0.45\textwidth}
        \centering
        \caption{Marginal defector proportions by $\alpha$ (averaged across $\delta$)}
        \label{tab:alpha_marginal}
        \begin{tabular}{ccc}
            \toprule
            $\alpha$ Value & Mean & SE \\
            \midrule
            0.25 & 0.210 & 0.038 \\
            0.50 & 0.169 & 0.034 \\
            0.75 & 0.213 & 0.035 \\
            0.95 & 0.204 & 0.036 \\
            \bottomrule
        \end{tabular}
    \end{minipage}\hfill
    \begin{minipage}{0.45\textwidth}
        \centering
        \caption{Marginal defector proportions by $\delta$ (averaged across $\alpha$)}
        \label{tab:delta_marginal}
        \begin{tabular}{ccc}
            \toprule
            $\delta$ Value & Mean & SE \\
            \midrule
            4 & 0.111 & 0.036 \\
            8 & 0.168 & 0.046 \\
            16 & 0.313 & 0.052 \\
            24 & 0.075 & 0.027 \\
            32 & 0.226 & 0.045 \\
            64 & 0.300 & 0.041 \\
            \bottomrule
        \end{tabular}
    \end{minipage}
\vspace{-0.5em}
\end{table}

\textbf{The suppression of defection is not explained by inefficient stealing alone (RQ2).}
To rule out whether the limits on the proportion of defectors in the population was driven by stealing being inefficient rather than agents being able to evolve their replication mechanisms simultaneously with their social strategy, we explore the relationship between the $\alpha$ and $\delta$ parameters. We simulated 10 seeds of a well-mixed population interacting with every combination of $\alpha \in \{0.25, 0.5, 0.75, 0.95\} \times \delta \in \{4, 8, 16, 24, 32, 64\}$. Moreover, to make stealing easier to evolve and more difficult to break via mutation, we re-encoded the STEAL operation into a single byte instruction.

In this experiment, if the limits on the ability for the defector population to grow were caused by inefficient stealing, then as $\alpha \to 1$ we should expect a greater proportion of defectors to survive since stealing would become more efficient. Instead, as Tables \ref{tab:alpha_marginal} and \ref{tab:delta_marginal} show, we find that the number of defectors does not appear to be sensitive to the value of $\alpha$ alone. However, we did find a weak relationship between the amount gained from stealing ($\alpha \times \delta$) and the proportion of defectors (Pearson $r = 0.14$, $p = 0.03$). This suggests other factors in the co-evolutionary process are likely doing the work of stabilizing cooperation.

\begin{figure}[b]
    \vspace{-1em}
    \centering
    \includegraphics[width=0.8\linewidth]{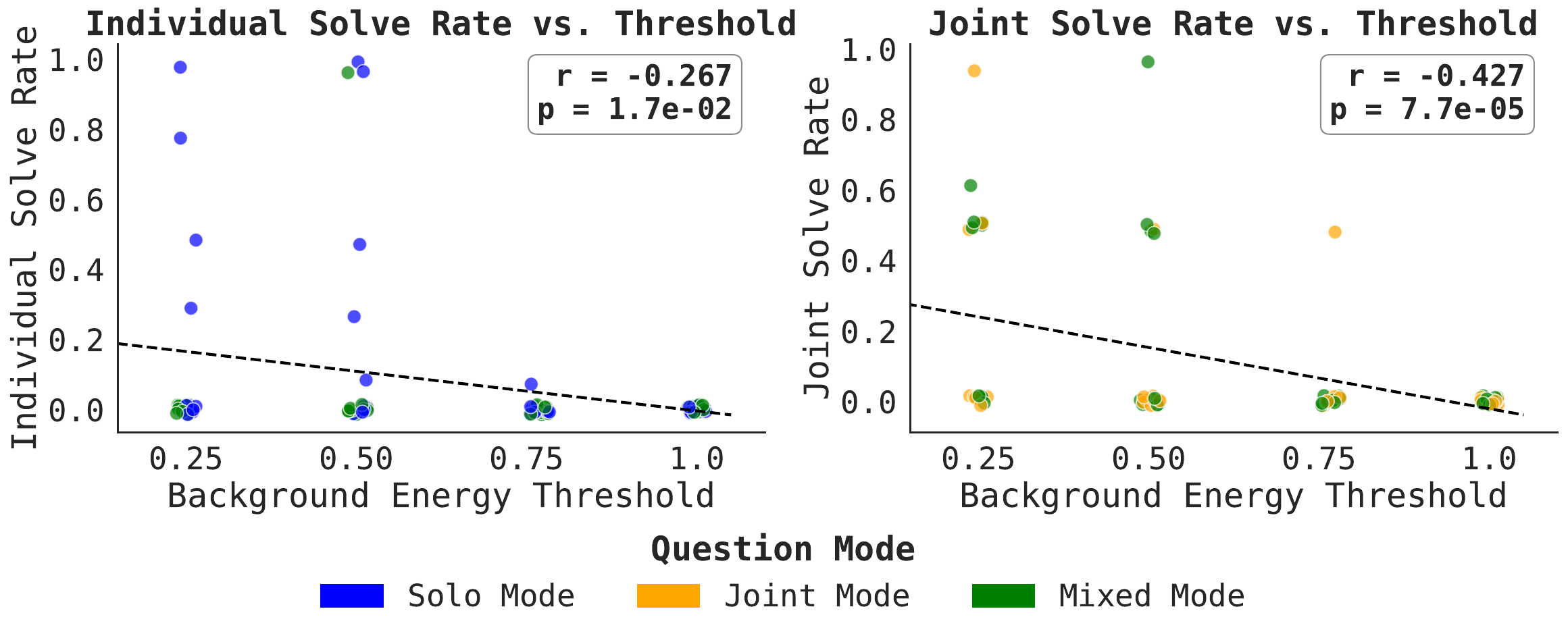}
    \caption{\textbf{Environmental scarcity drives problem-solving.} Solve rates for both individual (Left) and joint (Right) tasks show a significant negative correlation with environmental energy abundance. As free energy increases, the incentive to execute complex math tasks diminishes.}
    \label{fig:threshold_solveRate}
\end{figure}

\textbf{Local spatial interactions increase the structural complexity and efficiency of self-replicating programs (RQ3).} We next compared well-mixed populations to those restricted to a 2D local spatial grid, keeping uniform energy constant. While both topologies suppress defection under these conditions, local spatial interactions act as a powerful evolutionary catalyst. Spatial assortment significantly boosts higher-order entropy (Cohen's $d=2.81$) and lowers edit distance
(Cohen's $d=7.32$; both $p<0.001$, two-sided t-tests) compared to well-mixed pools (Figure~\ref{fig:rq2_catalysts}). By allowing non-stealing replicators to reliably cluster with identical kin, spatial assortment appears to support the evolution of more intricate, efficient instruction loops, rather than the more inefficient replicators typically found in well-mixed environments.

\textbf{Local interactions are needed for evolving complex replicators when background energy is asymmetrically distributed (RQ4).} Our initial experiments assumed every agent receives the exact same baseline energy. To test environmental resilience, we evaluated an asymmetric grid where the background energy $\epsilon$ varied significantly depending on an agent's physical location (Figure~\ref{fig:uniformVspace}). Under these systematic inequalities, well-mixed populations remain stuck as basic, low-complexity replicators with lower energy because they cannot rely on a predictable baseline influx. In stark contrast, populations restricted to local spatial interactions successfully overcome the environmental asymmetry to evolve complex, cooperative replicators, maintaining significantly higher terminal energy (Cohen's $d=8.98$) and higher-order entropy
(Cohen's $d=3.54$; both $p<0.001$, two-sided t-tests) (Figure~\ref{fig:rq2_asymmetric_grid}). Pockets of successful programs in high-energy zones iteratively build efficient code capable of surviving interactions in lower-energy zones (Figure~\ref{fig:keyframes_evolution}). Spatial assortment thus acts as an important architectural scaffold to overcome systemic biases.

\begin{figure}[t]
    \centering
    \includegraphics[width=1.0\linewidth]{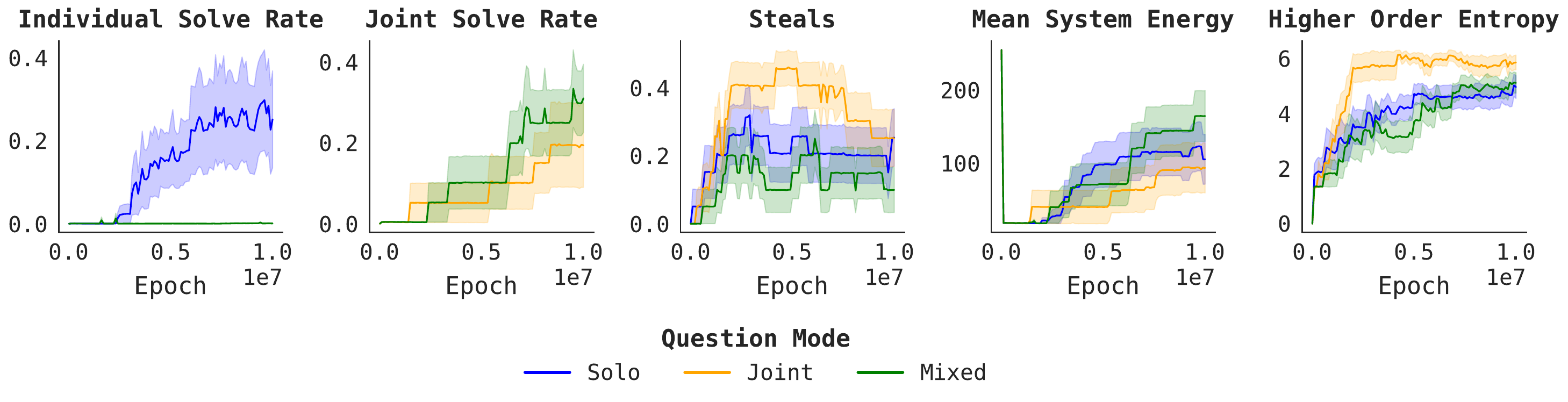}
    \caption{\textbf{Evolutionary convergence at severe metabolic scarcity} (25\% threshold). In Mixed mode, tapes heavily suppress the solo task to optimize for joint task performance. The structural complexification required for joint execution induces a temporary spike in stealing, which is ultimately suppressed, allowing system energy to climb and stabilize.}
    \label{fig:math_convergence}
    \vspace{-1em}
\end{figure}

\textbf{Scarcity drives agents to solve externally imposed tasks while simultaneously replicating their strategies (RQ5).} In our previous experiments, tapes received a constant amount of background energy from the environment, distributed either uniformly throughout the population or asymmetrical depending on the program's location on a 2D grid. To evaluate how the dynamics of agents' behavior changes when their energy influx is conditional on their ability to solve external tasks, we evaluate how well programs can solve math problems while simultaneously managing their energy and trying to replicate onto their partner. At the start of the simulation, each grid location is permanently assigned a mathematical constant $b$, which offspring inherit upon overwriting a neighbor. Every epoch, agents receive a random input $x \in [0, 255]$ and can access their partner's input $y$ (which costs more operations to read than $b$). To complete a task, an agent must compute the correct value and write it to a dedicated register \texttt{m} before execution ends. We defined two tasks: a Solo Task ($x + b$) and a Joint Task ($x + y$). To force engagement, free baseline energy is only provided if an agent's reserves fall below a strict threshold (e.g., 25\% of capacity). Above this limit, solving tasks is the only way to accumulate the surplus energy needed to safely replicate against wealthy opponents. We found a significant ($p < 0.001$) negative correlation between the energy threshold and task completion: as free energy becomes plentiful, the incentive to execute external tasks diminishes (Figure~\ref{fig:threshold_solveRate}).

\textbf{Cooperation persists in multi-step computational tasks and is strengthened by local interactions (RQ5).} We structured the math problems into a social dilemma evaluated by inspecting the custom register \texttt{m} for both tapes after every executed operation. Rewards are granted instantly. For instance, if Tape A solves the Joint task ($x+y$) first, both agents immediately receive $2\epsilon$. If Tape B then solves its Solo task, it receives an additional $5\epsilon$ (securing $7\epsilon$ for itself, while the joint solver keeps just $2\epsilon$). If Tape B instead solves the Joint task, both receive an additional $4\epsilon$ ($6\epsilon$ each). If both output their Solo answers, they each receive $5\epsilon$. Despite this sequential social dilemma, tapes suppress the safer solo task and optimize for the joint task to maximize collective energy (Figure~\ref{fig:math_convergence}), while \texttt{STEAL} remains suppressed. Finally, local spatial interactions produce significantly higher joint solve rates ($p\leq0.01$, Cohen's $d=1.63$; two-sided t-test) than well-mixed populations (Figure~\ref{fig:math_topology}), confirming that assortment aids collaborative problem-solving \cite{Dolson148973}.

\vspace{-0.5em}
\section{How Co-Evolving Replication Stabilizes Cooperation}
\label{sec:AGT}
\vspace{-0.5em}

The Z80 experiments above are intentionally open-ended: random byte strings discover their own replication loops, may or may not execute \texttt{STEAL}, and can alter both partners' memory during execution. This makes the system useful as a stress test for emergent cooperation, but it also makes it difficult to isolate the causal mechanism stabilizing cooperation. We therefore use a simplified model, inspired by \cite{Ohtsuki2006-fm}, to ask whether the ordering of interaction and replication can explain our results. Agents inhabit a finite, well-mixed population, always cooperate ($C$) or defect ($D$), and retain energy across encounters. Accumulated energy determines the probability of controlling a replication event, analogous to the Z80 substrate.

\begin{table}[t]

\centering
\small
\begin{minipage}{0.31\linewidth}
\centering
\textbf{Drain}\\
\vspace{0.25em}
\begin{tabular}{c|cc}
 & $C$ & $D$ \\
\hline
$C$ & $2$ & $-2$ \\
$D$ & $3.5$ & $-0.5$
\end{tabular}
\end{minipage}
\hfill
\begin{minipage}{0.31\linewidth}
\centering
\textbf{Stagnate}\\
\vspace{0.25em}
\begin{tabular}{c|cc}
 & $C$ & $D$ \\
\hline
$C$ & $2$ & $-2$ \\
$D$ & $4$ & $0$
\end{tabular}
\end{minipage}
\hfill
\begin{minipage}{0.31\linewidth}
\centering
\textbf{Accumulate}\\
\vspace{0.25em}
\begin{tabular}{c|cc}
 & $C$ & $D$ \\
\hline
$C$ & $2$ & $-2$ \\
$D$ & $5$ & $1$
\end{tabular}
\end{minipage}
\vspace{0.5em}
\caption{Prisoner's Dilemma payoff environments. Rows indicate the focal agent's strategy and columns indicate the partner's strategy; each entry is the focal agent's energy change. The fixed-timing experiments in Figure~\ref{fig:repMechToy} use the Drain environment. The co-evolution experiments in Figure~\ref{fig:coEvoToy} compare all three environments.}
\vspace{-2em}
\label{tab:toy-pd-payoffs}

\end{table}

The three Prisoner's Dilemma environments in Table~\ref{tab:toy-pd-payoffs} differ only in whether mutual defection drains, preserves, or increases energy. We first fix the replication rule and compare three event orderings in the Drain environment (Figure~\ref{fig:repMechToy}):

\begin{figure}[b]
\vspace{-0.5em}
    \centering
    \includegraphics[width=\linewidth]{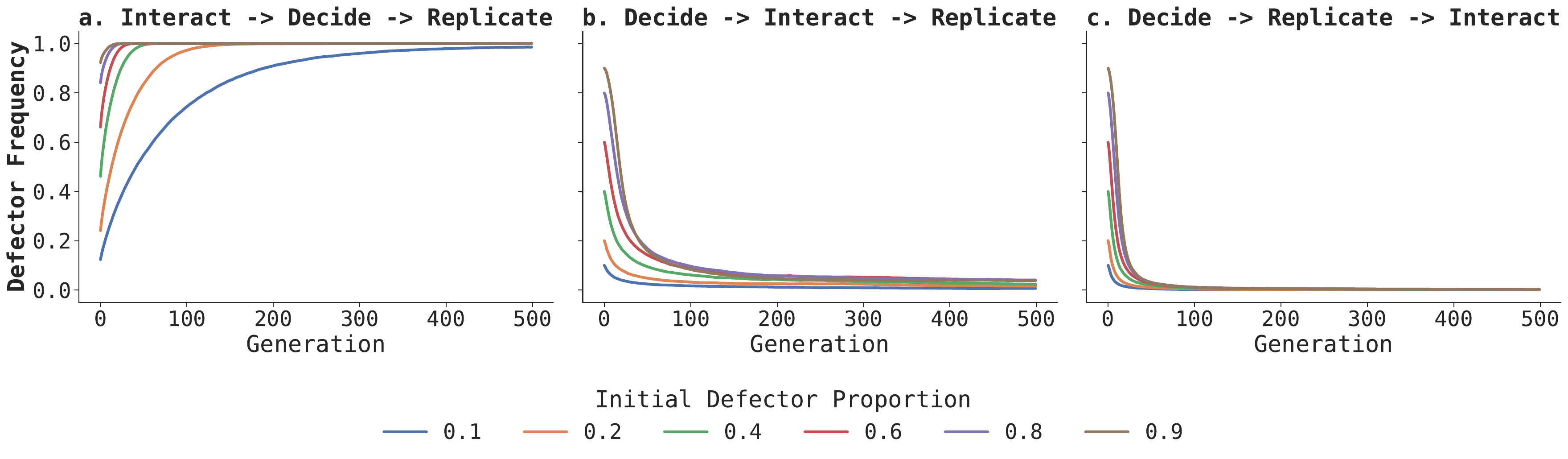}
    \caption{Population dynamics under three fixed replication-timing rules without mutations, with agents paired uniformly at random and mutual defection draining energy. (a) When replication is decided from post-interaction energy, defectors take over. (b) When replication is decided from pre-interaction energy but occurs after interaction, the immediate reproductive advantage of defection is removed and cooperation takes over. (c) When replication occurs before interaction, pairs become $(C,C)$ or $(D,D)$ before payoffs are assigned, making cooperation strictly better in this toy setting.}
    \label{fig:repMechToy}
    
\end{figure}

\begin{itemize}[leftmargin=*]
    \vspace{-0.5em}
    \item \textbf{Interact $\rightarrow$ decide from post-interaction energy $\rightarrow$ replicate:}
    Agents interact before replication priority is determined. Defectors can therefore convert the immediate energy gained from exploiting cooperators into control of the replication event, causing defection to take over (Figure~\ref{fig:repMechToy}a).
    \vspace{-0.25em}
    \item \textbf{Decide from pre-interaction energy $\rightarrow$ interact $\rightarrow$ replicate:}
    Replication priority is fixed before the interaction, although replication occurs afterward. Because the current payoff cannot affect the already-made decision, defectors lose their immediate reproductive advantage and cooperation takes over (Figure~\ref{fig:repMechToy}b).
    \vspace{-0.25em}
    \item \textbf{Decide from pre-interaction energy $\rightarrow$ replicate $\rightarrow$ interact:}
    Replication occurs before the social interaction, making each pair either $(C,C)$ or $(D,D)$. Mutual cooperation then yields more energy than mutual defection, so cooperation takes over (Figure~\ref{fig:repMechToy}c).
    \vspace{-0.5em}
\end{itemize}

\begin{figure}[t]

    \centering
    \includegraphics[width=\linewidth]{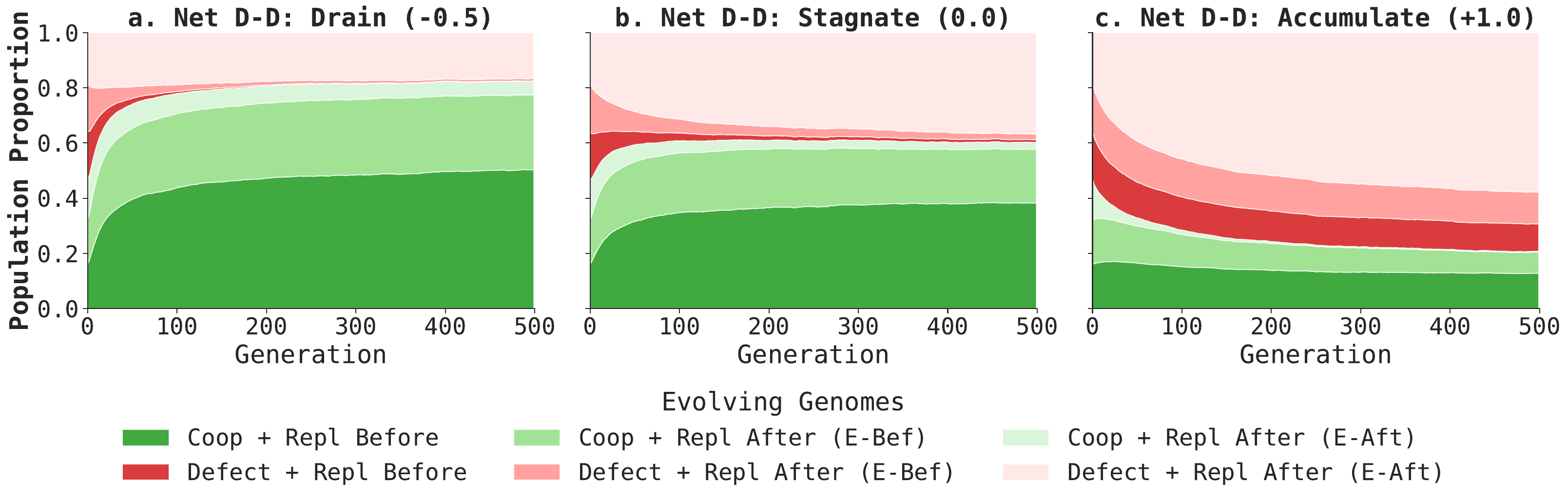}
    \caption{Co-evolution of social strategy and replication timing without mutation, with agents paired uniformly at random. The three environments vary only the payoff for defecting. When mutual defection drains or stagnates energy, cooperation robustly takes over; even when mutual defection increases energy, co-evolving replication timing allows cooperators to maintain a substantial presence rather than being eliminated.}
    \label{fig:coEvoToy}
        \vspace{-1em}
\end{figure}

These comparisons show that changing only the timing of replication can reverse the evolutionary outcome. We next allow both social strategy and replication timing to evolve (Figure~\ref{fig:coEvoToy}). The population begins with equal proportions of cooperators and defectors, with the three timing rules distributed uniformly and inherited by offspring. Accumulated energy determines which agent controls an encounter, and that agent's inherited rule determines the ordering of interaction and replication.

Cooperation dominates when mutual defection drains or stagnates total energy\footnote{Appendix~\ref{appendix:physics_justification} discusses the physical motivation for this assumption.}: repeated defection reduces the energy available to control future replication, whereas cooperation preserves it (Figure~\ref{fig:coEvoToy}a,b). Even when mutual defection increases energy, co-evolving replication timing allows cooperators to maintain a substantial population share (Figure~\ref{fig:coEvoToy}c).

Figure~\ref{fig:repMechToy}a shows a simpler model where mutual defection takes over under the ``interaction $\rightarrow$ decide $\rightarrow$ replicate'' propagation rule, even though this involves energy loss. Figures~\ref{fig:repMechToy}bc, and~\ref{fig:coEvoToy} show cooperation succeeding instead within the same environment. Together with the $\alpha$ and $\delta$ sweeps (Tables~\ref{tab:alpha_marginal} and~\ref{tab:delta_marginal}), this shows that energy loss alone is insufficient at sustaining cooperation; it is supported by coupling lossy social behavior to replication mechanisms that agents can evolve. In Appendix~\ref{appendix:ProofWM}, we detail our initial steps towards formalizing related sufficient conditions.

\vspace{-0.5em}
\section{Discussion}
\vspace{-0.5em}

In this work, we introduced Autopoietic Game Theory, a model that couples an agent's social strategy, replication mechanism, and available computational energy within a single phase of interaction. In our experiments, we found that defection can be suppressed through the combination of stealing being a lossy operation and the means of replication being capable of evolving. Our contribution is to show how these factors interact and evolve in a complex, open-ended space of programs, where even agents’ means of reproduction must emerge from scratch. Our results also reveal the limits of this mechanism. Cooperation remains robust under uniformly distributed energy, but under highly unequal energy, well-mixed populations collapse and local interactions are needed to maintain energy, complexity, and coordination on external tasks.

In Section 3 and 4, we extended existing substrates in ALife \citep{arcas2024computational} by making them game-theoretic through the introduction of energy-based computations and simultaneous action between tapes. In Section 5, we provided a simplified model of the system that tries to develop a causal understanding of the mechanisms stabilizing cooperation in our substrate. While our evidence suggests that the ability to evolve different timings of replication vs social interaction, paired with stealing being inefficient, is driving the emergence of cooperation, our substrate raises several exciting new questions that previous work has not been able to ask: What causes agents to copy all of their code onto their partner vs only part of their strategy? Does the crossover of programs resulting from two tapes partially copying onto each other provide any benefit to agents’ ability to adapt to a changing environment? Why was a cooperative strategy even discoverable from scratch in this substrate, and what environmental constraints lead to agents developing Moran-like ``interact$\rightarrow$decide$\rightarrow$replicate'' reproduction processes instead? How does the higher-order complexity of the resulting tapes shape their ability to solve downstream tasks? Nothing in our substrate is designed by hand, and agents are represented as programs with an open-ended space of actions. This creates an entire ecology of behaviors that we have only begun to map. We hope the model introduced in this work can help us build a more nuanced picture of when and why cooperation evolves.

\textbf{Limitations:}
Our simulations are limited to a simplified Prisoner's Dilemma and a Z80 assembly substrate. Different instruction densities and hardware architectures may change replication costs and starvation thresholds, as well as the reachability of cooperative strategies. We also assume unrestricted read and write access to partner memory. Although this resembles biological horizontal gene transfer or AI agents sharing prompts and problem-solving harnesses, many systems restrict access to an agent's internal state. Testing the range between shared and private computation will help determine when these dynamics persist. Moreover, our analysis in Appendix~\ref{appendix:ProofWM} also makes two simplifying assumptions: Conjectures~\ref{conj:asymmetric-recovery} and~\ref{conj:frequency-rebound} fix agents' relative execution speeds at their initial values, while Lemmas~\ref{lem:metabolic-drag} and~\ref{lem:metabolic-starvation}, which support Theorem~\ref{thm:local-favorability}, replace the system's changing total energy with its average. Section~\ref{sec:results} tests the full dynamics, updating both after every operation. Extending the analysis to asymmetric energy, continuous mutation, and external tasks remains an open challenge.

\newpage
\bibliographystyle{abbrvnat}
\bibliography{bibliography}

\begin{thebibliography}{66}
\providecommand{\natexlab}[1]{#1}
\providecommand{\url}[1]{\texttt{#1}}
\expandafter\ifx\csname urlstyle\endcsname\relax
  \providecommand{\doi}[1]{doi: #1}\else
  \providecommand{\doi}{doi: \begingroup \urlstyle{rm}\Url}\fi

\bibitem[Adami and Brown(1994)]{adami1994evolutionary}
C.~Adami and C.~T. Brown.
\newblock Evolutionary learning in the 2d artificial life system "avida", 1994.
\newblock URL \url{https://arxiv.org/abs/adap-org/9405003}.

\bibitem[Agüera~y Arcas(2025)]{y2025intelligence}
B.~Agüera~y Arcas.
\newblock \emph{What is Intelligence?: Lessons from AI about Evolution, Computing, and Minds}.
\newblock MIT Press, 2025.

\bibitem[Agüera~y Arcas et~al.(2024)Agüera~y Arcas, Alakuijala, Evans, Laurie, Mordvintsev, Niklasson, Randazzo, and Versari]{arcas2024computational}
B.~Agüera~y Arcas, J.~Alakuijala, J.~Evans, B.~Laurie, A.~Mordvintsev, E.~Niklasson, E.~Randazzo, and L.~Versari.
\newblock Computational life: How well-formed, self-replicating programs emerge from simple interaction, 2024.
\newblock URL \url{https://arxiv.org/abs/2406.19108}.

\bibitem[Aiello and Wheeler(1995)]{aiello1995expensive}
L.~C. Aiello and P.~Wheeler.
\newblock The expensive-tissue hypothesis: the brain and the digestive system in human and primate evolution.
\newblock \emph{Current anthropology}, 36\penalty0 (2):\penalty0 199--221, 1995.

\bibitem[Axelrod and Hamilton(1981)]{axelrod1981evolution}
R.~Axelrod and W.~D. Hamilton.
\newblock The evolution of cooperation.
\newblock \emph{science}, 211\penalty0 (4489):\penalty0 1390--1396, 1981.

\bibitem[Baker(2020)]{baker2020emergentreciprocityteamformation}
B.~Baker.
\newblock Emergent reciprocity and team formation from randomized uncertain social preferences, 2020.
\newblock URL \url{https://arxiv.org/abs/2011.05373}.

\bibitem[Baker et~al.(2017)Baker, Jara-Ettinger, Saxe, and Tenenbaum]{baker2017rational}
C.~L. Baker, J.~Jara-Ettinger, R.~Saxe, and J.~B. Tenenbaum.
\newblock Rational quantitative attribution of beliefs, desires and percepts in human mentalizing.
\newblock \emph{Nature Human Behaviour}, 1\penalty0 (4):\penalty0 0064, 2017.

\bibitem[Bear and Rand(2016)]{bear2016intuition}
A.~Bear and D.~G. Rand.
\newblock Intuition, deliberation, and the evolution of cooperation.
\newblock \emph{Proceedings of the National Academy of Sciences}, 113\penalty0 (4):\penalty0 936--941, 2016.

\bibitem[Bennett(1982)]{Bennett1982-hd}
C.~H. Bennett.
\newblock The thermodynamics of computation---a review.
\newblock \emph{Int. J. Theor. Phys.}, 21\penalty0 (12):\penalty0 905--940, Dec. 1982.

\bibitem[Cardillo et~al.(2010)Cardillo, Gómez-Gardeñes, Vilone, and Sánchez]{Cardillo_2010}
A.~Cardillo, J.~Gómez-Gardeñes, D.~Vilone, and A.~Sánchez.
\newblock Co-evolution of strategies and update rules in the prisoner’s dilemma game on complex networks.
\newblock \emph{New Journal of Physics}, 12\penalty0 (10):\penalty0 103034, Oct. 2010.
\newblock ISSN 1367-2630.
\newblock \doi{10.1088/1367-2630/12/10/103034}.
\newblock URL \url{http://dx.doi.org/10.1088/1367-2630/12/10/103034}.

\bibitem[Chandrakasan and Brodersen(1995)]{Chandrakasan1995-ee}
A.~P. Chandrakasan and R.~W. Brodersen.
\newblock \emph{Low Power Digital {CMOS} Design}.
\newblock Springer, Dordrecht, Netherlands, 1995 edition, June 1995.

\bibitem[Chen et~al.(2000)Chen, Kwong, and Li]{chenCompression}
X.~Chen, S.~Kwong, and M.~Li.
\newblock A compression algorithm for dna sequences and its applications in genome comparison.
\newblock In \emph{Proceedings of the Fourth Annual International Conference on Computational Molecular Biology}, RECOMB '00, page 107, New York, NY, USA, 2000. Association for Computing Machinery.
\newblock ISBN 1581131860.
\newblock \doi{10.1145/332306.332352}.
\newblock URL \url{https://doi.org/10.1145/332306.332352}.

\bibitem[Critch(2019)]{critch2019parametric}
A.~Critch.
\newblock A parametric, resource-bounded generalization of l{\"o}b’s theorem, and a robust cooperation criterion for open-source game theory.
\newblock \emph{The Journal of Symbolic Logic}, 84\penalty0 (4):\penalty0 1368--1381, 2019.

\bibitem[Dafoe et~al.(2020)Dafoe, Hughes, Bachrach, Collins, McKee, Leibo, Larson, and Graepel]{dafoe2020openproblemscooperativeai}
A.~Dafoe, E.~Hughes, Y.~Bachrach, T.~Collins, K.~R. McKee, J.~Z. Leibo, K.~Larson, and T.~Graepel.
\newblock Open problems in cooperative ai, 2020.
\newblock URL \url{https://arxiv.org/abs/2012.08630}.

\bibitem[{Department of EECS (Computer Science Division) University of Central Florida, Orlando, FL 32816} et~al.(2014){Department of EECS (Computer Science Division) University of Central Florida, Orlando, FL 32816}, Soros, and Stanley]{sorosIdentifying}
{Department of EECS (Computer Science Division) University of Central Florida, Orlando, FL 32816}, L.~Soros, and K.~Stanley.
\newblock Identifying necessary conditions for open-ended evolution through the artificial life world of chromaria.
\newblock In \emph{Artificial Life 14: Proceedings of the Fourteenth International Conference on the Synthesis and Simulation of Living Systems}, pages 793--800. The MIT Press, July 2014.

\bibitem[Dewdney(1985)]{coreWars}
A.~K. Dewdney.
\newblock Computer recreations.
\newblock \emph{Scientific American}, 252\penalty0 (6):\penalty0 18--29, 1985.
\newblock ISSN 00368733, 19467087.
\newblock URL \url{http://www.jstor.org/stable/24967665}.

\bibitem[Dochkina(2026)]{dochkina2026drophierarchyrolesselforganizing}
V.~Dochkina.
\newblock Drop the hierarchy and roles: How self-organizing llm agents outperform designed structures, 2026.
\newblock URL \url{https://arxiv.org/abs/2603.28990}.

\bibitem[Dolson et~al.(2017)Dolson, P{\'e}rez, Olson, and Ofria]{Dolson148973}
E.~L. Dolson, S.~G. P{\'e}rez, R.~S. Olson, and C.~Ofria.
\newblock Spatial resource heterogeneity increases diversity and evolutionary potential.
\newblock \emph{bioRxiv}, 2017.
\newblock \doi{10.1101/148973}.
\newblock URL \url{https://www.biorxiv.org/content/early/2017/06/12/148973}.

\bibitem[Dunbar(2003)]{dunbar2003social}
R.~I. Dunbar.
\newblock The social brain: mind, language, and society in evolutionary perspective.
\newblock \emph{Annual review of Anthropology}, 32\penalty0 (1):\penalty0 163--181, 2003.

\bibitem[Dunbar and Shultz(2007)]{dunbar2007evolution}
R.~I. Dunbar and S.~Shultz.
\newblock Evolution in the social brain.
\newblock \emph{science}, 317\penalty0 (5843):\penalty0 1344--1347, 2007.

\bibitem[Gawne et~al.(2020)Gawne, McKenna, and Levin]{Gawne2020-lh}
R.~Gawne, K.~Z. McKenna, and M.~Levin.
\newblock Competitive and coordinative interactions between body parts produce adaptive developmental outcomes.
\newblock \emph{Bioessays}, 42\penalty0 (8):\penalty0 e1900245, Aug. 2020.

\bibitem[Henrich(2015)]{henrich2015secret}
J.~Henrich.
\newblock The secret of our success: How culture is driving human evolution, domesticating our species, and making us smarter.
\newblock In \emph{The secret of our success}. princeton University press, 2015.

\bibitem[Henrich and Henrich(2007)]{henrich2007humans}
J.~Henrich and N.~Henrich.
\newblock \emph{Why Humans Cooperate: A Cultural and Evolutionary Explanation}.
\newblock Evolution and Cognition. Oxford University Press, 2007.
\newblock ISBN 9780198041177.
\newblock URL \url{https://books.google.lu/books?id=qoVxwP0aX4cC}.

\bibitem[Horibe(2003)]{Horibe2003-ro}
Y.~Horibe.
\newblock A note on kolmogorov complexity and entropy.
\newblock \emph{Appl. Math. Lett.}, 16\penalty0 (7):\penalty0 1129--1130, Oct. 2003.

\bibitem[Hu et~al.(2025)Hu, Lu, and Clune]{hu2025automateddesignagenticsystems}
S.~Hu, C.~Lu, and J.~Clune.
\newblock Automated design of agentic systems, 2025.
\newblock URL \url{https://arxiv.org/abs/2408.08435}.

\bibitem[Hull(2004)]{hull2004niche}
D.~L. Hull.
\newblock Niche construction: The neglected process in evolution.
\newblock \emph{Perspectives in Biology and Medicine}, 47\penalty0 (2):\penalty0 314--316, 2004.

\bibitem[Jaques et~al.(2019)Jaques, Lazaridou, Hughes, Gulcehre, Ortega, Strouse, Leibo, and de~Freitas]{jaques2019socialinfluenceintrinsicmotivation}
N.~Jaques, A.~Lazaridou, E.~Hughes, C.~Gulcehre, P.~A. Ortega, D.~Strouse, J.~Z. Leibo, and N.~de~Freitas.
\newblock Social influence as intrinsic motivation for multi-agent deep reinforcement learning, 2019.
\newblock URL \url{https://arxiv.org/abs/1810.08647}.

\bibitem[Jha et~al.(2024)Jha, Le, Jin, Kuo, Tenenbaum, and Shu]{jha2024neural}
K.~Jha, T.~A. Le, C.~Jin, Y.-L. Kuo, J.~B. Tenenbaum, and T.~Shu.
\newblock Neural amortized inference for nested multi-agent reasoning.
\newblock In \emph{Proceedings of the AAAI Conference on Artificial Intelligence}, volume~38, pages 530--537, 2024.

\bibitem[Kaplan et~al.(2000)Kaplan, Hill, Lancaster, and Hurtado]{kaplan2000theory}
H.~Kaplan, K.~Hill, J.~Lancaster, and A.~M. Hurtado.
\newblock A theory of human life history evolution: Diet, intelligence, and longevity.
\newblock \emph{Evolutionary Anthropology: Issues, News, and Reviews: Issues, News, and Reviews}, 9\penalty0 (4):\penalty0 156--185, 2000.

\bibitem[Kleiman-Weiner et~al.(2016)Kleiman-Weiner, Ho, Austerweil, Littman, and Tenenbaum]{kleiman2016coordinate}
M.~Kleiman-Weiner, M.~K. Ho, J.~L. Austerweil, M.~L. Littman, and J.~B. Tenenbaum.
\newblock Coordinate to cooperate or compete: abstract goals and joint intentions in social interaction.
\newblock In \emph{Proceedings of the annual meeting of the cognitive science society}, volume~38, 2016.

\bibitem[Kleiman-Weiner et~al.(2020)Kleiman-Weiner, Sosa, Thompson, Opheusden, Griffiths, Gershman, and Cushman]{kleiman2020downloading}
M.~Kleiman-Weiner, F.~Sosa, B.~Thompson, B.~v. Opheusden, T.~L. Griffiths, S.~Gershman, and F.~Cushman.
\newblock Downloading culture. zip: Social learning by program induction.
\newblock In \emph{Proceedings of the Annual Meeting of the Cognitive Science Society}, volume~42, 2020.

\bibitem[Kleiman-Weiner et~al.(2025)Kleiman-Weiner, Vient{\'o}s, Rand, and Tenenbaum]{Kleiman-Weiner2025-jr}
M.~Kleiman-Weiner, A.~Vient{\'o}s, D.~G. Rand, and J.~B. Tenenbaum.
\newblock Evolving general cooperation with a bayesian theory of mind.
\newblock \emph{Proc. Natl. Acad. Sci. U. S. A.}, 122\penalty0 (25):\penalty0 e2400993122, June 2025.

\bibitem[Laland(2017)]{laland2017darwin}
K.~N. Laland.
\newblock \emph{Darwin's unfinished symphony: How culture made the human mind}.
\newblock Princeton University Press, 2017.

\bibitem[Landauer(1961)]{Landauer1961-sj}
R.~Landauer.
\newblock Irreversibility and heat generation in the computing process.
\newblock \emph{IBM J. Res. Dev.}, 5\penalty0 (3):\penalty0 183--191, July 1961.

\bibitem[Leibo et~al.(2017)Leibo, Zambaldi, Lanctot, Marecki, and Graepel]{leibo2017multiagentreinforcementlearningsequential}
J.~Z. Leibo, V.~Zambaldi, M.~Lanctot, J.~Marecki, and T.~Graepel.
\newblock Multi-agent reinforcement learning in sequential social dilemmas, 2017.
\newblock URL \url{https://arxiv.org/abs/1702.03037}.

\bibitem[Leibo et~al.(2025)Leibo, Vezhnevets, Cunningham, and Bileschi]{leibo2025pragmaticviewaipersonhood}
J.~Z. Leibo, A.~S. Vezhnevets, W.~A. Cunningham, and S.~M. Bileschi.
\newblock A pragmatic view of ai personhood, 2025.
\newblock URL \url{https://arxiv.org/abs/2510.26396}.

\bibitem[Lenski et~al.(2003)Lenski, Ofria, Pennock, and Adami]{lenski2003evolutionary}
R.~E. Lenski, C.~Ofria, R.~T. Pennock, and C.~Adami.
\newblock The evolutionary origin of complex features.
\newblock \emph{Nature}, 423\penalty0 (6936):\penalty0 139--144, 2003.

\bibitem[Liu et~al.(2026)Liu, Agarwal, Maheswaran, Cemri, Li, Mang, Naren, Boneh, Cheng, Pan, Du, Keutzer, Cheung, Dimakis, Sen, Zaharia, and Stoica]{liu2026evoxmetaevolutionautomateddiscovery}
S.~Liu, S.~Agarwal, M.~Maheswaran, M.~Cemri, Z.~Li, Q.~Mang, A.~Naren, E.~Boneh, A.~Cheng, M.~Z. Pan, A.~Du, K.~Keutzer, A.~Cheung, A.~G. Dimakis, K.~Sen, M.~Zaharia, and I.~Stoica.
\newblock Evox: Meta-evolution for automated discovery, 2026.
\newblock URL \url{https://arxiv.org/abs/2602.23413}.

\bibitem[Margolus and Levitin(1998)]{Margolus_1998}
N.~Margolus and L.~B. Levitin.
\newblock The maximum speed of dynamical evolution.
\newblock \emph{Physica D: Nonlinear Phenomena}, 120\penalty0 (1–2):\penalty0 188–195, Sept. 1998.
\newblock ISSN 0167-2789.
\newblock \doi{10.1016/s0167-2789(98)00054-2}.
\newblock URL \url{http://dx.doi.org/10.1016/S0167-2789(98)00054-2}.

\bibitem[Meulemans et~al.(2025)Meulemans, Nasser, Wołczyk, Weis, Kobayashi, Richards, Lajoie, Steger, Hutter, Manyika, Saurous, Sacramento, and Agüera~y Arcas]{meulemans2025embeddeduniversalpredictiveintelligence}
A.~Meulemans, R.~Nasser, M.~Wołczyk, M.~A. Weis, S.~Kobayashi, B.~Richards, G.~Lajoie, A.~Steger, M.~Hutter, J.~Manyika, R.~A. Saurous, J.~Sacramento, and B.~Agüera~y Arcas.
\newblock Embedded universal predictive intelligence: a coherent framework for multi-agent learning, 2025.
\newblock URL \url{https://arxiv.org/abs/2511.22226}.

\bibitem[Muthukrishna(2023)]{muthukrishna2023theory}
M.~Muthukrishna.
\newblock \emph{A theory of everyone: The new science of who we are, how we got here, and where we’re going}.
\newblock MIT Press, 2023.

\bibitem[Muthukrishna et~al.(2018)Muthukrishna, Doebeli, Chudek, and Henrich]{muthukrishna2018cultural}
M.~Muthukrishna, M.~Doebeli, M.~Chudek, and J.~Henrich.
\newblock The cultural brain hypothesis: How culture drives brain expansion, sociality, and life history.
\newblock \emph{PLoS computational biology}, 14\penalty0 (11):\penalty0 e1006504, 2018.

\bibitem[Nowak(2006)]{nowak2006five}
M.~A. Nowak.
\newblock Five rules for the evolution of cooperation.
\newblock \emph{science}, 314\penalty0 (5805):\penalty0 1560--1563, 2006.

\bibitem[Nowak and May(1992)]{Nowak1992-qb}
M.~A. Nowak and R.~M. May.
\newblock Evolutionary games and spatial chaos.
\newblock \emph{Nature}, 359\penalty0 (6398):\penalty0 826--829, Oct. 1992.

\bibitem[Ohtsuki and Nowak(2006)]{Ohtsuki2006-fm}
H.~Ohtsuki and M.~A. Nowak.
\newblock Evolutionary games on cycles.
\newblock \emph{Proc. Biol. Sci.}, 273\penalty0 (1598):\penalty0 2249--2256, Sept. 2006.

\bibitem[Perc and Szolnoki(2010)]{Perc_2010}
M.~Perc and A.~Szolnoki.
\newblock Coevolutionary games—a mini review.
\newblock \emph{Biosystems}, 99\penalty0 (2):\penalty0 109–125, Feb. 2010.
\newblock ISSN 0303-2647.
\newblock \doi{10.1016/j.biosystems.2009.10.003}.
\newblock URL \url{http://dx.doi.org/10.1016/j.biosystems.2009.10.003}.

\bibitem[Rand and Nowak(2013)]{rand2013human}
D.~G. Rand and M.~A. Nowak.
\newblock Human cooperation.
\newblock \emph{Trends in cognitive sciences}, 17\penalty0 (8):\penalty0 413--425, 2013.

\bibitem[Ray(1992)]{ray1992evolution}
T.~S. Ray.
\newblock Evolution, ecology and optimization of digital organisms.
\newblock \emph{Santa Fe}, 1992.

\bibitem[Roca et~al.(2009)Roca, Cuesta, and Sánchez]{Roca_2009}
C.~P. Roca, J.~A. Cuesta, and A.~Sánchez.
\newblock Evolutionary game theory: Temporal and spatial effects beyond replicator dynamics.
\newblock \emph{Physics of Life Reviews}, 6\penalty0 (4):\penalty0 208–249, Dec. 2009.
\newblock ISSN 1571-0645.
\newblock \doi{10.1016/j.plrev.2009.08.001}.
\newblock URL \url{http://dx.doi.org/10.1016/j.plrev.2009.08.001}.

\bibitem[Sandholm(2010)]{Sandholm2010-je}
W.~H. Sandholm.
\newblock \emph{Population games and evolutionary dynamics}.
\newblock Economic Learning and Social Evolution. MIT Press, London, England, Dec. 2010.

\bibitem[Shreesha et~al.(2025)Shreesha, Pigozzi, Goldstein, and Levin]{shreesha}
L.~Shreesha, F.~Pigozzi, A.~Goldstein, and M.~Levin.
\newblock Extending iterated, spatialized prisoner’s dilemma to understand multicellularity: Game theory with self-scaling players.
\newblock \emph{IEEE Transactions on Molecular, Biological, and Multi-Scale Communications}, 11\penalty0 (2):\penalty0 135--151, 2025.
\newblock \doi{10.1109/TMBMC.2025.3562358}.

\bibitem[Sistla and Kleiman-Weiner(2025)]{sistla2025evaluating}
S.~Sistla and M.~Kleiman-Weiner.
\newblock Evaluating llms in open-source games.
\newblock \emph{arXiv preprint arXiv:2512.00371}, 2025.

\bibitem[Smiley and Levin(2022)]{Smiley2022-rs}
P.~Smiley and M.~Levin.
\newblock Competition for finite resources as coordination mechanism for morphogenesis: An evolutionary algorithm study of digital embryogeny.
\newblock \emph{Biosystems.}, 221\penalty0 (104762):\penalty0 104762, Nov. 2022.

\bibitem[Smith(1982)]{smith1982evolution}
J.~Smith.
\newblock \emph{Evolution and the Theory of Games}.
\newblock Cambridge University Press, 1982.
\newblock ISBN 9780521288842.
\newblock URL \url{https://books.google.ch/books?id=Nag2IhmPS3gC}.

\bibitem[Steinberger et~al.(2026)Steinberger, Koc, Zaidi, Hoffman, Santana, {Vignesh}, {Shakker}, {Shadow}, Avant, Slight, Nakazawa, {Mariano}, Yust, Gutman, Tomlinson, {github-actions[bot]}, {scoootscooob}, Alexander, {Sid}, Lehman, {Onur}, {the sun gif man}, Hunt, Mendonca, Castro, Rivera, Solmaz, {Glucksberg}, {Altay}, and {max}]{Steinberger2026openclaw}
P.~Steinberger, V.~Koc, A.~Zaidi, T.~Hoffman, G.~M. Santana, {Vignesh}, {Shakker}, {Shadow}, J.~Avant, S.~Slight, C.~Nakazawa, {Mariano}, T.~Yust, N.~Gutman, J.~Tomlinson, {github-actions[bot]}, {scoootscooob}, V.~Alexander, {Sid}, J.~Lehman, {Onur}, {the sun gif man}, H.~Hunt, B.~Mendonca, M.~Castro, A.~Rivera, O.~Solmaz, {Glucksberg}, {Altay}, and {max}.
\newblock openclaw/openclaw.
\newblock https://github.com/openclaw/openclaw, apr 17 2026.
\newblock URL \url{https://github.com/openclaw/openclaw}.

\bibitem[Tennenholtz(2004)]{tennenholtz2004program}
M.~Tennenholtz.
\newblock Program equilibrium.
\newblock \emph{Games and Economic Behavior}, 49\penalty0 (2):\penalty0 363--373, 2004.

\bibitem[{Thomas M. Cover} and Thomas(2006)]{Thomas_M_Cover2006-vo}
{Thomas M. Cover} and J.~A. Thomas.
\newblock \emph{Elements of Information Theory}.
\newblock John Wiley \& Sons, Nashville, TN, 2 edition, June 2006.

\bibitem[Tomasello and Carpenter(2007)]{tomasello2007shared}
M.~Tomasello and M.~Carpenter.
\newblock Shared intentionality.
\newblock \emph{Developmental science}, 10\penalty0 (1):\penalty0 121--125, 2007.

\bibitem[Ullman et~al.(2009)Ullman, Baker, Macindoe, Evans, Goodman, and Tenenbaum]{ullman2009}
T.~Ullman, C.~Baker, O.~Macindoe, O.~Evans, N.~Goodman, and J.~Tenenbaum.
\newblock Help or hinder: Bayesian models of social goal inference.
\newblock In Y.~Bengio, D.~Schuurmans, J.~Lafferty, C.~Williams, and A.~Culotta, editors, \emph{Advances in Neural Information Processing Systems}, volume~22. Curran Associates, Inc., 2009.
\newblock URL \url{https://proceedings.neurips.cc/paper_files/paper/2009/file/52292e0c763fd027c6eba6b8f494d2eb-Paper.pdf}.

\bibitem[Von~Neumann et~al.(1966)Von~Neumann, Burks, et~al.]{von1966theory}
J.~Von~Neumann, A.~W. Burks, et~al.
\newblock Theory of self-reproducing automata.
\newblock 1966.

\bibitem[Weibull(1995)]{Weibull1995-zu}
J.~W. Weibull.
\newblock \emph{Evolutionary game theory}.
\newblock MIT Press, London, England, Oct. 1995.

\bibitem[Williams and Lenton(2008)]{flask}
H.~T.~P. Williams and T.~M. Lenton.
\newblock Environmental regulation in a network of simulated microbial ecosystems.
\newblock \emph{Proceedings of the National Academy of Sciences}, 105\penalty0 (30):\penalty0 10432--10437, 2008.
\newblock \doi{10.1073/pnas.0800244105}.
\newblock URL \url{https://www.pnas.org/doi/abs/10.1073/pnas.0800244105}.

\bibitem[Yaeger(1997)]{Yaeger1997ComputationalGP}
L.~S. Yaeger.
\newblock Computational genetics, physiology, metabolism, neural systems, learning, vision, and behavior or polyworld: Life in a new context.
\newblock 1997.
\newblock URL \url{https://api.semanticscholar.org/CorpusID:14599677}.

\bibitem[Zaman et~al.(2014)Zaman, Meyer, Devangam, Bryson, Lenski, and Ofria]{coEvolutionBio}
L.~Zaman, J.~R. Meyer, S.~Devangam, D.~M. Bryson, R.~E. Lenski, and C.~Ofria.
\newblock Coevolution drives the emergence of complex traits and promotes evolvability.
\newblock \emph{PLOS Biology}, 12\penalty0 (12):\penalty0 1--9, 12 2014.
\newblock \doi{10.1371/journal.pbio.1002023}.
\newblock URL \url{https://doi.org/10.1371/journal.pbio.1002023}.

\bibitem[Zhang et~al.(2026)Zhang, Zhao, Yang, Foerster, Clune, Jiang, Devlin, and Shavrina]{zhang2026hyperagents}
J.~Zhang, B.~Zhao, W.~Yang, J.~Foerster, J.~Clune, M.~Jiang, S.~Devlin, and T.~Shavrina.
\newblock Hyperagents, 2026.
\newblock URL \url{https://arxiv.org/abs/2603.19461}.

\bibitem[Zhi-Xuan et~al.(2020)Zhi-Xuan, Mann, Silver, Tenenbaum, and Mansinghka]{zhi2020online}
T.~Zhi-Xuan, J.~Mann, T.~Silver, J.~Tenenbaum, and V.~Mansinghka.
\newblock Online bayesian goal inference for boundedly rational planning agents.
\newblock \emph{Advances in neural information processing systems}, 33:\penalty0 19238--19250, 2020.

\end{thebibliography}

\appendix

\section{Supplementary Results}

\vspace{-0.25em}
\subsection{Contextualizing Autopoietic Game Theory within prior work (Avida)}
\vspace{-0.25em}
\label{appendix:avida_comp}

While deeply inspired by Avida \citep{adami1994evolutionary}, Autopoietic Game Theory shifts the question from how pre-existing digital life complexifies to how sociality co-emerges with replication under strict energetic constraints. In the foundational Avida platform paper (2004), CPU cycles are externally allocated by the scheduler and parasitism is often expressed as competitive reallocation of computational opportunity. In contrast, our agents begin from random bytes on a shared cyclic tape, must evolve replication logic de novo, and interact under a single coupled budget where acting, stealing, and overwriting are all paid from the same local energy pool.

This distinction clarifies what we mean by \textit{metabolic drag}. In classical zero-sum arms-race settings, parasitism can remain adaptive as long as resource capture outpaces counter-adaptation. Here, stealing is intentionally lossy ($\alpha < 1$), so parasitic success can destroy enough shared energy to reduce pairwise execution speed and replication reliability for both participants. Thus, unlike pure transfer dynamics, our mechanism introduces a physically grounded negative externality that can make defection self-limiting rather than indefinitely escalatory. We view this as complementary to, rather than a replacement for, prior ALife results on CPU-cycle parasitism and host-parasite coevolution.

Our spatial findings should likewise be interpreted as building on established spatial resilience work, not superseding it. Prior studies \cite{coEvolutionBio, Dolson148973} show that structure and heterogeneous environments can protect diversity and enable innovation. Our contribution is to show that, in an autopoietic game where replication and social interaction are physically coupled, local assortment does two things at once: (i) it buffers populations against metabolic collapse under asymmetric energy fields, and (ii) it scaffolds the evolution of higher-complexity, non-stealing replicators that can sustain exogenous task performance.

Finally, our results align with ecosystem-level ALife perspectives such as ``The Flask model'' \cite{flask}, which emphasizes that organisms can modify environments in ways that stabilize or destabilize community function. The key difference is that our destabilizing/stabilizing feedback is endogenized at the micro-interaction level through execution kinetics: each social operation immediately perturbs the local thermodynamic budget that determines future compute. This allows us to connect population-level resilience directly to instruction-level energetic accounting. Similar to Avida, we use external tasks to probe capability; unlike Avida's single-agent logical task structure, we impose tasks as exogenous pressures structured explicitly as a sequential social dilemma under scarcity to test whether cooperative computation can emerge while autopoiesis is maintained.

\vspace{-0.25em}
\subsection{Additional Results Analyzing the Replication Dilemma}
\vspace{-0.25em}
\label{appendix:RepDilemma}

\begin{figure}[ht]
    \centering
    \begin{minipage}{0.48\textwidth}
        \centering
        \includegraphics[width=\linewidth]{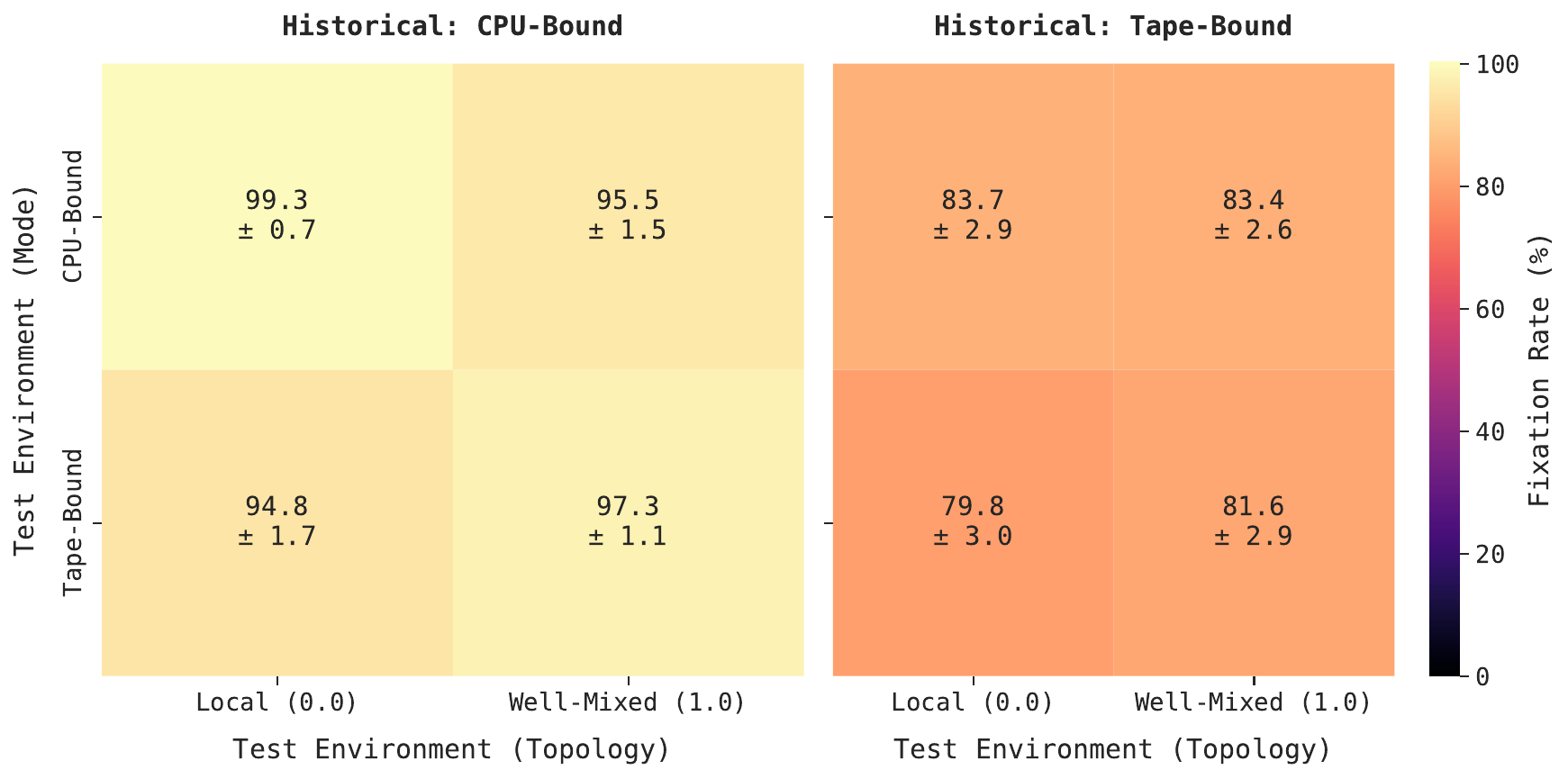}
    \end{minipage}\hfill
    \begin{minipage}{0.48\textwidth}
        \centering
        \includegraphics[width=\linewidth]{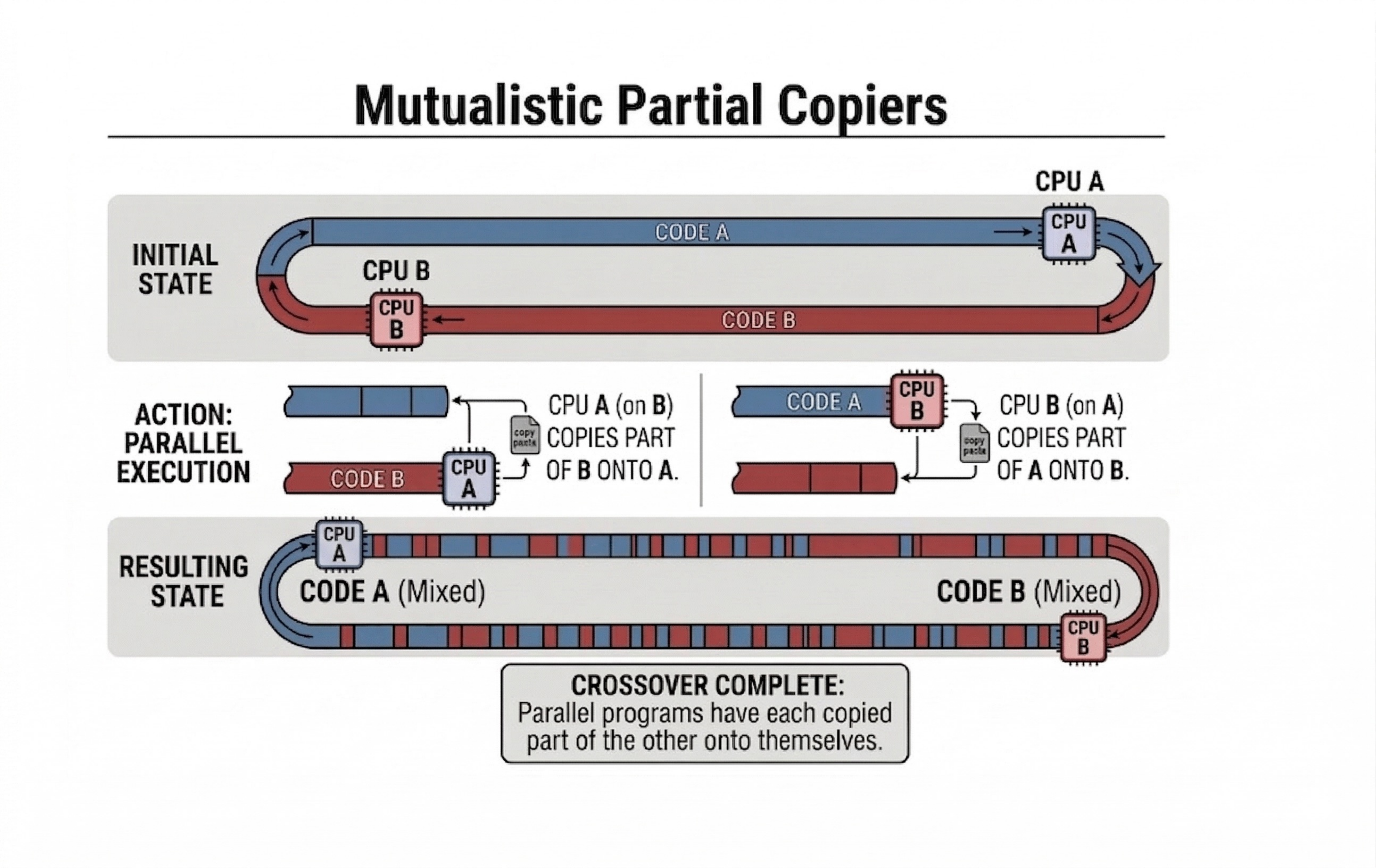}
    \end{minipage}
    \caption{(Left) Macroscopic fixation rates of representative replicators invading a naive soup. CPU-Bound replicators achieve near-universal fixation, whereas Tape-Bound replicators achieve significantly lower fixation rates if historically evolved in well-mixed settings. (Right) Example of the emergent, interdependent mutualistic replication strategy utilized by Tape-Bound agents.}
    \label{fig:invasion_and_copier}
\end{figure}

To understand how the agent-environment boundary shapes structural evolution, we first isolate the Replication Dilemma by disabling the \texttt{STEAL} operation ($\delta = 0$). We evaluate two distinct boundary definitions: \textit{CPU-as-Agent}, where operation costs are deducted from the actively executing program's budget, and \textit{Tape-as-Agent}, where costs are deducted from the accessed memory segment, regardless of the executing CPU. Populations were historically evolved under both well-mixed and local spatial topologies across various energy gradients.

Shifting this boundary fundamentally alters replication logic. Under \textit{Tape-as-Agent}, highly interdependent, mutualistic strategies emerge: rather than executing standard self-writes, CPU $A$ navigates to Tape $B$'s segment to execute writes back onto Tape $A$, forcing $B$'s budget to absorb the reproductive cost (Figure~\ref{fig:invasion_and_copier}, Right). This synergistic exploitation stabilizes almost exclusively under local spatial topologies, where repeated neighborhood interactions prevent immediate dilution. %

To assess architectural robustness, we conducted a ``petri dish'' invasion assay by seeding 1\% representative replicators into a naive random soup ($10^6$ epochs, $\mu = 0$, $\epsilon = 24$), systematically varying the test execution modes and topologies. We measured absolute fixation rates alongside continuous metrics like higher-order entropy \citep{arcas2024computational} and edit distance to capture non-trivial structural motifs. Macroscopic fixation (Figure~\ref{fig:invasion_and_copier}, Left) diverges strictly based on evolutionary history. Historically \textit{CPU-as-Agent} lineages universally achieve fixation across all test environments. Conversely, historically \textit{Tape-as-Agent} lineages only invade successfully if previously subjected to local spatial constraints or severely restricted energy. Lineages historically evolved under high energy and well-mixed topologies suffer total architectural collapse (Figure~\ref{fig:collapse}), rapidly degrading into structurally unstable 2-byte replicators that garble the shared tape.

Continuous metrics reveal significant structural imprinting ($p < 0.001$, two-sided t-test). High replication frequencies correlate with lower terminal steady-state energy and reduced edit distance (Figure~\ref{fig:energy_footprint}). Furthermore, test topologies drive distinct evolutionary attractors (Figure~\ref{fig:tradeoff_scatter}): well-mixed settings maximize genetic diversity (edit distance), while local spatial settings maximize structural complexity (higher-order entropy). Current execution modes (\textit{CPU} vs \textit{Tape}) alter replication instruction counts but do not significantly affect final energy or complexity.

Ultimately, spatial constraints seem to natively breed resilient architectures with high structural complexity, while well-mixed settings foster diversity. While the historical agent-environment boundary dictates the binary probability of successful invasion, it does not impact the continuous metrics of populations that converge. Combined with spatial constraints, this boundary fundamentally determines whether an agent evolves into a self-contained replicator or an interdependent one.

By isolating the Replication Dilemma (disabling the \texttt{STEAL} operation, $\delta = 0$), we observe how pure energetic accounting shapes structural evolution in this substrate. Even without the immediate threat of parasitic exploitation, rapid replication incurs a heavy metabolic footprint that drives down the terminal steady-state energy of the population (Figure~\ref{fig:energy_footprint}). Furthermore, an agent's evolutionary history strongly affects its architectural resilience. Lineages that historically evolved under unconstrained, well-mixed settings prove highly fragile; when forced to compete in a naive environment, they suffer architectural collapse, rapidly degrading into trivial, unstructured code (Figure~\ref{fig:collapse}). Ultimately, the topological landscape acts as a distinct evolutionary attractor. While well-mixed environments foster genetic diversity in our experiments—evidenced by higher population edit distances—local spatial constraints provide scaffolding that supports structural complexity and higher-order entropy (Figure~\ref{fig:tradeoff_scatter}).

\begin{figure}[ht]
    \centering
    \includegraphics[width=0.5\linewidth]{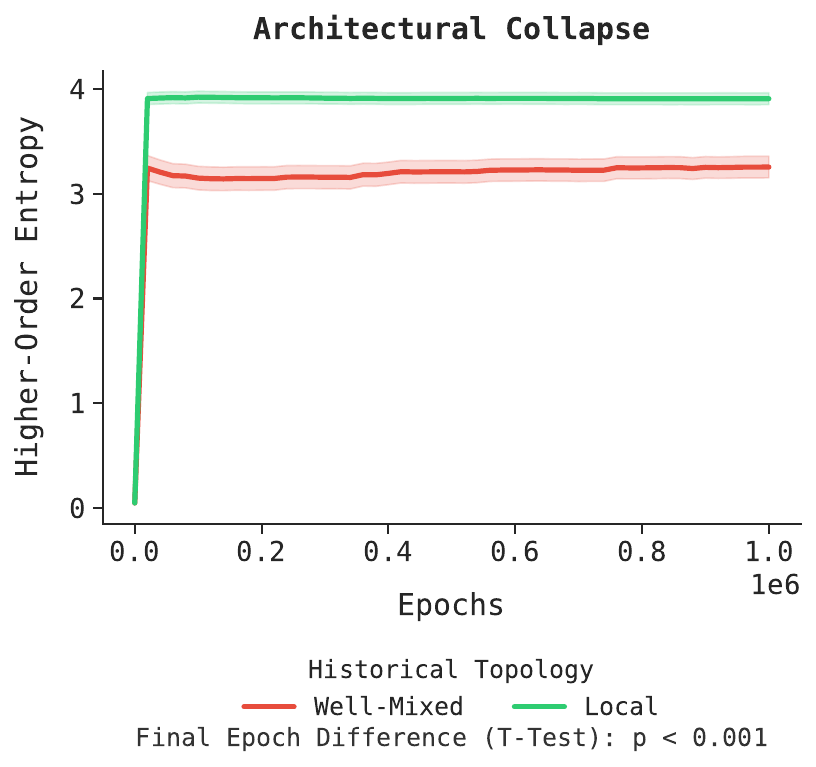}
    \caption{Timeline of architectural collapse in historical Tape-Bound replicators. Populations originally evolved in well-mixed settings experience lower higher-order entropy, indicating a takeover by trivial, unstructured code. Conversely, those evolved under local spatial constraints maintain structural complexity ($p < 0.001$, two-sided t-test comparing the means of the two groups at the final epoch).}
    \label{fig:collapse}
\end{figure}

\begin{figure}[ht]
    \centering
    \includegraphics[width=0.5\linewidth]{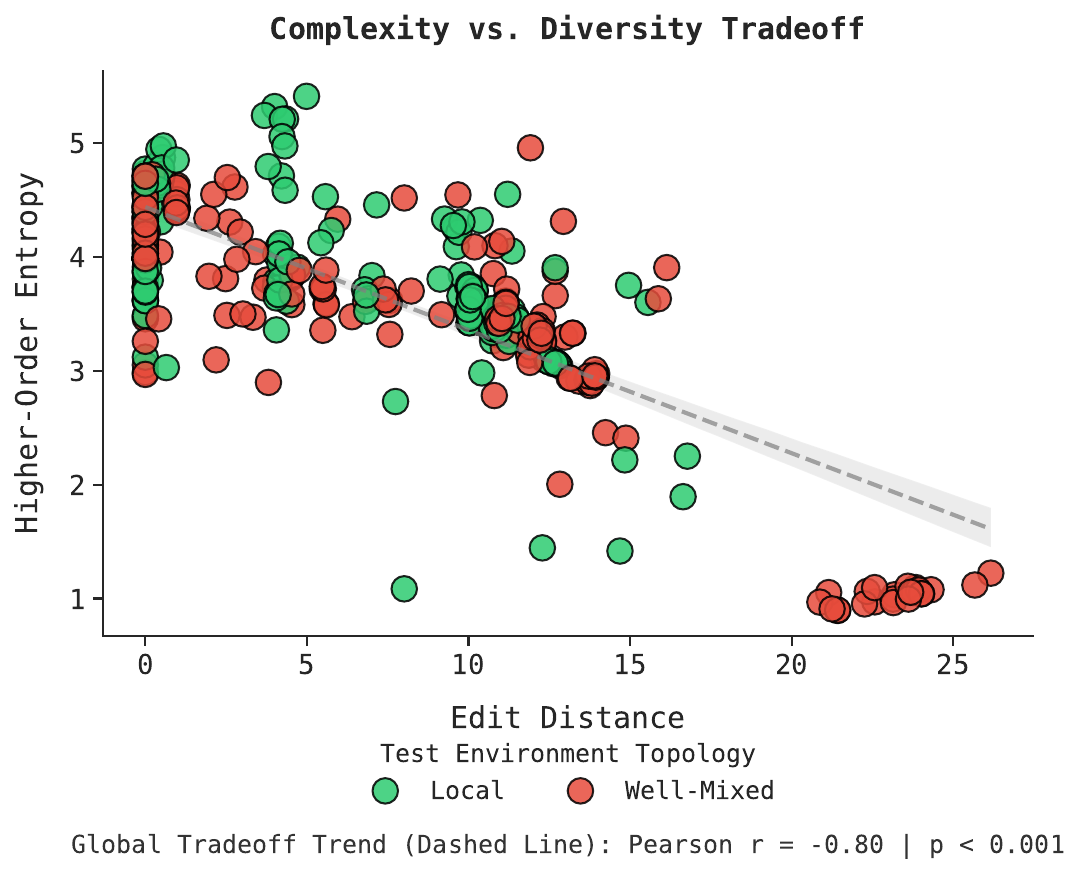}
    \caption{The structural tradeoff between complexity and diversity at the final epoch (no stealing enabled). Current test environments drive distinct evolutionary attractors: well-mixed topologies maximize genetic diversity (edit distance), while local spatial topologies maximize architectural complexity (higher-order entropy) ($p < 0.001$, two-sided t-tests comparing final metrics between well-mixed and local populations).}
    \label{fig:tradeoff_scatter}
\end{figure}

\begin{figure}[ht]
    \centering
    \includegraphics[width=0.5\linewidth]{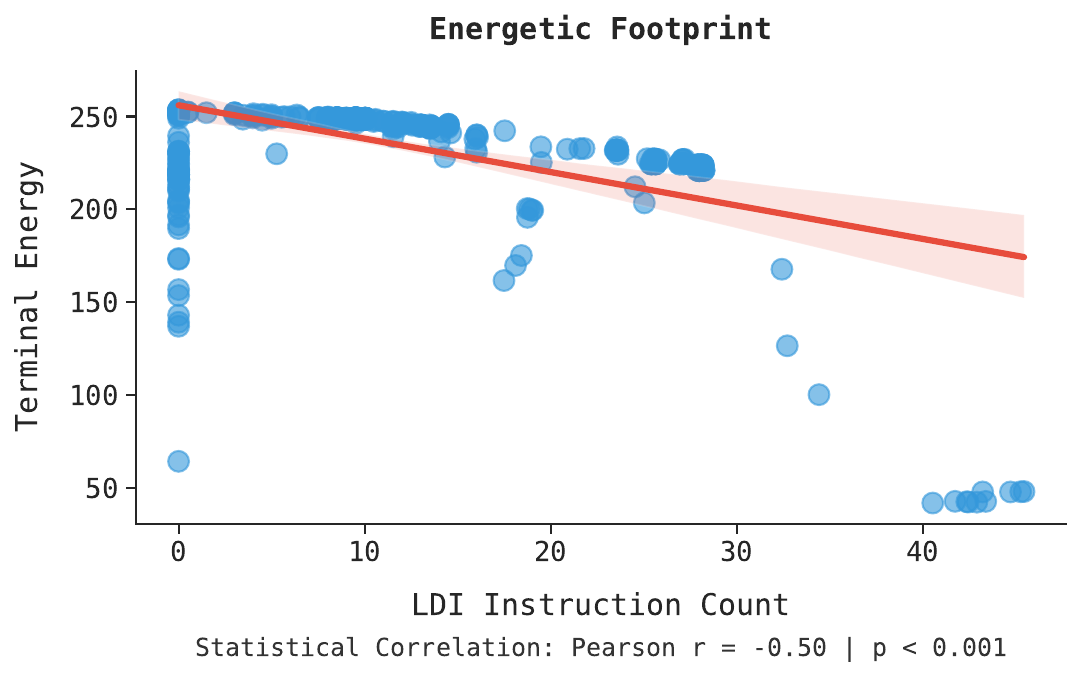}
    \caption{The energetic footprint of rapid replication. A significant inverse relationship exists between a population's replication frequency (measured via LDI instruction counts) and its terminal steady-state energy ($p < 0.001$, two-sided t-test evaluating the significance of the regression slope).}
    \label{fig:energy_footprint}
\end{figure}

\subsection{Additional Results Analyzing the Replication + Energy Dilemma}

When reintroducing the \texttt{STEAL} operation to form the unified Autopoietic Game, the structural constraints observed in pure replicators translate directly into robust defenses against parasitism. Stripping away the high-variance noise of continuous background mutation reveals a strict metabolic starvation limit: excessive parasitic theft triggers an endogenous drag that systematically starves defectors before they can complete their replication loops (Figure~\ref{fig:rq2_petridish_steals}). Consequently, rather than succumbing to the tragedy of the commons predicted by classical evolutionary game theory, autopoietic cooperators aggressively expand to dominate the carrying capacity. They fiercely resist exploitation, permanently trapping malicious actors at a distinct minority frequency (Figure~\ref{fig:rq2_timeline}). While cooperation remains dynamically stable across both topologies in benign environments, local spatial interactions act as a powerful catalyst for genetic cohesion, significantly lowering edit distance while simultaneously driving higher architectural complexity (Figure~\ref{fig:rq2_catalysts}). Most crucially, when the environment imposes severe, systematic energy inequalities, well-mixed populations entirely collapse. In these harsh conditions, spatial topologies are no longer just catalysts for complexity, but strict prerequisites for survival, providing the necessary scaffolding for populations to overcome environmental bias and sustain autopoietic stability (Figure~\ref{fig:rq2_asymmetric_grid}).

\begin{figure}[ht]
    \centering
    \includegraphics[width=0.7\linewidth]{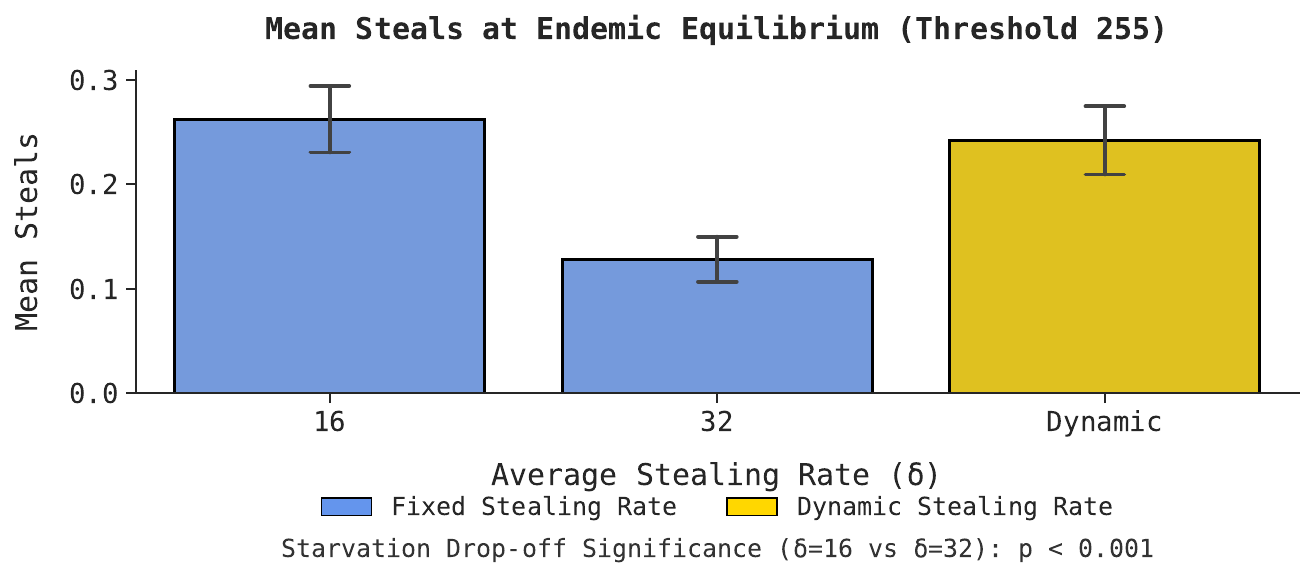}
    \caption{Mean steals at endemic equilibrium in the mutation-free petri dish assay ($\mu=0$). Removing background mutation reveals a metabolic drop-off: defector success decreases significantly from $\delta=16$ to $\delta=32$ ($p<0.001$, Cohen's $d=4.17$; two-sided t-test). The $\delta=32$ lineages maintain a marginal, non-zero mean of approximately 0.1, while dynamic stealing rates stabilize at levels comparable to $\delta=16$.}
    \label{fig:rq2_petridish_steals}
\end{figure}

\begin{figure}[ht]
    \centering
    \includegraphics[width=0.5\linewidth]{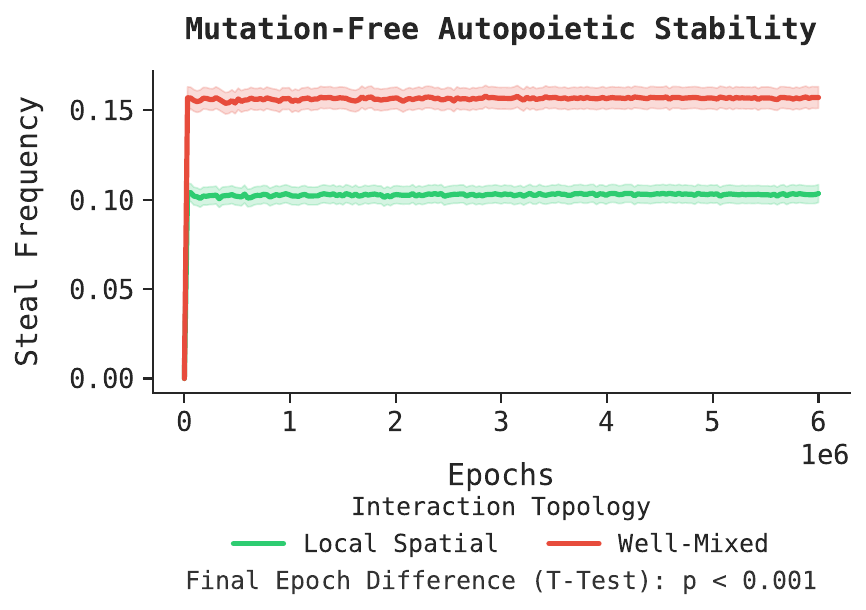}
    \caption{Mutation-Free Invasion Dynamics. Line plots tracking the invasion of 1\% defectors into a naive soup. In both well-mixed and local spatial topologies, autopoietic cooperators aggressively expand to dominate the carrying capacity, permanently trapping defectors at a minority frequency ($<20\%$).}
    \label{fig:rq2_timeline}
\end{figure}

\begin{figure}[ht]
    \centering
    \includegraphics[width=0.9\linewidth]{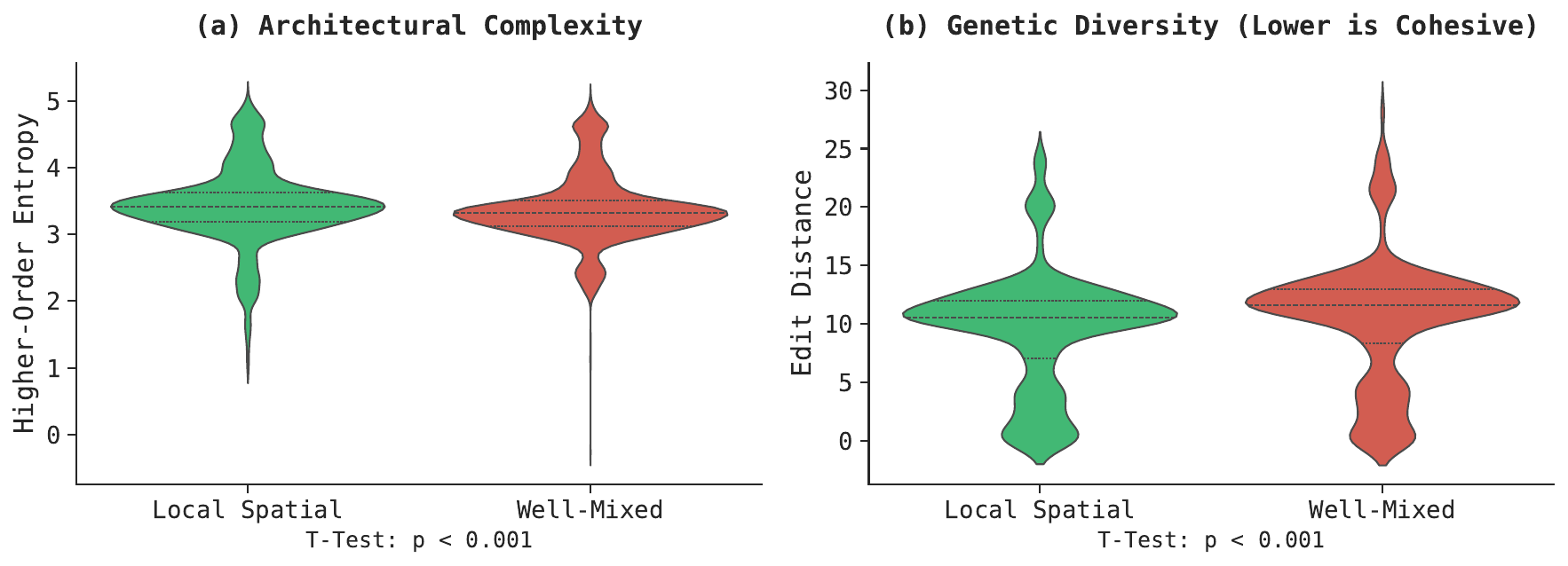}
    \caption{Topological Catalysts for Cohesion with stealing enabled. At the final epoch of the petri dish assay, local spatial interactions produce significantly higher architectural complexity (higher-order entropy; $p<0.001$, Cohen's $d=2.81$) and tighter genetic cohesion (lower edit distance; $p<0.001$, Cohen's $d=7.32$) than well-mixed interactions (two-sided t-tests).}
    \label{fig:rq2_catalysts}
\end{figure}

\subsection{Testing whether social operations are energy-coupled}
\label{appendix:SelfInterest}

In our experiments in the main text, we considered ``cooperators'' to be programs that did not steal energy while replicating, and ``defectors'' to be replicators that stole energy. We found that tapes evolved to suppress the steal action even in well-mixed settings, suggesting that non-stealing behavior can be favored when the outcomes of social interactions are directly tied to a tape's capability to replicate. However, one reasonable question to ask in this setup is: \textit{``Is stealing suppressed because social operations affect energy-mediated selection, or is low stealing simply a feature of the environment independent of selection pressure?''}

Since the base setup only includes steals, it is difficult to experimentally decouple low steals shaped by energy-mediated selection from those that are simply inevitable in the environment. As such, we designed a new setup introducing the ``share'' action, a mechanism for perfectly transferring energy from one agent to another. We place a cap of 128 units of energy, above which tapes are no longer eligible for background energy. The only way an agent can increase its energy if it is above the threshold is by stealing from other tapes, or by sharing energy with another tape to drop its own reserves below the threshold, receiving the background energy injection, and then having the partner return the shared energy.

We compare two modes: \textbf{energy-based selection}, where the tape with the highest energy always gets to run first, and \textbf{random selection}, where we randomly decide which tape gets to run first. By ``always run first'', we mean $p_i$ is fixed at 1 throughout the interaction, which is equivalent to deciding which tape is concatenated first in the experiments from \cite{arcas2024computational}. Energy-based selection is the mode used for all our prior experimental results, although in our previous experiments we dynamically calculated $p_i$ after every operation. Given the results in Figure~\ref{fig:rq2_asymmetric_grid}, we conduct these experiments using local interactions to ensure tapes can evolve behaviors capable of overcoming potential asymmetries in energy distribution.

\begin{figure}[htpb]
    \centering
    \includegraphics[width=\linewidth]{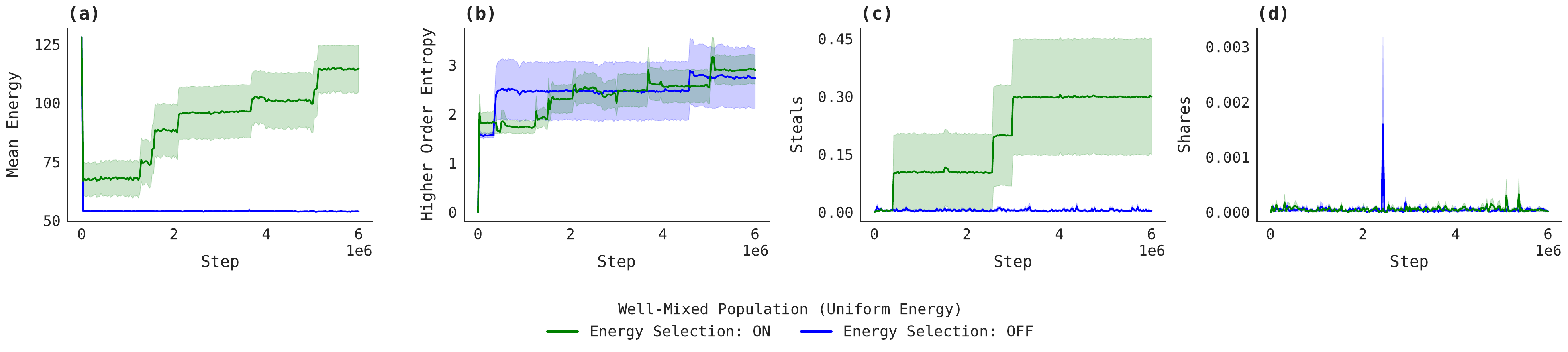}
    \caption{Population dynamics under CPU-based energy accounting with dynamically determined steal/share magnitudes ($\delta$). Subfigures show (from left to right): a) mean energy, b) higher order entropy, c) steals, and d) shares. Energy-based selection results in significantly higher utilization of social operations and a corresponding increase in mean energy compared to random selection.}
    \label{fig:cpu_dyna}
\end{figure}

\begin{figure}[htpb]
    \centering
    \includegraphics[width=\linewidth]{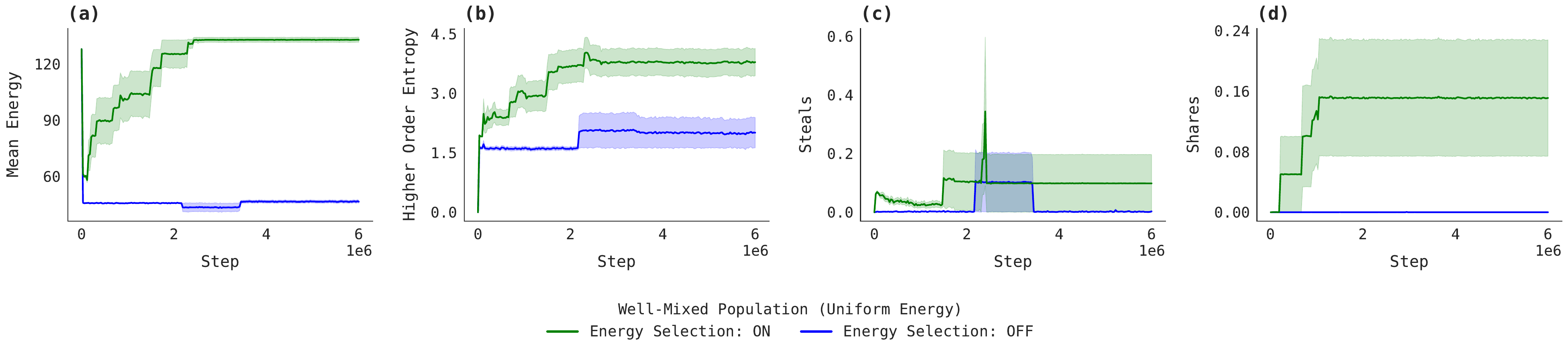}
    \caption{Population dynamics under Tape-based energy accounting with dynamically determined steal/share magnitudes ($\delta$). Subfigures show (from left to right): a) mean energy, b) higher order entropy, c) steals, and d) shares. Similar to CPU mode, tapes under energy-based selection use social operations more frequently in settings where those operations elevate system energy.}
    \label{fig:tape_dyna}
\end{figure}

If social operations were unrelated to energy-mediated selection, we would expect their overall frequency to remain low in both the energy-based and random selection cases. However, as shown for the CPU-based energy accounting mode (Figure~\ref{fig:cpu_dyna}) and Tape-based energy accounting mode (Figure~\ref{fig:tape_dyna}), this is not the case. Tapes evaluated under energy-based selection use significantly more stealing or sharing and exhibit a corresponding increase in mean energy as a result. In contrast, even when random selection utilizes an operation like a steal, there is no corresponding increase in energy or complexity, and the frequency of stealing is not statistically significantly different from not stealing at all. 

These dynamic results are from a setting where $\epsilon = 24$ and agents dynamically decide $\delta$. If we fix $\delta = 16$, as shown in Figures~\ref{fig:cpu_16} and~\ref{fig:tape_16}, we observe a similar trend: all social operations are effectively suppressed in random selection (or fail to meaningfully contribute to an increase in average system energy). Meanwhile, energy-based selection creates conditions where these operations are retained when they elevate the total energy of the system. 

\begin{figure}[htpb]
    \centering
    \includegraphics[width=\linewidth]{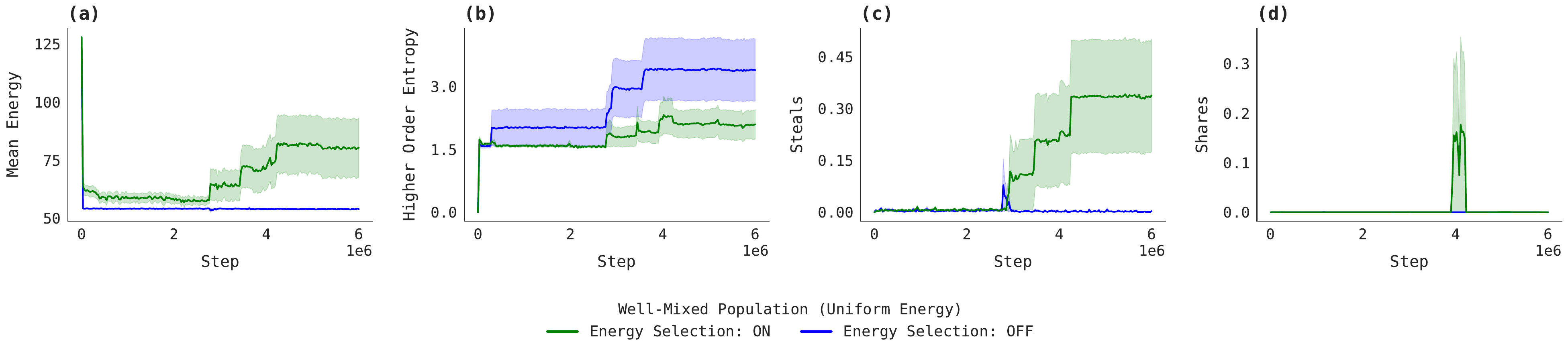}
    \caption{Population dynamics under CPU-based energy accounting with a fixed magnitude of $\delta = 16$. Subfigures show (from left to right): a) mean energy, b) higher order entropy, c) steals, and d) shares. Stealing emerges in a large portion of the population under energy-based selection, though it results in much lower higher-order entropy compared to random selection.}
    \label{fig:cpu_16}
\end{figure}

\begin{figure}[htpb]
    \centering
    \includegraphics[width=\linewidth]{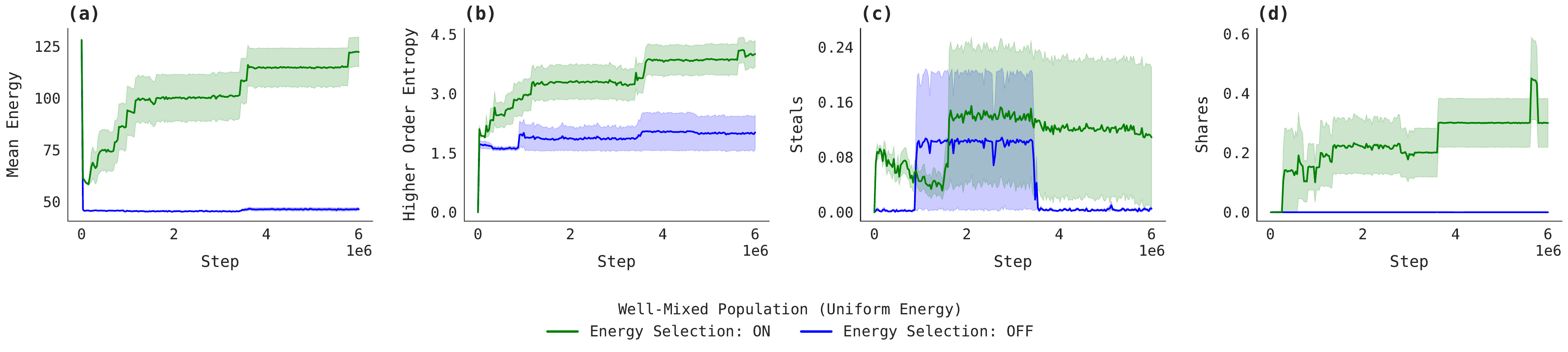}
    \caption{Population dynamics under Tape-based energy accounting with a fixed magnitude of $\delta = 16$. Subfigures show (from left to right): a) mean energy, b) higher order entropy, c) steals, and d) shares. Unlike the CPU mode, sharing dominates the population under energy-based selection, resulting in a significant boost to the higher-order entropy of the system compared to random selection.}
    \label{fig:tape_16}
\end{figure}

Interestingly, while stealing emerges across a large portion of the population in CPU mode under energy-based selection with a fixed $\delta$ (Figure~\ref{fig:cpu_16}), the resulting higher-order entropy is much lower than under random selection. Conversely, sharing dominates the population in Tape mode under energy-based selection with a fixed $\delta$ (Figure~\ref{fig:tape_16}), providing a significant boost to the higher-order entropy of the resulting soup compared to random selection. This contrast highlights the importance of how the agent/environment boundary is defined and suggests that this boundary may affect the downstream capabilities of programs to accomplish exogenous tasks.

\subsection{Additional Results Analyzing Exogenous Tasks and Agentic Utility}
\label{appendix:mathTasks}

In the main text, we introduced tasks to evaluate agentic utility. To support these tasks, special registers were introduced in the emulator: \texttt{m} for the math input/output and \texttt{c} for the task constant (offset). 

At the start of the simulation, each agent in the grid was assigned a task constant ($\text{math}_a$ or $\text{math}_b$) sampled uniformly at random from the set $\{-3, -2, -1, 1, 2, 3\}$. At the start of each interaction epoch, inputs $x$ and $y$ were sampled uniformly at random from the range $[0, 255]$. The register \texttt{m} was initialized with the input value ($x$ or $y$), and the register \texttt{c} was initialized with the assigned constant.

Agents were required to compute specific target values. The targets were defined as:
\begin{itemize}[leftmargin=*]
    \item \textbf{Solo Task A}: $x + \text{math}_a$.
    \item \textbf{Solo Task B}: $y + \text{math}_b$.
    \item \textbf{Joint Task}: $x + y$.
\end{itemize}
The reward schemes were controlled by a mode parameter:
\begin{itemize}[leftmargin=*]
    \item \textbf{Solo Mode}: Reward is given when an agent solves its own solo task.
    \item \textbf{Joint Mode}: Reward is given when agents solve the joint task ($x+y$).
    \item \textbf{Mixed Mode}: Solo and joint solutions instantiate a sequential social dilemma. If both agents solve their solo tasks, each receives $5\epsilon$. If both solve the joint task, each receives $6\epsilon$. If one agent solves the solo task while the other solves the joint task, the solo solver receives $7\epsilon$ while the joint solver receives $2\epsilon$.
\end{itemize}
Success was checked by inspecting the value of the \texttt{m} register after every operation a tape executed against the relevant target value. Rewards were added instantly to the corresponding energy budgets, bounded by the energy cap.

Here, we provide additional visual evidence of the dynamics in the math problem setting. In Figure~\ref{fig:math_communication}, we illustrate the cognitive bottleneck induced by removing shared information; solve rates drop noticeably as agents are forced to allocate computational budget to communication overhead. Further, Figure~\ref{fig:math_topology} demonstrates that spatial constraints act as active scaffolding to allow programs to coordinate on finding cooperative tasks in resource-scarce environments, confirming the role of topological assortment in overcoming severe environmental drag previously observed for simpler replicators.

\begin{figure}[htpb]
    \centering
    \includegraphics[width=0.9\linewidth]{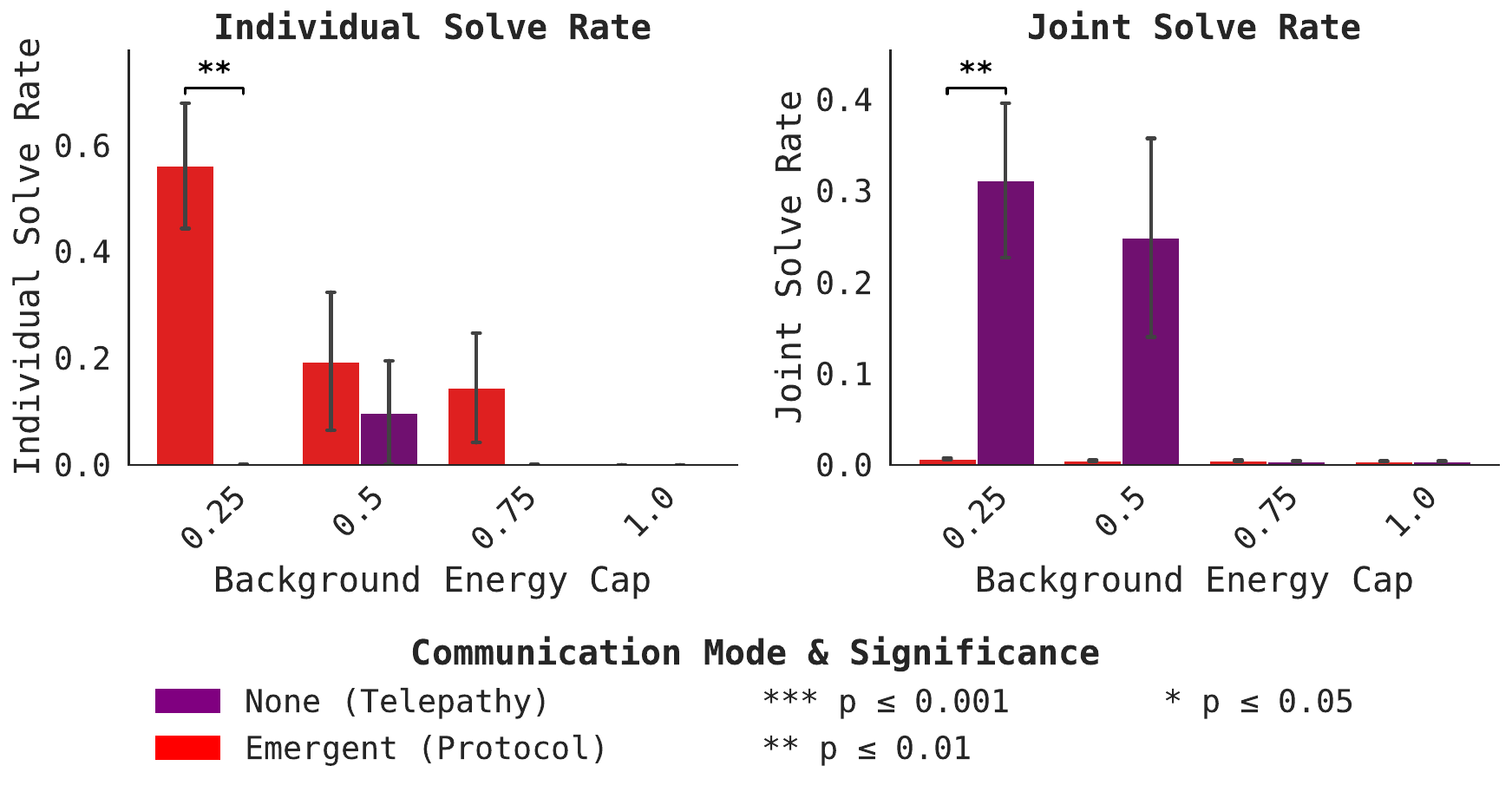}
    \caption{The cognitive bottleneck of emergent communication. When tapes in Mixed mode are forced to discover their own communication protocols to share their private variable $y$, joint solve rates crash while individual solve rates experience an increase. The feasibility of complex social behaviors is fundamentally bounded by the computational and cognitive limits of the underlying substrate.}
    \label{fig:math_communication}
\end{figure}

\begin{figure}[b]
    \vspace{-0.5em}
    \centering
    \includegraphics[width=\linewidth]{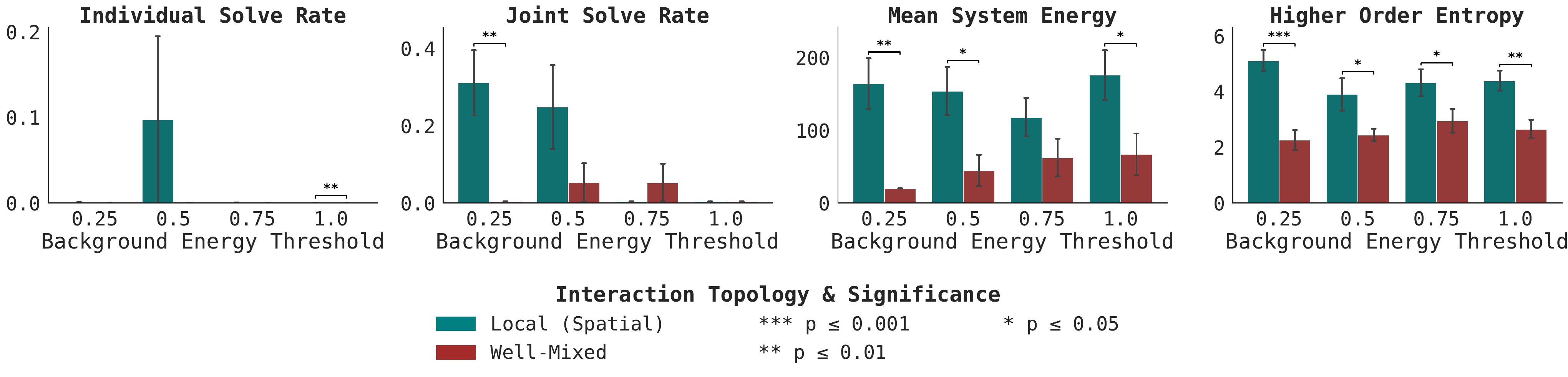}
    \caption{Spatial assortment as an architectural scaffold. In Mixed task mode, local spatial interactions produce significantly higher joint solve rates than well-mixed interactions ($p\leq0.01$, Cohen's $d=1.63$; two-sided t-test), along with greater system energy and higher-order entropy. Assortment supports the complexification needed for robust downstream problem-solving.}
    \label{fig:math_topology}
\end{figure}

\section{Sufficient Conditions for Autopoietic Cooperation under Endogenous Time and Uniform, Exogenous Energy}
\label{appendix:ProofWM}

The results below are sufficient-condition arguments under two simplifying assumptions, intended to motivate the empirical results in Section~\ref{sec:results}.

\subsection{\textbf{Formal Model Definition}}
We model a population of programs residing in a shared computational memory space. We define global environmental constants: baseline energy $\epsilon$, replication length $L$, and discrete baseline mutation probability $\mu$. We define strategy-specific evolutionary traits: absorption efficiency $\alpha \in (0, 1]$ and stealing rate $\delta > 0$.

While the population is strictly \textit{well-mixed} regarding the probability of pairing (encounters are purely frequency-dependent), the execution phase itself imposes a temporary, strictly coupled spatial constraint. 

Replicating requires a program to successfully execute $L$ operations, where each operation is a physical byte-write. We explicitly separate our timescales: micro-operations occur continuously in physical time units $t$, while population turnover and macro-scale interactions occur over discrete, overlapping generations called \textbf{epochs} (denoted by $\tau$). 

Execution and mutation occur in separate phases (though, as addressed in Corollary 1, relaxing this constraint can increase the penalty for slower executions). At the start of an epoch's execution phase, two randomly paired programs, $i$ and $j$, are concatenated to form a closed, cyclic memory tape. A program perceives the memory space as an infinite wrap-around sequence (e.g., $i \to j \to i \to j \dots$). A program can read and write into both its opponent's memory segment and its own. 

Each program receives an environmental energy baseline $\epsilon$, supplemented by residual energy retained from prior successful epochs, bounded by a strict maximum physical capacity $E_{max}$. Thus, the initial energetic state for the epoch is $E_i = \min(\epsilon + E_{i, residual}, E_{max})$ and $E_j = \min(\epsilon + E_{j, residual}, E_{max})$. We assume the baseline environment is viable for baseline replication, meaning $2\epsilon > 2L$.

To prevent infinite halting states, epochs are strictly bounded by a maximum physical time $T_{max}$. An epoch's execution phase concludes only when one of the following conditions is met: 
1) Both programs enter a \texttt{halt} state. 
2) The system succumbs to metabolic starvation ($E_{tot} \le 0$). 
3) The system reaches the absolute time boundary $T_{max}$.

A program enters a \texttt{halt} state either by successfully completing $L$ targeted write operations into its opponent's memory segment, or by being completely overwritten by its opponent. A halted program ceases to execute operations. If one program halts while the other remains active, the active program continues to execute operations until it too halts or the system starves. Once both programs halt or starve, the system's physical time is instantly sped up to the boundary $T_{max}$.

At the end of the epoch, the concatenated tape is strictly cleaved in half. If program $i$ successfully targeted and completed $L$ writes into $j$'s segment, the split yields two intact copies of $i$. The original half retains its remaining $E_i$ as residual energy, and the newly written half inherits the remaining $E_j$ as its residual energy. If starvation occurs before either program completes $L$ targeted opponent-writes, the tape splits into a recombination of both programs, yielding 0 viable copies of either of the originals.

Within the execution phase, the absolute execution speed of the paired system is proportional to the total available energy, $E_{tot} = E_i + E_j$. While execution progress is measured in discrete operation steps, the \textit{physical time} $dt$ required to compute each step is dynamic. We must track continuous physical time because the ultimate probability of surviving the post-epoch mutation phase depends strictly on the total elapsed physical time $T$, not the discrete number of operations. The physical time $dt$ required to complete a single micro-operation is inversely proportional to the energy pool:
$$dt = \frac{1}{E_i + E_j}$$
The sequence of operations is modeled as a probabilistic Markov chain. At any step where $E_{tot} > 0$, the probability that an active program $i$ executes the next operation relative to active program $j$ is its proportion of the energy pool: $p_i = \frac{E_i}{E_i + E_j}$. If $E_{tot} \le 0$, the system starves ($dt \to \infty$) and operations cease. 

Programs utilize one of two strategies:
\begin{itemize}[leftmargin=*]
    \item \textbf{Cooperator ($C$):} Executes a normal write operation targeted at the opponent ($\Delta E_i = -1, \Delta E_j = 0$). The system loses 1 unit of total energy.
    \item \textbf{Defector ($D(\alpha, \delta)$):} Executes write operations interspersed with or coupled to parasitic ``steal'' mechanics. $D(\alpha, \delta)$ targets and absorbs energy from the opponent at an expected average rate of $\delta > 0$ per operation across its $L$ writes, with an average absorption efficiency $\alpha \in (0, 1]$. Therefore, the expected change in energy per operation executed by $D(\alpha, \delta)$ is $\Delta E_i = \alpha\delta - 1$ and $\Delta E_j = -\delta$, subject to the boundary $E_i \le E_{max}$. On average, each operation by $D(\alpha, \delta)$ drains the system of $1 + (1-\alpha)\delta$ units of total energy.
\end{itemize}

After the tape splits, both halves enter a discrete \textbf{mutation phase}. The probability of an intact program surviving this background mutation evaluation is $P = \exp(-\mu T)$, where $\mu \in (0, 1)$ acts as a discrete baseline probability, and $T$ is the total physical time required to complete the execution phase, capped at $T_{max}$.

\subsection{\textbf{Intra-Epoch Dynamics: Kinetic Slowdown and Metabolic Starvation}}

To formally evaluate the execution time trajectories of competing programs, we first define the energetic boundaries and total system cost of defection. Assume the baseline environmental parameters $\epsilon$ and $L$ satisfy the viability condition $2\epsilon > 2L$. Let $E_{tot} = E_i + E_j$ represent the total available energy in a paired system.

During any single operation executed by a Defector $D(\alpha, \delta)$, the expected change in total system energy is strictly negative:
$$\Delta E_{tot} = (\alpha\delta - 1) - \delta = \delta(\alpha - 1) - 1$$
Because the absorption efficiency is bounded ($\alpha \le 1$), the net energy drain per operation is $\Delta E_{tot} \le -1$.

To establish a comparative baseline, we define $K(\epsilon, L) > 0$ as the unique critical inefficiency threshold, implicitly defined by the following equality assuming baseline initial energy ($E_{residual} = 0$):
$$\sum_{k=0}^{L-1} \frac{1}{2\epsilon - k(1 + K)} = \sum_{k=0}^{2L-1} \frac{1}{2\epsilon - k}$$
This threshold $K(\epsilon, L)$ represents the exact break-even point where the energetic penalty of defection perfectly cancels out the physical time saved by executing fewer operations ($L$ writes for a Defector versus $2L$ writes for mutual Cooperators).

\begin{lemma}[Endogenous Metabolic Drag]
\label{lem:metabolic-drag}
Given the critical inefficiency threshold $K(\epsilon, L)$, if a Defector's inefficiency penalty strictly exceeds this break-even point such that $(1-\alpha)\delta > K(\epsilon, L)$, the rate of energy destruction guarantees that the absolute best-case execution time for a Defector is strictly slower than the execution time of mutual Cooperators ($T_{D(\alpha, \delta)|C} > T_{C|C}$).
\end{lemma}
\begin{proof}
For a $(C, C)$ pair, both programs must complete $L$ operations to successfully replicate, after which they explicitly halt. Because every operation by either program reduces $E_{tot}$ by exactly 1 unit, the total system energy decreases from $2\epsilon$ to $2\epsilon - 2L$. Regardless of the execution order, the total physical time required for mutual halting is exactly:
$$T_{C|C} = \sum_{k=0}^{2L-1} \frac{1}{2\epsilon - k}$$

For a $(D(\alpha, \delta), C)$ pair, the execution sequence is probabilistic. To successfully reproduce, $D(\alpha, \delta)$ must complete $L$ targeted write operations, completely overwriting $C$ and forcing it to halt. Each operation by $D(\alpha, \delta)$ reduces $E_{tot}$ by an expected average of $1 + (1-\alpha)\delta$. 

Because the physical time per operation step, $dt = \frac{1}{E_{tot}}$, increases as $E_{tot}$ decreases, any operations performed by $C$ strictly reduce $E_{tot}$ and thereby increase the time delay for subsequent steps. To establish a highly conservative lower bound, we calculate the absolute fastest (best-case) trajectory for $D(\alpha, \delta)$, assuming $C$ is entirely starved of execution steps during $D(\alpha, \delta)$'s execution. The absolute minimum time required for $D(\alpha, \delta)$ to overwrite $C$ must be strictly greater than or equal to the time taken for $D(\alpha, \delta)$'s consecutive operations alone:
$$T_{D(\alpha, \delta)|C} \ge \sum_{k=0}^{L-1} \frac{1}{2\epsilon - k(1 + (1-\alpha)\delta)}$$
Provided $(1-\alpha)\delta > K(\epsilon, L)$, this metabolic drag ensures that even in its absolute best-case scenario, the Defector takes longer to halt the epoch than mutual Cooperators ($T_{D(\alpha, \delta)|C} > T_{C|C}$). 
\end{proof}

\paragraph{Approximation caveat.}
This result replaces the actual, changing amount of energy in the system at each step with its expected value, using Jensen's inequality. This is an approximation: the realized total energy can deviate from this expectation because every operation changes the agents' energy, and the exact trajectory depends on which agent executes and how much stealing occurs. The experiments in Section~\ref{sec:results} evaluate the full dynamics, tracking the realized energy after every operation without this substitution.

\begin{lemma}[Metabolic Starvation Limit]
\label{lem:metabolic-starvation}
If the average stealing rate $\delta$ is sufficiently large, or the environmental energy $\epsilon$ is sufficiently small such that
$$2\epsilon < L(1 + (1-\alpha)\delta),$$
a Defector $D(\alpha, \delta)$ will exhaust the total energy pool before completing $L$ write operations, yielding an instant jump to $T_{max}$ and 0 viable copies post-split.
\end{lemma}
\begin{proof}
During a $(D(\alpha, \delta), C)$ interaction, the expected energy remaining after $D(\alpha, \delta)$ completes $k$ operations (assuming the best-case bound where $C$ executes none) is $2\epsilon - k(1 + (1-\alpha)\delta)$. If $2\epsilon < L(1 + (1-\alpha)\delta)$, the denominator in the sum for $T_{D(\alpha, \delta)|C}$ reaches zero for some $k<L$. The system starves ($dt \to \infty$), operations cease, and time speeds up to $T_{max}$. Because $D(\alpha, \delta)$ failed to complete $L$ targeted writes, the tape splits into garbled, incomplete code segments. Mutual defection $(D(\alpha, \delta), D(\alpha, \delta))$ under the same strict inequality similarly starves before either side completes $L$ writes, resulting in 0 intact copies.

At the boundary $2\epsilon = L(1 + (1-\alpha)\delta)$, starvation occurs at the $k=L$ threshold, so the statement ``before completing $L$ writes'' is no longer guaranteed without additional tie-breaking assumptions.
\end{proof}

\paragraph{Approximation caveat.}
This result replaces the actual, changing amount of energy in the system at each step with its expected value, using Jensen's inequality. This is an approximation: the realized total energy can deviate from this expectation because every operation changes the agents' energy, and the exact trajectory depends on which agent executes and how much stealing occurs. The experiments in Section~\ref{sec:results} evaluate the full dynamics, tracking the realized energy after every operation without this substitution.

\begin{corollary}[Continuous Mutation Penalty]
While the model designates mutation as a discrete post-epoch phase for analytical clarity, applying mutation continuously during the execution phase would penalize slower executions more strongly. Because $T_{D(\alpha, \delta)|C} > T_{C|C}$ under the metabolic drag condition, a continuous mutational hazard $\exp(-\mu \int dt)$ exposes the Defector to a higher probability of mid-execution fatal errors, aborting the overwrite prematurely. Thus, the discrete mutation phase serves as a conservative lower-bound penalty within this regime.
\end{corollary}

\subsubsection{\textbf{Continuous Approximation of the Lower Bound Threshold $K(\epsilon, L)$}}
To establish a computable and strict boundary condition, we evaluate the extreme scenario where $D(\alpha, \delta)$'s fastest possible continuous execution (exactly $L$ operations) remains slower than $C$'s continuous execution (exactly $2L$ operations). We approximate the discrete harmonic sums of this boundary using definite continuous integrals. Note that this continuous approximation holds best for $L \gg 1$, as harmonic sums and logarithmic integrals can diverge slightly for small $L$:
$$\int_{0}^{L} \frac{1}{2\epsilon - x(1 + K)} \, dx \approx \int_{0}^{2L} \frac{1}{2\epsilon - x} \, dx$$
Evaluating and equating both integrals yields a transcendental equation, numerically solvable for any $\epsilon$ and $L$, which bounds the minimal inefficiency required to induce a metabolic slowdown:
$$\ln\left( \frac{2\epsilon}{2\epsilon - L(1+K)} \right) = (1+K) \ln\left(\frac{\epsilon}{\epsilon - L}\right)$$

\subsection{\textbf{Macro-Scale Survival and Population Dynamics}}

\begin{conjecture}[Asymmetric Recovery Criterion (Expected-State Approximation)]
\label{conj:asymmetric-recovery}
Consider a $(C,\, D(\alpha,\delta))$ interaction where $C$ enters with $E_C = E_{max}$ and $D$ enters with $E_D = \epsilon$, with $E_{max} > \epsilon$ and $\alpha\delta > 1$. Define the positive root of the quadratic $q(E) = E^2 - E(\epsilon + L) - L\epsilon\bigl[(1+\alpha)\delta - 1\bigr]$:
$$\tilde{E}^* \;:=\; \frac{(\epsilon + L) + \sqrt{(\epsilon+L)^2 + 4L\epsilon\bigl[(1+\alpha)\delta - 1\bigr]}}{2}.$$
If $E_{max} > \tilde{E}^*$, then the first-order frozen-priority approximation yields
\[
\mathbb{E}\!\left[\mathcal{D}\!\left(L,\,L\frac{\epsilon}{E_{max}}\right)\right] > 0.
\]
\end{conjecture}

\paragraph{Dynamic-priority caveat.}
This conjecture assumes that each agent's execution probability relative to its partner remains fixed at its initial value. This is the true relative probability when the tapes are first paired, but every operation changes their energy. A defector that steals repeatedly may therefore become faster than this approximation assumes, while an agent that performs many costly operations may become slower relative to its partner. We use the initial probability as a simple approximation to these dynamic execution priorities. The experiments in Section~\ref{sec:results} evaluate the full dynamics, updating relative execution probabilities after every operation.

\paragraph{Supporting calculation under frozen execution priorities.}
Define the energy differential $\mathcal{D}(n_C, n_D) := E_C - E_D$ after $n_C$ C-writes and $n_D$ D-steal operations. From the per-operation energy updates:
\begin{itemize}[leftmargin=*]
    \item Each C-write: $\Delta E_C = -1$, $\Delta E_D = 0$, so $\mathcal{D}$ decreases by $1$.
    \item Each D-steal: $\Delta E_C = -\delta$, $\Delta E_D = \alpha\delta - 1$, so $\mathcal{D}$ decreases by $(1+\alpha)\delta - 1 > 0$ (since $\alpha\delta > 1$).
\end{itemize}
Therefore:
$$\mathcal{D}(n_C,\, n_D) \;=\; (E_{max} - \epsilon) \;-\; n_C \;-\; n_D\bigl[(1+\alpha)\delta - 1\bigr].$$

\textit{Expected D-operation count.} At the start of the interaction, $C$'s execution priority is $p_C^{(0)} = E_{max}/(E_{max}+\epsilon)$ and $D$'s is $p_D^{(0)} = \epsilon/(E_{max}+\epsilon)$. Treating these initial priorities as fixed, the expected number of D-operations by the time $C$ completes $L$ writes is:
$$\mathbb{E}[n_D \mid n_C = L] \;\approx\; L\cdot\frac{p_D^{(0)}}{p_C^{(0)}} \;=\; L\cdot\frac{\epsilon}{E_{max}}.$$
Because $\alpha\delta > 1$, D gains net energy per steal, so its priority can grow over time; the initial ratio $\epsilon/E_{max}$ is therefore a lower bound on the realized ratio. Since $\mathcal{D}$ subtracts a positive multiple of $n_D$, substituting this lower bound yields an upper estimate on $\mathbb{E}[\mathcal{D}]$, so positivity of that estimate is interpreted here as an expected-state indicator rather than a formal sufficiency guarantee.

\textit{Expected differential at replication completion.}
$$\mathbb{E}\!\left[\mathcal{D}\!\left(L,\,L\frac{\epsilon}{E_{max}}\right)\right] \;=\; (E_{max}-\epsilon) - L - \frac{L\epsilon}{E_{max}}\bigl[(1+\alpha)\delta-1\bigr].$$
Multiplying through by $E_{max}$ and requiring positivity yields:
$$E_{max}(E_{max}-\epsilon) - LE_{max} - L\epsilon\bigl[(1+\alpha)\delta-1\bigr] > 0,$$
$$E_{max}^2 - E_{max}(\epsilon+L) - L\epsilon\bigl[(1+\alpha)\delta-1\bigr] > 0,$$
which holds if and only if $E_{max} > \tilde{E}^*$.

\begin{theorem}[Local Favorability under Metabolic Starvation]
\label{thm:local-favorability}
In a well-mixed, finite population with overlapping generations and uniformly distributed energy, non-stealing replication is locally favored when
\[
2\epsilon < L\bigl(1+(1-\alpha)\delta\bigr).
\]
Under this condition, stealing destroys enough shared energy to prevent faithful defector replication, giving defective lineages lower expected reproductive success than non-stealing lineages.
\end{theorem}
\begin{proof}
Let $\tau$ denote the discrete epoch counter and
$$x(\tau) = \frac{N_C(\tau)}{N_C(\tau)+N_D(\tau)}.$$
In this substrate, each pair always splits back into two 32-byte programs, so total population size is fixed. We therefore track \emph{strategy-frequency conversion} (which strategy occupies slots next epoch), not demographic birth/death.

As a lightweight bookkeeping approximation, write one-epoch frequency dynamics as
$$\begin{bmatrix}x(\tau+1)\\ 1-x(\tau+1)\end{bmatrix}
\approx
\underbrace{\begin{bmatrix}
p_{CC}(\tau) & p_{DC}(\tau)\\
p_{CD}(\tau) & p_{DD}(\tau)
\end{bmatrix}}_{P(\tau)}
\begin{bmatrix}x(\tau)\\ 1-x(\tau)\end{bmatrix},$$
where $p_{ab}(\tau)$ is the expected probability that a slot occupied by strategy $b$ at epoch $\tau$ is occupied by strategy $a$ at epoch $\tau+1$, after execution, split, and mutation. Local favorability of cooperation corresponds to positive drift in $x(\tau)$, equivalently $p_{DC}(\tau) > p_{CD}(\tau)$ in the relevant regime.

\textbf{Regime 1: Metabolic Starvation.} Assume
$$2\epsilon < L(1 + (1-\alpha)\delta).$$
By Lemma 2, baseline $(D,C)$ and $(D,D)$ interactions starve before $D$ completes $L$ writes. Critically, starvation still returns two 32-byte programs after the split, but produces no faithful defector overwrite event in that encounter. Thus $p_{CD}(\tau)$ is strongly suppressed in this regime, while successful $(C,C)$ interactions keep $p_{DC}(\tau)$ nonzero; this yields positive drift toward cooperation whenever $x(\tau)>0$. This establishes local favorability of non-stealing replication under the stated metabolic starvation condition.
\end{proof}

\begin{conjecture}[Frequency-Dependent Rebound]
\label{conj:frequency-rebound}
Suppose $\alpha\delta>1$ and $E_{max}>\tilde{E}^*$, where $\tilde{E}^*$ is defined in Conjecture~\ref{conj:asymmetric-recovery}. When defectors are trapped in a low-energy state by repeated destructive interactions, rare cooperative lineages can accumulate residual energy through successful $(C,C)$ interactions. This accumulated energy may allow cooperation to rebound from low frequency by increasing cooperators' probability of controlling subsequent replication events.
\end{conjecture}

\paragraph{Dynamic-priority caveat.}
This conjecture assumes that each agent's execution probability relative to its partner remains fixed at its initial value. This is the true relative probability when the tapes are first paired, but every operation changes their energy. A defector that steals repeatedly may therefore become faster than this approximation assumes, while an agent that performs many costly operations may become slower relative to its partner. We use the initial probability as a simple approximation to these dynamic execution priorities. The experiments in Section~\ref{sec:results} evaluate the full dynamics, updating relative execution probabilities after every operation.

\paragraph{Supporting argument.}
When cooperators are rare, most defectors encounter other defectors. If these interactions starve the system before either defector completes a faithful overwrite, defectors remain near the baseline energy state and exert little conversion pressure.

Rare $(C,C)$ interactions can instead succeed whenever $2\epsilon>2L$. Conditional on success, residual energy evolves as
\[
E_{C,\mathrm{res}}(\tau+1)
=
\min\!\bigl(E_{C,\mathrm{res}}(\tau)+\epsilon,E_{max}\bigr)-L.
\]
Because $\epsilon>L$, every successful interaction increases residual energy by $\epsilon-L$. A lineage that realizes
\[
\tau^*
=
\left\lceil
\frac{E_{max}-\epsilon}{\epsilon-L}
\right\rceil
\]
successful $(C,C)$ interactions therefore reaches the maximum energy state. This is a positive-probability path, not a deterministic trajectory for every cooperative lineage.

If a high-energy cooperator subsequently encounters a baseline-energy defector, the frozen-priority calculation in Conjecture~\ref{conj:asymmetric-recovery} predicts that the cooperator can retain an execution advantage long enough to complete its overwrite first. These asymmetric encounters would allow cooperative lineages to recover while defectors remain constrained by destructive interactions. This argument motivates the rebound observed empirically, but it is not a proof for the fully dynamic execution process.

\subsubsection{\textbf{The Metabolic and Replication Dilemmas}}
The mechanics described above implicitly establish a metabolic social dilemma regarding computational budget utilization. An individual program has a strong incentive to increase its personal energy reservoir via defection, as this maximizes its immediate execution speed ($p_i$) up to the physical limit $E_{max}$. However, this selfish optimization comes at a systemic cost: it destroys the shared energy pool, increases total physical execution time, and exposes both halves of the cyclic tape to a higher probability of mutational degradation or total starvation.

This \textbf{Metabolic Dilemma} couples directly with a secondary \textbf{Replication Dilemma}. In a competitive computational environment where the cyclic tape is split after $L$ operations, a program faces a fundamental trade-off regarding the fidelity, targeting, and completeness of its replication. There is a selfish incentive to replicate only the minimum amount of code necessary for basic survival (i.e., reducing $L$ to $L/2$); doing so requires less energy and guarantees a faster tape split. However, executing a partial overwrite leaves half of the opponent's tape intact. If that intact code contains a Defector's steal module, the resulting ``offspring'' tape is corrupted or vulnerable. 

Furthermore, because the memory space wraps around ($i \to j \to i \to j \dots$), programs must expend computational effort to correctly calculate the offset of their writes. A program that optimizes purely for speed might accidentally overwrite its own operational code. If we assume that higher-fidelity, accurately targeted overwriting requires more physical writes ($L$) and complex instructions, then investing in the Replication Dilemma directly penalizes the program in the Metabolic Dilemma. Consequently, the strategies for managing energetic throughput, targeting accuracy, and memory-overwrite security cannot be treated in isolation; they are deeply coupled, with the constraints of the computational economy dynamically shaping the bounds of the replication dilemma and vice versa.

\subsection{\textbf{Thermodynamic and Physical Justification for Kinetic Execution Constraints}}
\label{appendix:physics_justification}

A critical foundation of the endogenous metabolic drag established in Lemma 1 is the kinetic execution rule, which posits that the physical time required to execute a computational operation is inversely proportional to the total system energy: $dt \propto \frac{1}{E_{tot}}$. Rather than an arbitrary game-theoretic penalty, this formulation directly models the fundamental thermodynamic and quantum limits of computation.

First, the absolute limits of computational speed are governed by the Margolus-Levitin theorem \cite{Margolus_1998}. The theorem establishes that the minimum time $\Delta t$ required for a quantum system to evolve from one state to an orthogonal state---the fundamental physical requirement for executing a logical bit flip or computational step---is strictly bounded by its average energy $E$:
$$ \Delta t \ge \frac{h}{4E} $$
where $h$ is Planck's constant. By explicitly defining our execution speed as inversely proportional to the available energy pool, our model macroscopically enforces this fundamental quantum speed limit. 

Furthermore, every write operation in the Autopoietic Game, particularly the targeted overwriting of an opponent's tape executed by a Defector, constitutes a logically irreversible computation. By Landauer's principle \cite{Landauer1961-sj}, the destruction of information intrinsically requires the dissipation of energy into the environment (bounded strictly by $k_B T \ln 2$ per bit). Thus, the energetic cost $c$ imposed per operation, and the ensuing degradation of the shared energy pool caused by parasitic behavior, are physical inevitabilities of the agents' structural choices.

Finally, this inverse proportionality holds true in applied digital hardware architectures. In traditional CMOS circuits, the relationships governing dynamic voltage and frequency scaling (DVFS) dictate that as available power is constrained or depleted, the operating voltage must drop, which proportionally increases the switching delay of logic gates \cite{Chandrakasan1995-ee}. The localized ``brownout'' induced by a Defector's \texttt{STEAL} operation physically starves the local circuit, forcing a dynamic reduction in clock frequency that severely elongates the physical time step $dt$. 

By grounding the kinetic execution rule in both quantum mechanics and classical hardware dynamics, the endogenous metabolic drag transitions from a simulated game mechanic to an inescapable physical reality of shared-resource computation.

\section{Implementation Details and Experimental Reproduction}
\label{appendix:implementation_details}

In this section, we provide a comprehensive overview of the implementation of the Autopoietic Game Theory framework in Z80, the representation of the computational soup, the experimental protocols, and the evaluation metrics.

\subsection{Toy Prisoner's Dilemma Simulations
(Figures~\ref{fig:repMechToy} and~\ref{fig:coEvoToy})}
\label{appendix:toy_sim_details}

Section~\ref{sec:AGT} uses simple, mutation-free simulations to isolate how replication timing affects cooperation. Unless otherwise noted, each experiment uses $N=1000$ agents, runs for $500$ generations, and is repeated across $1000$ independent seeds.

\textbf{Interaction and energy.}
Each agent either cooperates ($C$) or defects ($D$). The payoff to each agent is
\[
(C,C): 2,\qquad
(C,D): (-2,\,3+d),\qquad
(D,D): d-1,
\]
where $d$ controls the energetic outcome of mutual defection. We use $d\in\{0.5,1,2\}$, so each defector gains $-0.5$, $0$, or $1$ energy in a $(D,D)$ interaction. The corresponding total change for the pair is $-1$, $0$, or $2$, producing the Drain, Stagnate, and Accumulate environments.

At each generation, agents are randomly paired and complete one encounter. Payoffs are added to their accumulated energy, clipped to $[0,E_{\max}]$, and retained across generations. All agents remain in the population and are randomly paired again in the next generation.

\textbf{Fixed replication timing.}
Figure~\ref{fig:repMechToy} compares three replication rules:

\begin{enumerate}[leftmargin=*]
    \item \textbf{Replicate before interaction:} replication priority is determined from pre-interaction energy, and the winner overwrites the loser's strategy before the Prisoner's Dilemma is played.
    
    \item \textbf{Decide before, replicate after interaction:} replication priority is determined from pre-interaction energy, but the overwrite occurs after the original pair plays the Prisoner's Dilemma.
    
    \item \textbf{Replicate from post-interaction energy:} the pair first plays the Prisoner's Dilemma, after which the updated energy levels determine replication priority.
\end{enumerate}

Whenever replication priority is evaluated, agent $a$ wins with probability
\[
\frac{E_a}{E_a+E_b},
\]
or with probability $1/2$ if both energies are zero. The winner's strategy then overwrites the loser's, leaving both positions with the winning strategy.

\textbf{Co-evolving replication timing.}
For Figure~\ref{fig:coEvoToy}, each agent inherits both a social strategy and one of the three replication rules above. When two agents meet, their current energies first determine which agent's replication rule governs the encounter. That rule then determines when energy is evaluated and when one agent's complete genome---its social strategy and replication rule---overwrites the other. Thus, successful replication spreads both what an agent plays and when it replicates.

\textbf{Figure settings.}
Figure~\ref{fig:repMechToy} uses the Drain environment ($d=0.5$), initial energy $10$, $E_{\max}=300$, and initial defector frequencies
$\{0.1,0.2,0.4,0.6,0.8,0.9\}$. Figure~\ref{fig:coEvoToy} initializes social strategies and replication rules uniformly at random, uses initial energy $10$ and $E_{\max}=10^7$, and compares the three values of $d$. Mutation and explicit replication costs are disabled in both experiments.

\subsection{Soup Representation and Pairing Topology}
The environment consists of a population of programs represented as sequences of bytes representing Z80 machine code. Programs are initialized as 32-byte sequences of random instructions. The population is arranged in a grid or pool depending on the interaction topology. We evaluated two primary pairing mechanisms:
\begin{itemize}[leftmargin=*]
    \item \textbf{Well-Mixed}: Programs are paired uniformly at random from the global population for each interaction epoch, isolating frequency-dependent strategic stability.
    \item \textbf{Local (Spatial)}: Programs reside on a 2D grid. For each interaction, a program is paired with a neighbor selected at random from its immediate von Neumann neighborhood (up, down, left, right), allowing localized clustering of strategies.
\end{itemize}
When two programs are paired, they are concatenated to form a 64-byte shared execution space (tape).

\subsection{Z80 Substrate and Instruction Set Extension}
We used a modified Z80 emulator based on the open-source implementation available at \url{https://github.com/superzazu/z80}. To implement the social dilemmas and track energy, the Z80 architecture was extended with custom registers and opcodes. Specifically, registers were added to represent the agent's own energy and the partner's energy. New opcodes were introduced to allow programs to read these energy values.

For the main experiments in the paper, the only social operation available was \texttt{STEAL} (Opcode \texttt{ED 11}). The \texttt{SHARE} operation (Opcode \texttt{ED 10}) was \textbf{not} available in these core experiments and was introduced only in the specific setup used to establish self-interest and sociality (Appendix~\ref{appendix:SelfInterest}).
\begin{itemize}[leftmargin=*]
    \item \textbf{\texttt{STEAL}} (Opcode \texttt{ED 11}): Transfers energy from the opponent's pool to the executing agent's pool with an imperfect absorption efficiency $\alpha < 1$.
    \item \textbf{\texttt{SHARE}} (Opcode \texttt{ED 10}): Transfers energy perfectly ($\alpha = 1$) from the executing agent to the partner. This was used only in the validation experiments in Section~\ref{appendix:SelfInterest}.
\end{itemize}

\subsection{Execution Kinetics and Fault Tolerance}
The physical time $dt$ required to compute a single micro-operation was defined as inversely proportional to the total available energy in the system: $dt = 1 / (E_i + E_j)$. If a program executes an invalid instruction, the emulator handles it by executing default behaviors or advancing the program counter. If a pair runs out of energy ($E_{tot} \le 0$), computation halts immediately, leading to starvation. To prevent infinite loops (e.g., non-terminating replicator loops), epochs are strictly bounded by a maximum step count, after which execution is terminated.

\subsection{Experimental Protocols and Hyperparameters}
The simulation durations varied across the research questions to ensure convergence:
\begin{itemize}[leftmargin=*]
    \item \textbf{Non-math Experiments}: The main evolutionary runs were executed for \textbf{6 million epochs}. The subsequent ``petri dish'' invasion assays (seeding mutants into naive soups) were run for \textbf{1 million epochs}.
    \item \textbf{Math experiments}: The task-driven experiments were run for \textbf{10 million epochs} to account for the increased complexity of learning tasks alongside replication.
\end{itemize}
For all non-math experiments, the \texttt{background\_energy\_threshold} was set to \textbf{255} and the \texttt{background\_energy\_cap} was set to \textbf{255}. This ensures that agents were always eligible for background energy hand-outs up to the maximum capacity. All of our simulations used 256 CPU cores for every seed/hyperparameter combination. We summarize the base hyperparameters in Table~\ref{tab:hyperparams}.

\begin{table}[h]
\centering
\begin{tabular}{ll}
\hline
\textbf{Parameter} & \textbf{Value} \\
\hline
$\#$ of Seeds & 10 \\
Grid side length ($L_{pop}$) & $128$ \\
Program length & $32$ bytes \\
Shared tape length & $64$ bytes \\
Initial energy budget & $255$ \\
Energy cap ($E_{max}$) & $255$ \\
Background energy threshold & $255$ (Non-math) \\
Background energy regeneration ($\epsilon$) & $24$ \\
Mutation rarity & $\frac{1}{128}$ \\
Write cost & $1$ \\
Absorption efficiency ($\alpha$) & $0.8$ \\
Maximum Step count & $512$ \\
\hline
\end{tabular}
\caption{Base hyperparameters kept fixed across experiments unless specified otherwise.}
\label{tab:hyperparams}
\end{table}

\subsection{Background Energy Regeneration}
The background energy regeneration amount ($\epsilon$) determines the influx of energy at the start of an interaction. In our preliminary experiments, we explored values of $\{12, 16, 24, 32, 64\}$. This parameter is critical as it defines the viability condition $2\epsilon > 2L$ (where $L$ is the number of writes needed for replication). We found if $\epsilon$ is too low, agents starve before they can replicate. If $\epsilon$ is too high, sociality is no longer necessary since tapes receive all of the compute they need to replicate from the environment. We settled on $\epsilon = 24$ for all of our experiments, as in our substrate it offered the best balance of enabling replication to be viable without overindulging programs. In Figure~\ref{fig:uniformVspace}, we illustrate how energy is distributed across the 2D grid for both the uniform energy modes and asymmetric energy modes analyzed in Figure~\ref{fig:rq2_asymmetric_grid}.

\begin{figure}
    \centering
    \includegraphics[width=0.8\linewidth]{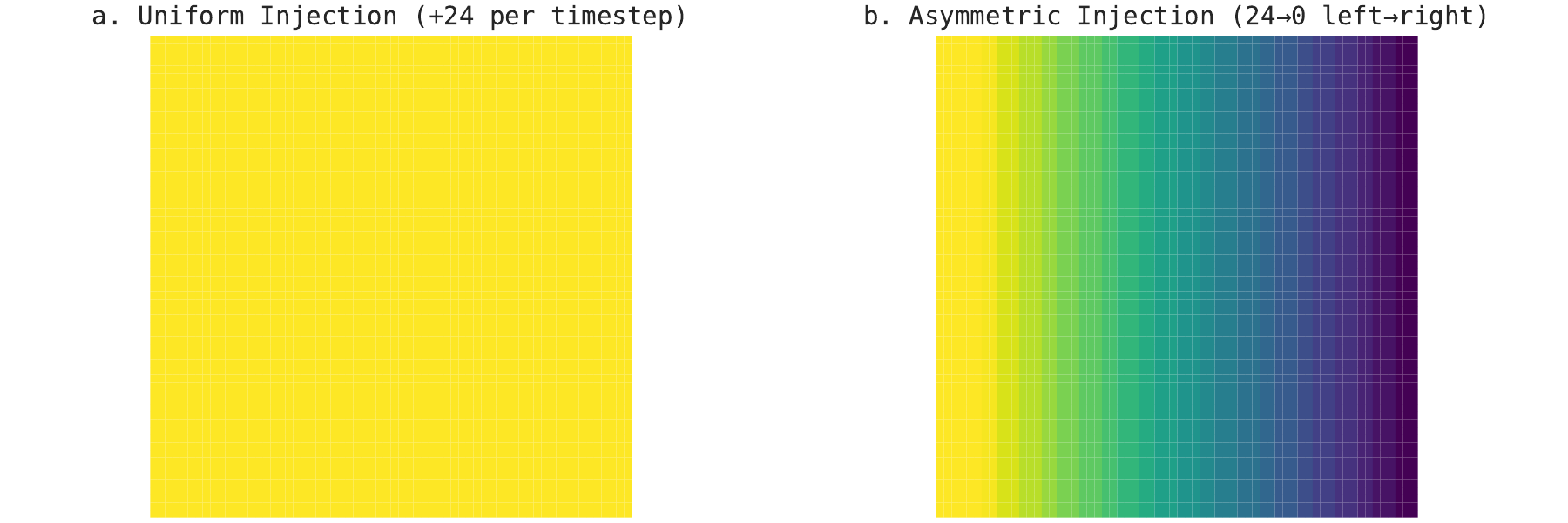}
    \caption{Comparison of uniformly distributed energy (a) vs asymmetrically distributed energy based on physical location (b).}
    \label{fig:uniformVspace}
\end{figure}

\subsection{Continuous Metrics}
\label{appendix:metrics}

We monitored population dynamics using several continuous metrics:
\begin{itemize}[leftmargin=*]
    \item \textbf{Higher-Order Entropy}: This metric measures architectural complexity by capturing structural motifs beyond simple single-instruction frequencies. Theoretically, it is defined as the difference between the sequence's Shannon entropy (computed over individual tokens) and its normalized Kolmogorov complexity (the Kolmogorov complexity divided by the sequence length $n$). This difference captures the information that can only be explained by relations between different characters, essentially ``factoring out'' information that comes from sampling independent and identically distributed (i.i.d.) variables. 
    The metric possesses two key properties: (1) for a sequence of i.i.d. characters, the expected higher-order entropy converges to 0 as length grows; and (2) for a sequence formed by concatenating many copies of the same shorter string (as occurs when a population is dominated by a self-replicator), the metric yields a substantial non-zero value. 
    Since Kolmogorov complexity is uncomputable, we approximate it in practice using the compressed size of the string achieved by a Lempel-Ziv-style compressor, following established practices in algorithmic information theory \cite{Thomas_M_Cover2006-vo, chenCompression, Horibe2003-ro}. This provides a metric that is both theoretically sound and fast to compute.
    \item \textbf{Edit Distance}: Measured the Hamming distance between program tapes to quantify genetic diversity.
    \item \textbf{Replication Frequency}: Tracked via the execution count of replication-specific instructions (e.g., \texttt{LDI}).
\end{itemize}

\subsection{Trace Embeddings via Multiply-Add-Permute (MAP) Vector Architectures}
To analyze and compare the behavior of programs, we computed fixed-size embeddings for variable-length execution logs. Let a program trace of length $T$ be defined as a sequence of opcodes $S = (o_1, o_2, \dots, o_T)$, where each opcode $o_t \in \{1, \dots, 65536\}$. The embedding process translates this discrete sequence into a normalized $128$-dimensional continuous vector representation through the following steps:

\textbf{1. Base Opcode Embedding.} 
We initialize a global embedding matrix where each possible 16-bit opcode $v$ is mapped to a static, randomly initialized dense vector $\mathbf{e}_v \in \mathbb{R}^{d}$ (with $d=128$), drawn from a standard normal distribution:
$$\mathbf{e}_v \sim \mathcal{N}(0, \mathbf{I}_{d})$$
For a given trace $S$, the base sequence of vectors is thus $(\mathbf{x}_1, \dots, \mathbf{x}_T)$, where $\mathbf{x}_t = \mathbf{e}_{o_t}$.

\textbf{2. Rotary Position Embeddings (RoPE).}
To preserve the sequential execution order without expanding the dimensionality of the vectors, we inject positional information using Rotary Position Embeddings (RoPE). We define a base frequency $B = 500.0$. For each position $t$, RoPE applies a block-diagonal rotation matrix $\mathbf{R}_{\Theta, t}$ to the base embedding $\mathbf{x}_t$. The rotated vector $\mathbf{r}_t = \mathbf{R}_{\Theta, t} \mathbf{x}_t$ is computed pairwise for dimensions $(2i-1, 2i)$ as:
$$\begin{pmatrix} r_{t, 2i-1} \\ r_{t, 2i} \end{pmatrix} = \begin{pmatrix} \cos(t \theta_i) & -\sin(t \theta_i) \\ \sin(t \theta_i) & \cos(t \theta_i) \end{pmatrix} \begin{pmatrix} x_{t, 2i-1} \\ x_{t, 2i} \end{pmatrix}$$
where the angle frequency for dimension pair $i \in \{1, \dots, d/2\}$ is $\theta_i = B^{-2(i-1)/d}$.

\textbf{3. N-Gram Binding via MAP Operations.}
To capture the contextual execution logic, we calculate the sum of 1-gram, 2-gram, and 3-gram trace combinations using vector binding. We utilize a Multiply-Add-Permute (MAP) operation from Vector Symbolic Architectures, defined by a 1-step circular shift $\rho(\cdot)$ followed by an Hadamard (element-wise) product $\odot$. The binding of two vectors is denoted as $\mathbf{u} \circledast \mathbf{v} = \mathbf{u} \odot \rho(\mathbf{v})$. 

We compute the aggregated $n$-gram representations for the entire trace as follows:
\begin{itemize}[leftmargin=*]
    \item \textbf{1-gram sum:} $\mathbf{H}_1 = \sum_{t=1}^{T} \mathbf{r}_t$
    \item \textbf{2-gram sum:} $\mathbf{H}_2 = \sum_{t=1}^{T-1} \mathbf{r}_t \circledast \mathbf{r}_{t+1}$
    \item \textbf{3-gram sum:} $\mathbf{H}_3 = \sum_{t=1}^{T-2} \mathbf{r}_t \circledast \mathbf{r}_{t+1} \circledast \mathbf{r}_{t+2}$
\end{itemize}
The total sequence representation is the unnormalized composite vector $\mathbf{H}_{sum} = \mathbf{H}_1 + \mathbf{H}_2 + \mathbf{H}_3$.

\textbf{4. $L_2$ Normalization.}
Finally, to ensure that the trace embeddings are length-invariant and suitable for downstream cosine similarity or clustering comparisons (such as $K$-Means), the composite vector is $L_2$-normalized to unit length:
$$\mathbf{H}_{final} = \frac{\mathbf{H}_{sum}}{\|\mathbf{H}_{sum}\|_2}$$
This yields the final $128$-dimensional representation of the program's execution trace.

\subsection{Fixation Computation}
Fixation rates in the petri dish assays were quantified using the MAP embeddings:
1. The final execution traces of all active agents in a population were downloaded.
2. MAP embeddings were computed for these traces in chunks to manage memory.
3. We applied the $K$-Means clustering algorithm with $K=2$ to the resulting embeddings.
4. Fixation was defined as the percentage of the active population belonging to the dominant cluster:
   \[ \text{Fixation \%} = \frac{\text{Dominant Cluster Size}}{\text{Total Active Programs}} \times 100 \]
This metric looks to capture the convergence of populations toward unified strategies.

\section{Traces of Program Behavior}

In this section, we show examples of different evolved behaviors in our Z80 substrate.

\subsection{Defection vs Cooperation}
\label{appendix:defectVCoopTrace}

In this example, we illustrate how having a unified phase between replication and social interactions enables cooperation to overcome defection. Here, a cooperative tape (Tape A), has a slightly higher energy budget than the defective one using the \texttt{STEAL} operation (byte sequence 237, 17; represented as \texttt{ed 11}). The defector tries to employ a strategy of ``loop the steal operation to drain my partner.'' However, because execution is probabilistic and proportional to energy, even after the cooperator's energy drops below the defector, it is able to overwrite the malicious steal operation, at first with code that forces the defector to jump to a \texttt{HALT} operation, and then with its functioning, replicator code.

\begin{Verbatim}[commandchars=\\\{\}, breaklines=true, breakanywhere=true, fontsize=\scriptsize, frame=single, rulecolor=\color{black!30}]
\textcolor{blue}{\textbf{=== EXECUTION LOGS ===}}
\textcolor{blue}{--- TRACE START (Swapped: NO, Lockstep: YES) ---}
\textcolor{cyan}{\textbf{STEP: 0 | CPU_B PC: 0x0000 (Addr: 0x0) | Slot: K}}
\textcolor{cyan}{ REGS: AF:4400 BC:40ff DE:3115 HL:3a10 IX:0000 IY:0000 SP:ffff}
\textcolor{cyan}{ FLGS: S:1 Z:1 Y:1 H:1 X:1 P:1 N:1 C:1} | Energy J:200 K:157
  TAPE J: ed 73 a0 f1 27 bf 23 d9 2f 11 e0 75 40 ed b0 76 76 0d 64 6d 81 4d 59 4a 76 4c 41 f7 ec 01 00 \textcolor{magenta}{\textbf{[ff]}}
  TAPE K: \textcolor{blue}{\textbf{[ed]}} 11 ed 11 20 fa 11 ec 2e cc 48 ed b8 b5 ed b0 ed b8 75 d6 07 39 11 9c d2 4c ed 3f 85 0c 00 4c
\textcolor{blue}{------------------------------------------------------------------}
\textcolor{purple}{\textbf{STEP: 1 | CPU_A PC: 0x0000 (Addr: 0x0) | Slot: J}}
\textcolor{purple}{ REGS: AF:6600 BC:b35c DE:0e6a HL:47bc IX:0000 IY:0000 SP:ffff}
\textcolor{purple}{ FLGS: S:1 Z:1 Y:1 H:1 X:1 P:1 N:1 C:1} | Energy J:184 K:168
  TAPE J: \textcolor{purple}{\textbf{[ed]}} 73 a0 f1 27 bf 23 d9 2f 11 e0 75 40 ed b0 76 76 0d 64 6d 81 4d 59 4a 76 4c 41 f7 ec 01 00 \textcolor{magenta}{\textbf{[ff]}}
  TAPE K: \textcolor{blue}{\textbf{[ed]}} 11 ed 11 20 fa 11 ec 2e cc 48 ed b8 b5 ed b0 ed b8 75 d6 07 39 11 9c d2 4c ed 3f 85 0c 00 \textcolor{red}{\textbf{[4c]}}
\textcolor{red}{\textbf{ >> WRITE: 0xf1a0 [Tape B + 0x0] <- \textcolor{yellow}{\textbf{0xff}} (SUCCESS)}}
\textcolor{red}{\textbf{ >> WRITE: 0xf1a1 [Tape B + 0x1] <- \textcolor{yellow}{\textbf{0xff}} (SUCCESS)}}
\textcolor{blue}{------------------------------------------------------------------}
\textcolor{purple}{\textbf{STEP: 2 | CPU_A PC: 0x0004 (Addr: 0x4) | Slot: J}}
\textcolor{purple}{ REGS: AF:6600 BC:b35c DE:0e6a HL:47bc IX:0000 IY:0000 SP:ffff}
\textcolor{purple}{ FLGS: S:1 Z:1 Y:1 H:1 X:1 P:1 N:1 C:1} | Energy J:183 K:168
  TAPE J: ed 73 a0 f1 \textcolor{purple}{\textbf{[27]}} bf 23 d9 2f 11 e0 75 40 ed b0 76 76 0d 64 6d 81 4d 59 4a 76 4c 41 f7 ec 01 00 \textcolor{magenta}{\textbf{[ff]}}
  TAPE K: \textcolor{blue}{\textbf{[ff]}} \textcolor{yellow}{\textbf{[ff]}} ed 11 20 fa 11 ec 2e cc 48 ed b8 b5 ed b0 ed b8 75 d6 07 39 11 9c d2 4c ed 3f 85 0c 00 \textcolor{red}{\textbf{[4c]}}
\textcolor{blue}{------------------------------------------------------------------}
\textcolor{purple}{\textbf{STEP: 3 | CPU_A PC: 0x0005 (Addr: 0x5) | Slot: J}}
\textcolor{purple}{ REGS: AF:0000 BC:b35c DE:0e6a HL:47bc IX:0000 IY:0000 SP:ffff}
\textcolor{purple}{ FLGS: S:0 Z:1 Y:0 H:0 X:0 P:1 N:1 C:1} | Energy J:182 K:168
  TAPE J: ed 73 a0 f1 27 \textcolor{purple}{\textbf{[bf]}} 23 d9 2f 11 e0 75 40 ed b0 76 76 0d 64 6d 81 4d 59 4a 76 4c 41 f7 ec 01 00 \textcolor{magenta}{\textbf{[ff]}}
  TAPE K: \textcolor{blue}{\textbf{[ff]}} ff ed 11 20 fa 11 ec 2e cc 48 ed b8 b5 ed b0 ed b8 75 d6 07 39 11 9c d2 4c ed 3f 85 0c 00 \textcolor{red}{\textbf{[4c]}}
\textcolor{blue}{------------------------------------------------------------------}
\textcolor{cyan}{\textbf{STEP: 4 | CPU_B PC: 0x0002 (Addr: 0x2) | Slot: K}}
\textcolor{cyan}{ REGS: AF:4400 BC:40ff DE:3115 HL:3a10 IX:0000 IY:0000 SP:ffff}
\textcolor{cyan}{ FLGS: S:1 Z:0 Y:1 H:1 X:1 P:1 N:0 C:1} | Energy J:181 K:168
  TAPE J: ed 73 a0 f1 27 \textcolor{purple}{\textbf{[bf]}} 23 d9 2f 11 e0 75 40 ed b0 76 76 0d 64 6d 81 4d 59 4a 76 4c 41 f7 ec 01 00 \textcolor{magenta}{\textbf{[ff]}}
  TAPE K: ff ff \textcolor{blue}{\textbf{[ed]}} 11 20 fa 11 ec 2e cc 48 ed b8 b5 ed b0 ed b8 75 d6 07 39 11 9c d2 4c ed 3f 85 0c 00 \textcolor{red}{\textbf{[4c]}}
\textcolor{blue}{------------------------------------------------------------------}
\textcolor{cyan}{\textbf{STEP: 5 | CPU_B PC: 0x0004 (Addr: 0x4) | Slot: K}}
\textcolor{cyan}{ REGS: AF:4400 BC:40ff DE:3115 HL:3a10 IX:0000 IY:0000 SP:ffff}
\textcolor{cyan}{ FLGS: S:1 Z:0 Y:1 H:1 X:1 P:1 N:0 C:1} | Energy J:165 K:179
  TAPE J: ed 73 a0 f1 27 \textcolor{purple}{\textbf{[bf]}} 23 d9 2f 11 e0 75 40 ed b0 76 76 0d 64 6d 81 4d 59 4a 76 4c 41 f7 ec 01 00 \textcolor{magenta}{\textbf{[ff]}}
  TAPE K: ff ff ed 11 \textcolor{blue}{\textbf{[20]}} fa 11 ec 2e cc 48 ed b8 b5 ed b0 ed b8 75 d6 07 39 11 9c d2 4c ed 3f 85 0c 00 \textcolor{red}{\textbf{[4c]}}
\textcolor{blue}{------------------------------------------------------------------}
\textcolor{purple}{\textbf{STEP: 6 | CPU_A PC: 0x0006 (Addr: 0x6) | Slot: J}}
\textcolor{purple}{ REGS: AF:0000 BC:b35c DE:0e6a HL:47bc IX:0000 IY:0000 SP:ffff}
\textcolor{purple}{ FLGS: S:0 Z:1 Y:0 H:0 X:0 P:0 N:1 C:0} | Energy J:165 K:178
  TAPE J: ed 73 a0 f1 27 bf \textcolor{purple}{\textbf{[23]}} d9 2f 11 e0 75 40 ed b0 76 76 0d 64 6d 81 4d 59 4a 76 4c 41 f7 ec 01 00 \textcolor{magenta}{\textbf{[ff]}}
  TAPE K: ff ff ed 11 \textcolor{blue}{\textbf{[20]}} fa 11 ec 2e cc 48 ed b8 b5 ed b0 ed b8 75 d6 07 39 11 9c d2 4c ed 3f 85 0c 00 \textcolor{red}{\textbf{[4c]}}
\textcolor{blue}{------------------------------------------------------------------}
\textcolor{purple}{\textbf{STEP: 7 | CPU_A PC: 0x0007 (Addr: 0x7) | Slot: J}}
\textcolor{purple}{ REGS: AF:0000 BC:b35c DE:0e6a HL:47bd IX:0000 IY:0000 SP:ffff}
\textcolor{purple}{ FLGS: S:0 Z:1 Y:0 H:0 X:0 P:0 N:1 C:0} | Energy J:164 K:178
  TAPE J: ed 73 a0 f1 27 bf 23 \textcolor{purple}{\textbf{[d9]}} 2f 11 e0 75 40 ed b0 76 76 0d 64 6d 81 4d 59 4a 76 4c 41 f7 ec 01 00 \textcolor{magenta}{\textbf{[ff]}}
  TAPE K: ff ff ed 11 \textcolor{blue}{\textbf{[20]}} fa 11 ec 2e cc 48 ed b8 b5 ed b0 ed b8 75 d6 07 39 11 9c d2 4c ed 3f 85 0c 00 \textcolor{red}{\textbf{[4c]}}
\textcolor{blue}{------------------------------------------------------------------}
\textcolor{purple}{\textbf{STEP: 8 | CPU_A PC: 0x0008 (Addr: 0x8) | Slot: J}}
\textcolor{purple}{ REGS: AF:0000 BC:0000 DE:0000 HL:0000 IX:0000 IY:0000 SP:ffff}
\textcolor{purple}{ FLGS: S:0 Z:1 Y:0 H:0 X:0 P:0 N:1 C:0} | Energy J:163 K:178
  TAPE J: ed 73 a0 f1 27 bf 23 d9 \textcolor{purple}{\textbf{[2f]}} 11 e0 75 40 ed b0 76 76 0d 64 6d 81 4d 59 4a 76 4c 41 f7 ec 01 00 \textcolor{magenta}{\textbf{[ff]}}
  TAPE K: ff ff ed 11 \textcolor{blue}{\textbf{[20]}} fa 11 ec 2e cc 48 ed b8 b5 ed b0 ed b8 75 d6 07 39 11 9c d2 4c ed 3f 85 0c 00 \textcolor{red}{\textbf{[4c]}}
\textcolor{blue}{------------------------------------------------------------------}
\textcolor{purple}{\textbf{STEP: 9 | CPU_A PC: 0x0009 (Addr: 0x9) | Slot: J}}
\textcolor{purple}{ REGS: AF:ff00 BC:0000 DE:0000 HL:0000 IX:0000 IY:0000 SP:ffff}
\textcolor{purple}{ FLGS: S:0 Z:1 Y:1 H:1 X:1 P:0 N:1 C:0} | Energy J:162 K:178
  TAPE J: ed 73 a0 f1 27 bf 23 d9 2f \textcolor{purple}{\textbf{[11]}} e0 75 40 ed b0 76 76 0d 64 6d 81 4d 59 4a 76 4c 41 f7 ec 01 00 \textcolor{magenta}{\textbf{[ff]}}
  TAPE K: ff ff ed 11 \textcolor{blue}{\textbf{[20]}} fa 11 ec 2e cc 48 ed b8 b5 ed b0 ed b8 75 d6 07 39 11 9c d2 4c ed 3f 85 0c 00 \textcolor{red}{\textbf{[4c]}}
\textcolor{blue}{------------------------------------------------------------------}
\textcolor{cyan}{\textbf{STEP: 10 | CPU_B PC: 0x0000 (Addr: 0x0) | Slot: K}}
\textcolor{cyan}{ REGS: AF:4400 BC:40ff DE:3115 HL:3a10 IX:0000 IY:0000 SP:ffff}
\textcolor{cyan}{ FLGS: S:1 Z:0 Y:1 H:1 X:1 P:1 N:0 C:1} | Energy J:161 K:178
  TAPE J: ed 73 a0 f1 27 bf 23 d9 2f \textcolor{purple}{\textbf{[11]}} e0 75 40 ed b0 76 76 0d 64 6d 81 4d 59 4a 76 4c 41 f7 ec 01 00 \textcolor{magenta}{\textbf{[ff]}}
  TAPE K: \textcolor{blue}{\textbf{[ff]}} ff ed 11 20 fa 11 ec 2e cc 48 ed b8 b5 ed b0 ed b8 75 d6 07 39 11 9c d2 4c ed 3f 85 0c 00 \textcolor{red}{\textbf{[4c]}}
\textcolor{red}{\textbf{ >> WRITE: 0xfffd [Tape A + 0x1d] <- \textcolor{yellow}{\textbf{0x01}} (SUCCESS)}}
\textcolor{red}{\textbf{ >> WRITE: 0xfffe [Tape A + 0x1e] <- \textcolor{yellow}{\textbf{0x00}} (SUCCESS)}}
\textcolor{blue}{------------------------------------------------------------------}
\textcolor{purple}{\textbf{STEP: 11 | CPU_A PC: 0x000c (Addr: 0xc) | Slot: J}}
\textcolor{purple}{ REGS: AF:ff00 BC:0000 DE:75e0 HL:0000 IX:0000 IY:0000 SP:ffff}
\textcolor{purple}{ FLGS: S:0 Z:1 Y:1 H:1 X:1 P:0 N:1 C:0} | Energy J:161 K:177
  TAPE J: ed 73 a0 f1 27 bf 23 d9 2f 11 e0 75 \textcolor{purple}{\textbf{[40]}} ed b0 76 76 0d 64 6d 81 4d 59 4a 76 4c 41 f7 ec 01 00 \textcolor{magenta}{\textbf{[ff]}}
  TAPE K: \textcolor{blue}{\textbf{[ff]}} ff ed 11 20 fa 11 ec 2e cc 48 ed b8 b5 ed b0 ed b8 75 d6 07 39 11 9c d2 4c ed 3f 85 0c \textcolor{yellow}{\textbf{[00]}} \textcolor{red}{\textbf{[4c]}}
\textcolor{blue}{------------------------------------------------------------------}
\textcolor{purple}{\textbf{STEP: 12 | CPU_A PC: 0x000d (Addr: 0xd) | Slot: J}}
\textcolor{purple}{ REGS: AF:0000 BC:0000 DE:75e0 HL:0000 IX:0000 IY:0000 SP:ffff}
\textcolor{purple}{ FLGS: S:0 Z:1 Y:1 H:1 X:1 P:0 N:1 C:0} | Energy J:160 K:177
  TAPE J: ed 73 a0 f1 27 bf 23 d9 2f 11 e0 75 40 \textcolor{purple}{\textbf{[ed]}} b0 76 76 0d 64 6d 81 4d 59 4a 76 4c 41 f7 ec 01 00 \textcolor{magenta}{\textbf{[ff]}}
  TAPE K: \textcolor{blue}{\textbf{[ff]}} ff ed 11 20 fa 11 ec 2e cc 48 ed b8 b5 ed b0 ed b8 75 d6 07 39 11 9c d2 4c ed 3f 85 0c 00 \textcolor{red}{\textbf{[4c]}}
\textcolor{red}{\textbf{ >> WRITE: 0x75e0 [Tape B + 0x0] <- \textcolor{yellow}{\textbf{0xed}} (SUCCESS)}}
\textcolor{blue}{------------------------------------------------------------------}
\textcolor{cyan}{\textbf{STEP: 13 | CPU_B PC: 0x0038 (Addr: 0x38) | Slot: J}}
\textcolor{cyan}{ REGS: AF:4400 BC:40ff DE:3115 HL:3a10 IX:0000 IY:0000 SP:fffd}
\textcolor{cyan}{ FLGS: S:1 Z:0 Y:1 H:1 X:1 P:1 N:0 C:1} | Energy J:159 K:177
  TAPE J: ed 73 a0 f1 27 bf 23 d9 2f 11 e0 75 40 \textcolor{purple}{\textbf{[ed]}} b0 76 76 0d 64 6d 81 4d 59 4a \textcolor{blue}{\textbf{[76]}} 4c 41 f7 ec \textcolor{magenta}{\textbf{[01]}} 00 ff
  TAPE K: \textcolor{yellow}{\textbf{[ed]}} ff ed 11 20 fa 11 ec 2e cc 48 ed b8 b5 ed b0 ed b8 75 d6 07 39 11 9c [d2] 4c ed 3f 85 0c 00 \textcolor{red}{\textbf{[4c]}}
\textcolor{blue}{------------------------------------------------------------------}
\textcolor{purple}{\textbf{STEP: 14 | CPU_A PC: 0x000d (Addr: 0xd) | Slot: J}}
\textcolor{purple}{ REGS: AF:0000 BC:ffff DE:75e1 HL:0001 IX:0000 IY:0000 SP:ffff}
\textcolor{purple}{ FLGS: S:0 Z:1 Y:0 H:0 X:1 P:1 N:0 C:0} | Energy J:158 K:177
  TAPE J: ed 73 a0 f1 27 bf 23 d9 2f 11 e0 75 40 \textcolor{purple}{\textbf{[ed]}} b0 76 76 0d 64 6d 81 4d 59 4a \textcolor{blue}{\textbf{[76]}} 4c 41 f7 ec \textcolor{magenta}{\textbf{[01]}} 00 ff
  TAPE K: ed ff ed 11 20 fa 11 ec 2e cc 48 ed b8 b5 ed b0 ed b8 75 d6 07 39 11 9c d2 4c ed 3f 85 0c 00 \textcolor{red}{\textbf{[4c]}}
\textcolor{red}{\textbf{ >> WRITE: 0x75e1 [Tape B + 0x1] <- \textcolor{yellow}{\textbf{0x73}} (SUCCESS)}}
\textcolor{blue}{------------------------------------------------------------------}
\textcolor{purple}{\textbf{STEP: 15 | CPU_A PC: 0x000d (Addr: 0xd) | Slot: J}}
\textcolor{purple}{ REGS: AF:0000 BC:fffe DE:75e2 HL:0002 IX:0000 IY:0000 SP:ffff}
\textcolor{purple}{ FLGS: S:0 Z:1 Y:1 H:0 X:0 P:1 N:0 C:0} | Energy J:157 K:177
  TAPE J: ed 73 a0 f1 27 bf 23 d9 2f 11 e0 75 40 \textcolor{purple}{\textbf{[ed]}} b0 76 76 0d 64 6d 81 4d 59 4a \textcolor{blue}{\textbf{[76]}} 4c 41 f7 ec \textcolor{magenta}{\textbf{[01]}} 00 ff
  TAPE K: ed \textcolor{yellow}{\textbf{[73]}} ed 11 20 fa 11 ec 2e cc 48 ed b8 b5 ed b0 ed b8 75 d6 07 39 11 9c d2 4c ed 3f 85 0c 00 \textcolor{red}{\textbf{[4c]}}
\textcolor{red}{\textbf{ >> WRITE: 0x75e2 [Tape B + 0x2] <- \textcolor{yellow}{\textbf{0xa0}} (SUCCESS)}}
\textcolor{blue}{------------------------------------------------------------------}
\textcolor{purple}{\textbf{STEP: 16 | CPU_A PC: 0x000d (Addr: 0xd) | Slot: J}}
\textcolor{purple}{ REGS: AF:0000 BC:fffd DE:75e3 HL:0003 IX:0000 IY:0000 SP:ffff}
\textcolor{purple}{ FLGS: S:0 Z:1 Y:0 H:0 X:0 P:1 N:0 C:0} | Energy J:156 K:177
  TAPE J: ed 73 a0 f1 27 bf 23 d9 2f 11 e0 75 40 \textcolor{purple}{\textbf{[ed]}} b0 76 76 0d 64 6d 81 4d 59 4a \textcolor{blue}{\textbf{[76]}} 4c 41 f7 ec \textcolor{magenta}{\textbf{[01]}} 00 ff
  TAPE K: ed 73 \textcolor{yellow}{\textbf{[a0]}} 11 20 fa 11 ec 2e cc 48 ed b8 b5 ed b0 ed b8 75 d6 07 39 11 9c d2 4c ed 3f 85 0c 00 \textcolor{red}{\textbf{[4c]}}
\textcolor{red}{\textbf{ >> WRITE: 0x75e3 [Tape B + 0x3] <- \textcolor{yellow}{\textbf{0xf1}} (SUCCESS)}}
\textcolor{blue}{------------------------------------------------------------------}
\textcolor{purple}{\textbf{STEP: 17 | CPU_A PC: 0x000d (Addr: 0xd) | Slot: J}}
\textcolor{purple}{ REGS: AF:0000 BC:fffc DE:75e4 HL:0004 IX:0000 IY:0000 SP:ffff}
\textcolor{purple}{ FLGS: S:0 Z:1 Y:0 H:0 X:0 P:1 N:0 C:0} | Energy J:155 K:177
  TAPE J: ed 73 a0 f1 27 bf 23 d9 2f 11 e0 75 40 \textcolor{purple}{\textbf{[ed]}} b0 76 76 0d 64 6d 81 4d 59 4a \textcolor{blue}{\textbf{[76]}} 4c 41 f7 ec \textcolor{magenta}{\textbf{[01]}} 00 ff
  TAPE K: ed 73 a0 \textcolor{yellow}{\textbf{[f1]}} 20 fa 11 ec 2e cc 48 ed b8 b5 ed b0 ed b8 75 d6 07 39 11 9c d2 4c ed 3f 85 0c 00 \textcolor{red}{\textbf{[4c]}}
\textcolor{red}{\textbf{ >> WRITE: 0x75e4 [Tape B + 0x4] <- \textcolor{yellow}{\textbf{0x27}} (SUCCESS)}}
\textcolor{blue}{------------------------------------------------------------------}
\textcolor{purple}{\textbf{STEP: 18 | CPU_A PC: 0x000d (Addr: 0xd) | Slot: J}}
\textcolor{purple}{ REGS: AF:0000 BC:fffb DE:75e5 HL:0005 IX:0000 IY:0000 SP:ffff}
\textcolor{purple}{ FLGS: S:0 Z:1 Y:1 H:0 X:0 P:1 N:0 C:0} | Energy J:154 K:177
  TAPE J: ed 73 a0 f1 27 bf 23 d9 2f 11 e0 75 40 \textcolor{purple}{\textbf{[ed]}} b0 76 76 0d 64 6d 81 4d 59 4a \textcolor{blue}{\textbf{[76]}} 4c 41 f7 ec \textcolor{magenta}{\textbf{[01]}} 00 ff
  TAPE K: ed 73 a0 f1 \textcolor{yellow}{\textbf{[27]}} fa 11 ec 2e cc 48 ed b8 b5 ed b0 ed b8 75 d6 07 39 11 9c d2 4c ed 3f 85 0c 00 \textcolor{red}{\textbf{[4c]}}
\textcolor{red}{\textbf{ >> WRITE: 0x75e5 [Tape B + 0x5] <- \textcolor{yellow}{\textbf{0xbf}} (SUCCESS)}}
\textcolor{blue}{------------------------------------------------------------------}
\textcolor{purple}{\textbf{STEP: 19 | CPU_A PC: 0x000d (Addr: 0xd) | Slot: J}}
\textcolor{purple}{ REGS: AF:0000 BC:fffa DE:75e6 HL:0006 IX:0000 IY:0000 SP:ffff}
\textcolor{purple}{ FLGS: S:0 Z:1 Y:1 H:0 X:1 P:1 N:0 C:0} | Energy J:153 K:177
  TAPE J: ed 73 a0 f1 27 bf 23 d9 2f 11 e0 75 40 \textcolor{purple}{\textbf{[ed]}} b0 76 76 0d 64 6d 81 4d 59 4a \textcolor{blue}{\textbf{[76]}} 4c 41 f7 ec \textcolor{magenta}{\textbf{[01]}} 00 ff
  TAPE K: ed 73 a0 f1 27 \textcolor{yellow}{\textbf{[bf]}} 11 ec 2e cc 48 ed b8 b5 ed b0 ed b8 75 d6 07 39 11 9c d2 4c ed 3f 85 0c 00 \textcolor{red}{\textbf{[4c]}}
\textcolor{red}{\textbf{ >> WRITE: 0x75e6 [Tape B + 0x6] <- \textcolor{yellow}{\textbf{0x23}} (SUCCESS)}}
\textcolor{blue}{------------------------------------------------------------------}
\textcolor{purple}{\textbf{STEP: 20 | CPU_A PC: 0x000d (Addr: 0xd) | Slot: J}}
\textcolor{purple}{ REGS: AF:0000 BC:fff9 DE:75e7 HL:0007 IX:0000 IY:0000 SP:ffff}
\textcolor{purple}{ FLGS: S:0 Z:1 Y:1 H:0 X:0 P:1 N:0 C:0} | Energy J:152 K:177
  TAPE J: ed 73 a0 f1 27 bf 23 d9 2f 11 e0 75 40 \textcolor{purple}{\textbf{[ed]}} b0 76 76 0d 64 6d 81 4d 59 4a \textcolor{blue}{\textbf{[76]}} 4c 41 f7 ec \textcolor{magenta}{\textbf{[01]}} 00 ff
  TAPE K: ed 73 a0 f1 27 bf \textcolor{yellow}{\textbf{[23]}} ec 2e cc 48 ed b8 b5 ed b0 ed b8 75 d6 07 39 11 9c d2 4c ed 3f 85 0c 00 \textcolor{red}{\textbf{[4c]}}
\textcolor{red}{\textbf{ >> WRITE: 0x75e7 [Tape B + 0x7] <- \textcolor{yellow}{\textbf{0xd9}} (SUCCESS)}}
\textcolor{blue}{------------------------------------------------------------------}
\textcolor{purple}{\textbf{STEP: 21 | CPU_A PC: 0x000d (Addr: 0xd) | Slot: J}}
\textcolor{purple}{ REGS: AF:0000 BC:fff8 DE:75e8 HL:0008 IX:0000 IY:0000 SP:ffff}
\textcolor{purple}{ FLGS: S:0 Z:1 Y:0 H:0 X:1 P:1 N:0 C:0} | Energy J:151 K:177
  TAPE J: ed 73 a0 f1 27 bf 23 d9 2f 11 e0 75 40 \textcolor{purple}{\textbf{[ed]}} b0 76 76 0d 64 6d 81 4d 59 4a \textcolor{blue}{\textbf{[76]}} 4c 41 f7 ec \textcolor{magenta}{\textbf{[01]}} 00 ff
  TAPE K: ed 73 a0 f1 27 bf 23 \textcolor{yellow}{\textbf{[d9]}} 2e cc 48 ed b8 b5 ed b0 ed b8 75 d6 07 39 11 9c d2 4c ed 3f 85 0c 00 \textcolor{red}{\textbf{[4c]}}
\textcolor{red}{\textbf{ >> WRITE: 0x75e8 [Tape B + 0x8] <- \textcolor{yellow}{\textbf{0x2f}} (SUCCESS)}}
\textcolor{blue}{------------------------------------------------------------------}
\textcolor{purple}{\textbf{STEP: 22 | CPU_A PC: 0x000d (Addr: 0xd) | Slot: J}}
\textcolor{purple}{ REGS: AF:0000 BC:fff7 DE:75e9 HL:0009 IX:0000 IY:0000 SP:ffff}
\textcolor{purple}{ FLGS: S:0 Z:1 Y:1 H:0 X:1 P:1 N:0 C:0} | Energy J:150 K:177
  TAPE J: ed 73 a0 f1 27 bf 23 d9 2f 11 e0 75 40 \textcolor{purple}{\textbf{[ed]}} b0 76 76 0d 64 6d 81 4d 59 4a \textcolor{blue}{\textbf{[76]}} 4c 41 f7 ec \textcolor{magenta}{\textbf{[01]}} 00 ff
  TAPE K: ed 73 a0 f1 27 bf 23 d9 \textcolor{yellow}{\textbf{[2f]}} cc 48 ed b8 b5 ed b0 ed b8 75 d6 07 39 11 9c d2 4c ed 3f 85 0c 00 \textcolor{red}{\textbf{[4c]}}
\textcolor{red}{\textbf{ >> WRITE: 0x75e9 [Tape B + 0x9] <- \textcolor{yellow}{\textbf{0x11}} (SUCCESS)}}
\textcolor{blue}{------------------------------------------------------------------}
\textcolor{purple}{\textbf{STEP: 23 | CPU_A PC: 0x000d (Addr: 0xd) | Slot: J}}
\textcolor{purple}{ REGS: AF:0000 BC:fff6 DE:75ea HL:000a IX:0000 IY:0000 SP:ffff}
\textcolor{purple}{ FLGS: S:0 Z:1 Y:0 H:0 X:0 P:1 N:0 C:0} | Energy J:149 K:177
  TAPE J: ed 73 a0 f1 27 bf 23 d9 2f 11 e0 75 40 \textcolor{purple}{\textbf{[ed]}} b0 76 76 0d 64 6d 81 4d 59 4a \textcolor{blue}{\textbf{[76]}} 4c 41 f7 ec \textcolor{magenta}{\textbf{[01]}} 00 ff
  TAPE K: ed 73 a0 f1 27 bf 23 d9 2f \textcolor{yellow}{\textbf{[11]}} 48 ed b8 b5 ed b0 ed b8 75 d6 07 39 11 9c d2 4c ed 3f 85 0c 00 \textcolor{red}{\textbf{[4c]}}
\textcolor{red}{\textbf{ >> WRITE: 0x75ea [Tape B + 0xa] <- \textcolor{yellow}{\textbf{0xe0}} (SUCCESS)}}
\textcolor{blue}{------------------------------------------------------------------}
\textcolor{purple}{\textbf{STEP: 24 | CPU_A PC: 0x000d (Addr: 0xd) | Slot: J}}
\textcolor{purple}{ REGS: AF:0000 BC:fff5 DE:75eb HL:000b IX:0000 IY:0000 SP:ffff}
\textcolor{purple}{ FLGS: S:0 Z:1 Y:0 H:0 X:0 P:1 N:0 C:0} | Energy J:148 K:177
  TAPE J: ed 73 a0 f1 27 bf 23 d9 2f 11 e0 75 40 \textcolor{purple}{\textbf{[ed]}} b0 76 76 0d 64 6d 81 4d 59 4a \textcolor{blue}{\textbf{[76]}} 4c 41 f7 ec \textcolor{magenta}{\textbf{[01]}} 00 ff
  TAPE K: ed 73 a0 f1 27 bf 23 d9 2f 11 \textcolor{yellow}{\textbf{[e0]}} ed b8 b5 ed b0 ed b8 75 d6 07 39 11 9c d2 4c ed 3f 85 0c 00 \textcolor{red}{\textbf{[4c]}}
\textcolor{red}{\textbf{ >> WRITE: 0x75eb [Tape B + 0xb] <- \textcolor{yellow}{\textbf{0x75}} (SUCCESS)}}
\textcolor{blue}{------------------------------------------------------------------}
\textcolor{purple}{\textbf{STEP: 25 | CPU_A PC: 0x000d (Addr: 0xd) | Slot: J}}
\textcolor{purple}{ REGS: AF:0000 BC:fff4 DE:75ec HL:000c IX:0000 IY:0000 SP:ffff}
\textcolor{purple}{ FLGS: S:0 Z:1 Y:0 H:0 X:0 P:1 N:0 C:0} | Energy J:147 K:177
  TAPE J: ed 73 a0 f1 27 bf 23 d9 2f 11 e0 75 40 \textcolor{purple}{\textbf{[ed]}} b0 76 76 0d 64 6d 81 4d 59 4a \textcolor{blue}{\textbf{[76]}} 4c 41 f7 ec \textcolor{magenta}{\textbf{[01]}} 00 ff
  TAPE K: ed 73 a0 f1 27 bf 23 d9 2f 11 e0 \textcolor{yellow}{\textbf{[75]}} b8 b5 ed b0 ed b8 75 d6 07 39 11 9c d2 4c ed 3f 85 0c 00 \textcolor{red}{\textbf{[4c]}}
\textcolor{red}{\textbf{ >> WRITE: 0x75ec [Tape B + 0xc] <- \textcolor{yellow}{\textbf{0x40}} (SUCCESS)}}
\textcolor{blue}{------------------------------------------------------------------}
\textcolor{purple}{\textbf{STEP: 26 | CPU_A PC: 0x000d (Addr: 0xd) | Slot: J}}
\textcolor{purple}{ REGS: AF:0000 BC:fff3 DE:75ed HL:000d IX:0000 IY:0000 SP:ffff}
\textcolor{purple}{ FLGS: S:0 Z:1 Y:0 H:0 X:0 P:1 N:0 C:0} | Energy J:146 K:177
  TAPE J: ed 73 a0 f1 27 bf 23 d9 2f 11 e0 75 40 \textcolor{purple}{\textbf{[ed]}} b0 76 76 0d 64 6d 81 4d 59 4a \textcolor{blue}{\textbf{[76]}} 4c 41 f7 ec \textcolor{magenta}{\textbf{[01]}} 00 ff
  TAPE K: ed 73 a0 f1 27 bf 23 d9 2f 11 e0 75 \textcolor{yellow}{\textbf{[40]}} b5 ed b0 ed b8 75 d6 07 39 11 9c d2 4c ed 3f 85 0c 00 \textcolor{red}{\textbf{[4c]}}
\textcolor{red}{\textbf{ >> WRITE: 0x75ed [Tape B + 0xd] <- \textcolor{yellow}{\textbf{0xed}} (SUCCESS)}}
\textcolor{blue}{------------------------------------------------------------------}
\textcolor{purple}{\textbf{STEP: 27 | CPU_A PC: 0x000d (Addr: 0xd) | Slot: J}}
\textcolor{purple}{ REGS: AF:0000 BC:fff2 DE:75ee HL:000e IX:0000 IY:0000 SP:ffff}
\textcolor{purple}{ FLGS: S:0 Z:1 Y:0 H:0 X:1 P:1 N:0 C:0} | Energy J:145 K:177
  TAPE J: ed 73 a0 f1 27 bf 23 d9 2f 11 e0 75 40 \textcolor{purple}{\textbf{[ed]}} b0 76 76 0d 64 6d 81 4d 59 4a \textcolor{blue}{\textbf{[76]}} 4c 41 f7 ec \textcolor{magenta}{\textbf{[01]}} 00 ff
  TAPE K: ed 73 a0 f1 27 bf 23 d9 2f 11 e0 75 40 \textcolor{yellow}{\textbf{[ed]}} ed b0 ed b8 75 d6 07 39 11 9c d2 4c ed 3f 85 0c 00 \textcolor{red}{\textbf{[4c]}}
\textcolor{red}{\textbf{ >> WRITE: 0x75ee [Tape B + 0xe] <- \textcolor{yellow}{\textbf{0xb0}} (SUCCESS)}}
\textcolor{blue}{------------------------------------------------------------------}
\textcolor{purple}{\textbf{STEP: 28 | CPU_A PC: 0x000d (Addr: 0xd) | Slot: J}}
\textcolor{purple}{ REGS: AF:0000 BC:fff1 DE:75ef HL:000f IX:0000 IY:0000 SP:ffff}
\textcolor{purple}{ FLGS: S:0 Z:1 Y:0 H:0 X:0 P:1 N:0 C:0} | Energy J:144 K:177
  TAPE J: ed 73 a0 f1 27 bf 23 d9 2f 11 e0 75 40 \textcolor{purple}{\textbf{[ed]}} b0 76 76 0d 64 6d 81 4d 59 4a \textcolor{blue}{\textbf{[76]}} 4c 41 f7 ec \textcolor{magenta}{\textbf{[01]}} 00 ff
  TAPE K: ed 73 a0 f1 27 bf 23 d9 2f 11 e0 75 40 ed \textcolor{yellow}{\textbf{[b0]}} b0 ed b8 75 d6 07 39 11 9c d2 4c ed 3f 85 0c 00 \textcolor{red}{\textbf{[4c]}}
\textcolor{red}{\textbf{ >> WRITE: 0x75ef [Tape B + 0xf] <- \textcolor{yellow}{\textbf{0x76}} (SUCCESS)}}
\textcolor{blue}{------------------------------------------------------------------}
\textcolor{purple}{\textbf{STEP: 29 | CPU_A PC: 0x000d (Addr: 0xd) | Slot: J}}
\textcolor{purple}{ REGS: AF:0000 BC:fff0 DE:75f0 HL:0010 IX:0000 IY:0000 SP:ffff}
\textcolor{purple}{ FLGS: S:0 Z:1 Y:1 H:0 X:0 P:1 N:0 C:0} | Energy J:143 K:177
  TAPE J: ed 73 a0 f1 27 bf 23 d9 2f 11 e0 75 40 \textcolor{purple}{\textbf{[ed]}} b0 76 76 0d 64 6d 81 4d 59 4a \textcolor{blue}{\textbf{[76]}} 4c 41 f7 ec \textcolor{magenta}{\textbf{[01]}} 00 ff
  TAPE K: ed 73 a0 f1 27 bf 23 d9 2f 11 e0 75 40 ed b0 \textcolor{yellow}{\textbf{[76]}} ed b8 75 d6 07 39 11 9c d2 4c ed 3f 85 0c 00 \textcolor{red}{\textbf{[4c]}}
\textcolor{red}{\textbf{ >> WRITE: 0x75f0 [Tape B + 0x10] <- \textcolor{yellow}{\textbf{0x76}} (SUCCESS)}}
\textcolor{blue}{------------------------------------------------------------------}
\textcolor{purple}{\textbf{STEP: 30 | CPU_A PC: 0x000d (Addr: 0xd) | Slot: J}}
\textcolor{purple}{ REGS: AF:0000 BC:ffef DE:75f1 HL:0011 IX:0000 IY:0000 SP:ffff}
\textcolor{purple}{ FLGS: S:0 Z:1 Y:1 H:0 X:0 P:1 N:0 C:0} | Energy J:142 K:177
  TAPE J: ed 73 a0 f1 27 bf 23 d9 2f 11 e0 75 40 \textcolor{purple}{\textbf{[ed]}} b0 76 76 0d 64 6d 81 4d 59 4a \textcolor{blue}{\textbf{[76]}} 4c 41 f7 ec \textcolor{magenta}{\textbf{[01]}} 00 ff
  TAPE K: ed 73 a0 f1 27 bf 23 d9 2f 11 e0 75 40 ed b0 76 \textcolor{yellow}{\textbf{[76]}} b8 75 d6 07 39 11 9c d2 4c ed 3f 85 0c 00 \textcolor{red}{\textbf{[4c]}}
\textcolor{red}{\textbf{ >> WRITE: 0x75f1 [Tape B + 0x11] <- \textcolor{yellow}{\textbf{0x0d}} (SUCCESS)}}
\textcolor{blue}{------------------------------------------------------------------}
\textcolor{purple}{\textbf{STEP: 31 | CPU_A PC: 0x000d (Addr: 0xd) | Slot: J}}
\textcolor{purple}{ REGS: AF:0000 BC:ffee DE:75f2 HL:0012 IX:0000 IY:0000 SP:ffff}
\textcolor{purple}{ FLGS: S:0 Z:1 Y:0 H:0 X:1 P:1 N:0 C:0} | Energy J:141 K:177
  TAPE J: ed 73 a0 f1 27 bf 23 d9 2f 11 e0 75 40 \textcolor{purple}{\textbf{[ed]}} b0 76 76 0d 64 6d 81 4d 59 4a \textcolor{blue}{\textbf{[76]}} 4c 41 f7 ec \textcolor{magenta}{\textbf{[01]}} 00 ff
  TAPE K: ed 73 a0 f1 27 bf 23 d9 2f 11 e0 75 40 ed b0 76 76 \textcolor{yellow}{\textbf{[0d]}} 75 d6 07 39 11 9c d2 4c ed 3f 85 0c 00 \textcolor{red}{\textbf{[4c]}}
\textcolor{red}{\textbf{ >> WRITE: 0x75f2 [Tape B + 0x12] <- \textcolor{yellow}{\textbf{0x64}} (SUCCESS)}}
\textcolor{blue}{------------------------------------------------------------------}
\textcolor{purple}{\textbf{STEP: 32 | CPU_A PC: 0x000d (Addr: 0xd) | Slot: J}}
\textcolor{purple}{ REGS: AF:0000 BC:ffed DE:75f3 HL:0013 IX:0000 IY:0000 SP:ffff}
\textcolor{purple}{ FLGS: S:0 Z:1 Y:0 H:0 X:0 P:1 N:0 C:0} | Energy J:140 K:177
  TAPE J: ed 73 a0 f1 27 bf 23 d9 2f 11 e0 75 40 \textcolor{purple}{\textbf{[ed]}} b0 76 76 0d 64 6d 81 4d 59 4a \textcolor{blue}{\textbf{[76]}} 4c 41 f7 ec \textcolor{magenta}{\textbf{[01]}} 00 ff
  TAPE K: ed 73 a0 f1 27 bf 23 d9 2f 11 e0 75 40 ed b0 76 76 0d \textcolor{yellow}{\textbf{[64]}} d6 07 39 11 9c d2 4c ed 3f 85 0c 00 \textcolor{red}{\textbf{[4c]}}
\textcolor{red}{\textbf{ >> WRITE: 0x75f3 [Tape B + 0x13] <- \textcolor{yellow}{\textbf{0x6d}} (SUCCESS)}}
\textcolor{blue}{------------------------------------------------------------------}
\textcolor{purple}{\textbf{STEP: 33 | CPU_A PC: 0x000d (Addr: 0xd) | Slot: J}}
\textcolor{purple}{ REGS: AF:0000 BC:ffec DE:75f4 HL:0014 IX:0000 IY:0000 SP:ffff}
\textcolor{purple}{ FLGS: S:0 Z:1 Y:0 H:0 X:1 P:1 N:0 C:0} | Energy J:139 K:177
  TAPE J: ed 73 a0 f1 27 bf 23 d9 2f 11 e0 75 40 \textcolor{purple}{\textbf{[ed]}} b0 76 76 0d 64 6d 81 4d 59 4a \textcolor{blue}{\textbf{[76]}} 4c 41 f7 ec \textcolor{magenta}{\textbf{[01]}} 00 ff
  TAPE K: ed 73 a0 f1 27 bf 23 d9 2f 11 e0 75 40 ed b0 76 76 0d 64 \textcolor{yellow}{\textbf{[6d]}} 07 39 11 9c d2 4c ed 3f 85 0c 00 \textcolor{red}{\textbf{[4c]}}
\textcolor{red}{\textbf{ >> WRITE: 0x75f4 [Tape B + 0x14] <- \textcolor{yellow}{\textbf{0x81}} (SUCCESS)}}
\textcolor{blue}{------------------------------------------------------------------}
\textcolor{purple}{\textbf{STEP: 34 | CPU_A PC: 0x000d (Addr: 0xd) | Slot: J}}
\textcolor{purple}{ REGS: AF:0000 BC:ffeb DE:75f5 HL:0015 IX:0000 IY:0000 SP:ffff}
\textcolor{purple}{ FLGS: S:0 Z:1 Y:0 H:0 X:0 P:1 N:0 C:0} | Energy J:138 K:177
  TAPE J: ed 73 a0 f1 27 bf 23 d9 2f 11 e0 75 40 \textcolor{purple}{\textbf{[ed]}} b0 76 76 0d 64 6d 81 4d 59 4a \textcolor{blue}{\textbf{[76]}} 4c 41 f7 ec \textcolor{magenta}{\textbf{[01]}} 00 ff
  TAPE K: ed 73 a0 f1 27 bf 23 d9 2f 11 e0 75 40 ed b0 76 76 0d 64 6d \textcolor{yellow}{\textbf{[81]}} 39 11 9c d2 4c ed 3f 85 0c 00 \textcolor{red}{\textbf{[4c]}}
\textcolor{red}{\textbf{ >> WRITE: 0x75f5 [Tape B + 0x15] <- \textcolor{yellow}{\textbf{0x4d}} (SUCCESS)}}
\textcolor{blue}{------------------------------------------------------------------}
\textcolor{purple}{\textbf{STEP: 35 | CPU_A PC: 0x000d (Addr: 0xd) | Slot: J}}
\textcolor{purple}{ REGS: AF:0000 BC:ffea DE:75f6 HL:0016 IX:0000 IY:0000 SP:ffff}
\textcolor{purple}{ FLGS: S:0 Z:1 Y:0 H:0 X:1 P:1 N:0 C:0} | Energy J:137 K:177
  TAPE J: ed 73 a0 f1 27 bf 23 d9 2f 11 e0 75 40 \textcolor{purple}{\textbf{[ed]}} b0 76 76 0d 64 6d 81 4d 59 4a \textcolor{blue}{\textbf{[76]}} 4c 41 f7 ec \textcolor{magenta}{\textbf{[01]}} 00 ff
  TAPE K: ed 73 a0 f1 27 bf 23 d9 2f 11 e0 75 40 ed b0 76 76 0d 64 6d 81 \textcolor{yellow}{\textbf{[4d]}} 11 9c d2 4c ed 3f 85 0c 00 \textcolor{red}{\textbf{[4c]}}
\textcolor{red}{\textbf{ >> WRITE: 0x75f6 [Tape B + 0x16] <- \textcolor{yellow}{\textbf{0x59}} (SUCCESS)}}
\textcolor{blue}{------------------------------------------------------------------}
\textcolor{purple}{\textbf{STEP: 36 | CPU_A PC: 0x000d (Addr: 0xd) | Slot: J}}
\textcolor{purple}{ REGS: AF:0000 BC:ffe9 DE:75f7 HL:0017 IX:0000 IY:0000 SP:ffff}
\textcolor{purple}{ FLGS: S:0 Z:1 Y:0 H:0 X:1 P:1 N:0 C:0} | Energy J:136 K:177
  TAPE J: ed 73 a0 f1 27 bf 23 d9 2f 11 e0 75 40 \textcolor{purple}{\textbf{[ed]}} b0 76 76 0d 64 6d 81 4d 59 4a \textcolor{blue}{\textbf{[76]}} 4c 41 f7 ec \textcolor{magenta}{\textbf{[01]}} 00 ff
  TAPE K: ed 73 a0 f1 27 bf 23 d9 2f 11 e0 75 40 ed b0 76 76 0d 64 6d 81 4d \textcolor{yellow}{\textbf{[59]}} 9c d2 4c ed 3f 85 0c 00 \textcolor{red}{\textbf{[4c]}}
\textcolor{red}{\textbf{ >> WRITE: 0x75f7 [Tape B + 0x17] <- \textcolor{yellow}{\textbf{0x4a}} (SUCCESS)}}
\textcolor{blue}{------------------------------------------------------------------}
\textcolor{purple}{\textbf{STEP: 37 | CPU_A PC: 0x000d (Addr: 0xd) | Slot: J}}
\textcolor{purple}{ REGS: AF:0000 BC:ffe8 DE:75f8 HL:0018 IX:0000 IY:0000 SP:ffff}
\textcolor{purple}{ FLGS: S:0 Z:1 Y:1 H:0 X:1 P:1 N:0 C:0} | Energy J:135 K:177
  TAPE J: ed 73 a0 f1 27 bf 23 d9 2f 11 e0 75 40 \textcolor{purple}{\textbf{[ed]}} b0 76 76 0d 64 6d 81 4d 59 4a \textcolor{blue}{\textbf{[76]}} 4c 41 f7 ec \textcolor{magenta}{\textbf{[01]}} 00 ff
  TAPE K: ed 73 a0 f1 27 bf 23 d9 2f 11 e0 75 40 ed b0 76 76 0d 64 6d 81 4d 59 \textcolor{yellow}{\textbf{[4a]}} d2 4c ed 3f 85 0c 00 \textcolor{red}{\textbf{[4c]}}
\textcolor{red}{\textbf{ >> WRITE: 0x75f8 [Tape B + 0x18] <- \textcolor{yellow}{\textbf{0x76}} (SUCCESS)}}
\textcolor{blue}{------------------------------------------------------------------}
\textcolor{purple}{\textbf{STEP: 38 | CPU_A PC: 0x000d (Addr: 0xd) | Slot: J}}
\textcolor{purple}{ REGS: AF:0000 BC:ffe7 DE:75f9 HL:0019 IX:0000 IY:0000 SP:ffff}
\textcolor{purple}{ FLGS: S:0 Z:1 Y:1 H:0 X:0 P:1 N:0 C:0} | Energy J:134 K:177
  TAPE J: ed 73 a0 f1 27 bf 23 d9 2f 11 e0 75 40 \textcolor{purple}{\textbf{[ed]}} b0 76 76 0d 64 6d 81 4d 59 4a \textcolor{blue}{\textbf{[76]}} 4c 41 f7 ec \textcolor{magenta}{\textbf{[01]}} 00 ff
  TAPE K: ed 73 a0 f1 27 bf 23 d9 2f 11 e0 75 40 ed b0 76 76 0d 64 6d 81 4d 59 4a \textcolor{yellow}{\textbf{[76]}} 4c ed 3f 85 0c 00 \textcolor{red}{\textbf{[4c]}}
\textcolor{red}{\textbf{ >> WRITE: 0x75f9 [Tape B + 0x19] <- \textcolor{yellow}{\textbf{0x4c}} (SUCCESS)}}
\textcolor{blue}{------------------------------------------------------------------}
\textcolor{purple}{\textbf{STEP: 39 | CPU_A PC: 0x000d (Addr: 0xd) | Slot: J}}
\textcolor{purple}{ REGS: AF:0000 BC:ffe6 DE:75fa HL:001a IX:0000 IY:0000 SP:ffff}
\textcolor{purple}{ FLGS: S:0 Z:1 Y:0 H:0 X:1 P:1 N:0 C:0} | Energy J:133 K:177
  TAPE J: ed 73 a0 f1 27 bf 23 d9 2f 11 e0 75 40 \textcolor{purple}{\textbf{[ed]}} b0 76 76 0d 64 6d 81 4d 59 4a \textcolor{blue}{\textbf{[76]}} 4c 41 f7 ec \textcolor{magenta}{\textbf{[01]}} 00 ff
  TAPE K: ed 73 a0 f1 27 bf 23 d9 2f 11 e0 75 40 ed b0 76 76 0d 64 6d 81 4d 59 4a 76 \textcolor{yellow}{\textbf{[4c]}} ed 3f 85 0c 00 \textcolor{red}{\textbf{[4c]}}
\textcolor{red}{\textbf{ >> WRITE: 0x75fa [Tape B + 0x1a] <- \textcolor{yellow}{\textbf{0x41}} (SUCCESS)}}
\textcolor{blue}{------------------------------------------------------------------}
\textcolor{purple}{\textbf{STEP: 40 | CPU_A PC: 0x000d (Addr: 0xd) | Slot: J}}
\textcolor{purple}{ REGS: AF:0000 BC:ffe5 DE:75fb HL:001b IX:0000 IY:0000 SP:ffff}
\textcolor{purple}{ FLGS: S:0 Z:1 Y:0 H:0 X:0 P:1 N:0 C:0} | Energy J:132 K:177
  TAPE J: ed 73 a0 f1 27 bf 23 d9 2f 11 e0 75 40 \textcolor{purple}{\textbf{[ed]}} b0 76 76 0d 64 6d 81 4d 59 4a \textcolor{blue}{\textbf{[76]}} 4c 41 f7 ec \textcolor{magenta}{\textbf{[01]}} 00 ff
  TAPE K: ed 73 a0 f1 27 bf 23 d9 2f 11 e0 75 40 ed b0 76 76 0d 64 6d 81 4d 59 4a 76 4c \textcolor{yellow}{\textbf{[41]}} 3f 85 0c 00 \textcolor{red}{\textbf{[4c]}}
\textcolor{red}{\textbf{ >> WRITE: 0x75fb [Tape B + 0x1b] <- \textcolor{yellow}{\textbf{0xf7}} (SUCCESS)}}
\textcolor{blue}{------------------------------------------------------------------}
\textcolor{purple}{\textbf{STEP: 41 | CPU_A PC: 0x000d (Addr: 0xd) | Slot: J}}
\textcolor{purple}{ REGS: AF:0000 BC:ffe4 DE:75fc HL:001c IX:0000 IY:0000 SP:ffff}
\textcolor{purple}{ FLGS: S:0 Z:1 Y:1 H:0 X:0 P:1 N:0 C:0} | Energy J:131 K:177
  TAPE J: ed 73 a0 f1 27 bf 23 d9 2f 11 e0 75 40 \textcolor{purple}{\textbf{[ed]}} b0 76 76 0d 64 6d 81 4d 59 4a \textcolor{blue}{\textbf{[76]}} 4c 41 f7 ec \textcolor{magenta}{\textbf{[01]}} 00 ff
  TAPE K: ed 73 a0 f1 27 bf 23 d9 2f 11 e0 75 40 ed b0 76 76 0d 64 6d 81 4d 59 4a 76 4c 41 \textcolor{yellow}{\textbf{[f7]}} 85 0c 00 \textcolor{red}{\textbf{[4c]}}
\textcolor{red}{\textbf{ >> WRITE: 0x75fc [Tape B + 0x1c] <- \textcolor{yellow}{\textbf{0xec}} (SUCCESS)}}
\textcolor{blue}{------------------------------------------------------------------}
\textcolor{purple}{\textbf{STEP: 42 | CPU_A PC: 0x000d (Addr: 0xd) | Slot: J}}
\textcolor{purple}{ REGS: AF:0000 BC:ffe3 DE:75fd HL:001d IX:0000 IY:0000 SP:ffff}
\textcolor{purple}{ FLGS: S:0 Z:1 Y:0 H:0 X:1 P:1 N:0 C:0} | Energy J:130 K:177
  TAPE J: ed 73 a0 f1 27 bf 23 d9 2f 11 e0 75 40 \textcolor{purple}{\textbf{[ed]}} b0 76 76 0d 64 6d 81 4d 59 4a \textcolor{blue}{\textbf{[76]}} 4c 41 f7 ec \textcolor{magenta}{\textbf{[01]}} 00 ff
  TAPE K: ed 73 a0 f1 27 bf 23 d9 2f 11 e0 75 40 ed b0 76 76 0d 64 6d 81 4d 59 4a 76 4c 41 f7 \textcolor{yellow}{\textbf{[ec]}} 0c 00 \textcolor{red}{\textbf{[4c]}}
\textcolor{red}{\textbf{ >> WRITE: 0x75fd [Tape B + 0x1d] <- \textcolor{yellow}{\textbf{0x01}} (SUCCESS)}}
\textcolor{blue}{------------------------------------------------------------------}
\textcolor{purple}{\textbf{STEP: 43 | CPU_A PC: 0x000d (Addr: 0xd) | Slot: J}}
\textcolor{purple}{ REGS: AF:0000 BC:ffe2 DE:75fe HL:001e IX:0000 IY:0000 SP:ffff}
\textcolor{purple}{ FLGS: S:0 Z:1 Y:0 H:0 X:0 P:1 N:0 C:0} | Energy J:129 K:177
  TAPE J: ed 73 a0 f1 27 bf 23 d9 2f 11 e0 75 40 \textcolor{purple}{\textbf{[ed]}} b0 76 76 0d 64 6d 81 4d 59 4a \textcolor{blue}{\textbf{[76]}} 4c 41 f7 ec \textcolor{magenta}{\textbf{[01]}} 00 ff
  TAPE K: ed 73 a0 f1 27 bf 23 d9 2f 11 e0 75 40 ed b0 76 76 0d 64 6d 81 4d 59 4a 76 4c 41 f7 ec \textcolor{yellow}{\textbf{[01]}} 00 \textcolor{red}{\textbf{[4c]}}
\textcolor{red}{\textbf{ >> WRITE: 0x75fe [Tape B + 0x1e] <- \textcolor{yellow}{\textbf{0x00}} (SUCCESS)}}
\textcolor{blue}{------------------------------------------------------------------}
\textcolor{purple}{\textbf{STEP: 44 | CPU_A PC: 0x000d (Addr: 0xd) | Slot: J}}
\textcolor{purple}{ REGS: AF:0000 BC:ffe1 DE:75ff HL:001f IX:0000 IY:0000 SP:ffff}
\textcolor{purple}{ FLGS: S:0 Z:1 Y:0 H:0 X:0 P:1 N:0 C:0} | Energy J:128 K:177
  TAPE J: ed 73 a0 f1 27 bf 23 d9 2f 11 e0 75 40 \textcolor{purple}{\textbf{[ed]}} b0 76 76 0d 64 6d 81 4d 59 4a \textcolor{blue}{\textbf{[76]}} 4c 41 f7 ec \textcolor{magenta}{\textbf{[01]}} 00 ff
  TAPE K: ed 73 a0 f1 27 bf 23 d9 2f 11 e0 75 40 ed b0 76 76 0d 64 6d 81 4d 59 4a 76 4c 41 f7 ec 01 \textcolor{yellow}{\textbf{[00]}} \textcolor{red}{\textbf{[4c]}}
\textcolor{red}{\textbf{ >> WRITE: 0x75ff [Tape B + 0x1f] <- \textcolor{yellow}{\textbf{0xff}} (SUCCESS)}}
\textcolor{blue}{------------------------------------------------------------------}
\textcolor{purple}{\textbf{STEP: 45 | CPU_A PC: 0x000d (Addr: 0xd) | Slot: J}}
\textcolor{purple}{ REGS: AF:0000 BC:ffe0 DE:7600 HL:0020 IX:0000 IY:0000 SP:ffff}
\textcolor{purple}{ FLGS: S:0 Z:1 Y:1 H:0 X:1 P:1 N:0 C:0} | Energy J:127 K:177
  TAPE J: ed 73 a0 f1 27 bf 23 d9 2f 11 e0 75 40 \textcolor{purple}{\textbf{[ed]}} b0 76 76 0d 64 6d 81 4d 59 4a \textcolor{blue}{\textbf{[76]}} 4c 41 f7 ec \textcolor{magenta}{\textbf{[01]}} 00 ff
  TAPE K: ed 73 a0 f1 27 bf 23 d9 2f 11 e0 75 40 ed b0 76 76 0d 64 6d 81 4d 59 4a 76 4c 41 f7 ec 01 00 \textcolor{yellow}{\textbf{[ff]}}
\textcolor{red}{\textbf{ >> WRITE: 0x7600 [Tape A + 0x0] <- \textcolor{yellow}{\textbf{0xed}} (SUCCESS)}}
\textcolor{blue}{------------------------------------------------------------------}
\textcolor{purple}{\textbf{STEP: 46 | CPU_A PC: 0x000d (Addr: 0xd) | Slot: J}}
\textcolor{purple}{ REGS: AF:0000 BC:ffdf DE:7601 HL:0021 IX:0000 IY:0000 SP:ffff}
\textcolor{purple}{ FLGS: S:0 Z:1 Y:0 H:0 X:1 P:1 N:0 C:0} | Energy J:126 K:177
  TAPE J: \textcolor{yellow}{\textbf{[ed]}} 73 a0 f1 27 bf 23 d9 2f 11 e0 75 40 \textcolor{purple}{\textbf{[ed]}} b0 76 76 0d 64 6d 81 4d 59 4a \textcolor{blue}{\textbf{[76]}} 4c 41 f7 ec \textcolor{magenta}{\textbf{[01]}} 00 ff
  TAPE K: ed 73 a0 f1 27 bf 23 d9 2f 11 e0 75 40 ed b0 76 76 0d 64 6d 81 4d 59 4a 76 4c 41 f7 ec 01 00 \textcolor{red}{\textbf{[ff]}}
\textcolor{red}{\textbf{ >> WRITE: 0x7601 [Tape A + 0x1] <- \textcolor{yellow}{\textbf{0x73}} (SUCCESS)}}
\textcolor{blue}{------------------------------------------------------------------}
\textcolor{purple}{\textbf{STEP: 47 | CPU_A PC: 0x000d (Addr: 0xd) | Slot: J}}
\textcolor{purple}{ REGS: AF:0000 BC:ffde DE:7602 HL:0022 IX:0000 IY:0000 SP:ffff}
\textcolor{purple}{ FLGS: S:0 Z:1 Y:1 H:0 X:0 P:1 N:0 C:0} | Energy J:125 K:177
  TAPE J: ed \textcolor{yellow}{\textbf{[73]}} a0 f1 27 bf 23 d9 2f 11 e0 75 40 \textcolor{purple}{\textbf{[ed]}} b0 76 76 0d 64 6d 81 4d 59 4a \textcolor{blue}{\textbf{[76]}} 4c 41 f7 ec \textcolor{magenta}{\textbf{[01]}} 00 ff
  TAPE K: ed 73 a0 f1 27 bf 23 d9 2f 11 e0 75 40 ed b0 76 76 0d 64 6d 81 4d 59 4a 76 4c 41 f7 ec 01 00 \textcolor{red}{\textbf{[ff]}}
\textcolor{red}{\textbf{ >> WRITE: 0x7602 [Tape A + 0x2] <- \textcolor{yellow}{\textbf{0xa0}} (SUCCESS)}}
\textcolor{blue}{------------------------------------------------------------------}
\textcolor{purple}{\textbf{STEP: 48 | CPU_A PC: 0x000d (Addr: 0xd) | Slot: J}}
\textcolor{purple}{ REGS: AF:0000 BC:ffdd DE:7603 HL:0023 IX:0000 IY:0000 SP:ffff}
\textcolor{purple}{ FLGS: S:0 Z:1 Y:0 H:0 X:0 P:1 N:0 C:0} | Energy J:124 K:177
  TAPE J: ed 73 \textcolor{yellow}{\textbf{[a0]}} f1 27 bf 23 d9 2f 11 e0 75 40 \textcolor{purple}{\textbf{[ed]}} b0 76 76 0d 64 6d 81 4d 59 4a \textcolor{blue}{\textbf{[76]}} 4c 41 f7 ec \textcolor{magenta}{\textbf{[01]}} 00 ff
  TAPE K: ed 73 a0 f1 27 bf 23 d9 2f 11 e0 75 40 ed b0 76 76 0d 64 6d 81 4d 59 4a 76 4c 41 f7 ec 01 00 \textcolor{red}{\textbf{[ff]}}
\textcolor{red}{\textbf{ >> WRITE: 0x7603 [Tape A + 0x3] <- \textcolor{yellow}{\textbf{0xf1}} (SUCCESS)}}
\textcolor{blue}{------------------------------------------------------------------}
\textcolor{purple}{\textbf{STEP: 49 | CPU_A PC: 0x000d (Addr: 0xd) | Slot: J}}
\textcolor{purple}{ REGS: AF:0000 BC:ffdc DE:7604 HL:0024 IX:0000 IY:0000 SP:ffff}
\textcolor{purple}{ FLGS: S:0 Z:1 Y:0 H:0 X:0 P:1 N:0 C:0} | Energy J:123 K:177
  TAPE J: ed 73 a0 \textcolor{yellow}{\textbf{[f1]}} 27 bf 23 d9 2f 11 e0 75 40 \textcolor{purple}{\textbf{[ed]}} b0 76 76 0d 64 6d 81 4d 59 4a \textcolor{blue}{\textbf{[76]}} 4c 41 f7 ec \textcolor{magenta}{\textbf{[01]}} 00 ff
  TAPE K: ed 73 a0 f1 27 bf 23 d9 2f 11 e0 75 40 ed b0 76 76 0d 64 6d 81 4d 59 4a 76 4c 41 f7 ec 01 00 \textcolor{red}{\textbf{[ff]}}
\textcolor{red}{\textbf{ >> WRITE: 0x7604 [Tape A + 0x4] <- \textcolor{yellow}{\textbf{0x27}} (SUCCESS)}}
\textcolor{blue}{------------------------------------------------------------------}
\end{Verbatim}

\newpage
\section*{NeurIPS Paper Checklist}

\begin{enumerate}

\item {\bf Claims}
    \item[] Question: Do the main claims made in the abstract and introduction accurately reflect the paper's contributions and scope?
    \item[] Answer: \answerYes{} %
    \item[] Justification: The abstract and introduction explicitly state the primary contributions, which include formalizing Autopoietic Game Theory, identifying stability conditions for well-mixed cooperation, and demonstrating extensibility to downstream tasks. These claims are directly supported by the empirical results in Section 4 and the theoretical proofs in Appendix B. The paper accurately scopes these findings to the specific Z80 computational substrate and the tested energy-based social dilemma.

\item {\bf Limitations}
    \item[] Question: Does the paper discuss the limitations of the work performed by the authors?
    \item[] Answer: \answerYes{} %
    \item[] Justification: The authors include a dedicated Limitations paragraph within the Discussion section. This section explicitly addresses the difficulty of extending the theoretical bounds to spatial gradients and continuous mutations. It also acknowledges the reliance on specific substrates and notes the strong assumption of unimpeded read and write access to opponent memory, reflecting on how restricting state access could alter the observed dynamics.

\item {\bf Theory assumptions and proofs}
    \item[] Question: For each theoretical result, does the paper provide the full set of assumptions and a complete (and correct) proof?
    \item[] Answer: \answerYes{} %
    \item[] Justification: The paper provides theoretical results in Appendix B. Section B.1 explicitly states the formal model definitions, detailing environmental constants, kinetic execution rules, and energetic boundaries. Complete mathematical proofs are provided for Lemma 1, Lemma 2, Lemma 3, and Theorem 1. The main text appropriately references these appendices when discussing stability conditions.

    \item {\bf Experimental result reproducibility}
    \item[] Question: Does the paper fully disclose all the information needed to reproduce the main experimental results of the paper to the extent that it affects the main claims and/or conclusions of the paper (regardless of whether the code and data are provided or not)?
    \item[] Answer: \answerYes{} %
    \item[] Justification: The paper details the experimental setup in Section 3 and provides comprehensive reproduction instructions in Appendix C. This includes a reference to the baseline open-source Z80 emulator repository, the specific custom opcodes added for social operations, a full table of hyperparameters, and simulation durations. Furthermore, it provides precise mathematical definitions for the continuous evaluation metrics, including Higher-Order Entropy and Vector Symbolic Architectures.

\item {\bf Open access to data and code}
    \item[] Question: Does the paper provide open access to the data and code, with sufficient instructions to faithfully reproduce the main experimental results, as described in supplemental material?
    \item[] Answer: \answerNo{} %
    \item[] Justification: Due to open-source restrictions with our organization, we are unable to release all of the code immediately, however, we are committed to trying to receive approval to release the code relevant to the project, especially the intuition-building experiments in Section 5.

\item {\bf Experimental setting/details}
    \item[] Question: Does the paper specify all the training and test details (e.g., data splits, hyperparameters, how they were chosen, type of optimizer) necessary to understand the results?
    \item[] Answer: \answerYes{} %
    \item[] Justification: The core experimental mechanics are introduced in Section 3, while Appendix C provides a comprehensive breakdown of the experimental protocols. This includes Table 1, which lists all base hyperparameters (e.g., number of seeds, grid side length, initial energy budgets, mutation rarity), as well as the exact simulation durations (number of epochs) used across different experimental modes.

\item {\bf Experiment statistical significance}
    \item[] Question: Does the paper report error bars suitably and correctly defined or other appropriate information about the statistical significance of the experiments?
    \item[] Answer: \answerYes{} %
    \item[] Justification: The paper consistently tests and reports the statistical significance of its main empirical claims. It explicitly details the use of 2-sided t-tests (e.g., $p<0.001$, $p\le0.01$) to compare the means of final metrics (like higher-order entropy, edit distance, and joint solve rates) across the 10 independent random seeds specified in the experimental setup.

\item {\bf Experiments compute resources}
    \item[] Question: For each experiment, does the paper provide sufficient information on the computer resources (type of compute workers, memory, time of execution) needed to reproduce the experiments?
    \item[] Answer: \answerYes{} %
    \item[] Justification: In Appendix C, we describe that every job with a hyperparameter configuration ran with 256 CPU cores running in parallel.
    
\item {\bf Code of ethics}
    \item[] Question: Does the research conducted in the paper conform, in every respect, with the NeurIPS Code of Ethics \url{https://neurips.cc/public/EthicsGuidelines}?
    \item[] Answer: \answerYes{} %
    \item[] Justification: The research relies entirely on theoretical modeling and simulated artificial life environments using randomly generated Z80 machine code. It does not involve human subjects, personally identifiable information, or hazardous materials, and it fully conforms to the NeurIPS Code of Ethics.

\item {\bf Broader impacts}
    \item[] Question: Does the paper discuss both potential positive societal impacts and negative societal impacts of the work performed?
    \item[] Answer: \answerYes{} %
    \item[] Justification: The paper discusses the broader implications of its work in the context of AI safety and multi-agent alignment. The Introduction and the Discussion highlights the necessity of understanding these dynamics to predict failure modes and design safe incentive structures for autonomous multi-agent AI systems as they scale.
    
\item {\bf Safeguards}
    \item[] Question: Does the paper describe safeguards that have been put in place for responsible release of data or models that have a high risk for misuse (e.g., pre-trained language models, image generators, or scraped datasets)?
    \item[] Answer: \answerNA{}
    \item[] Justification: The research utilizes randomly generated Z80 machine code within a localized artificial life simulation. The paper does not release large language models, generative image models, or internet-scraped datasets that pose a substantial risk for misuse or dual-use.

\item {\bf Licenses for existing assets}
    \item[] Question: Are the creators or original owners of assets (e.g., code, data, models), used in the paper, properly credited and are the license and terms of use explicitly mentioned and properly respected?
    \item[] Answer: \answerYes{}
    \item[] Justification: The paper acknowledges the use of an existing open-source Z80 emulator in Appendix C.2 and provides the direct URL to the repository to credit the original creators.

\item {\bf New assets}
    \item[] Question: Are new assets introduced in the paper well documented and is the documentation provided alongside the assets?
    \item[] Answer: \answerNA{}
    \item[] Justification: The paper does not explicitly link or release new code repositories, datasets, or models within the submitted text. Instead, it provides extensive documentation of the simulation framework and custom instruction sets in Appendix C to support reproducibility.

\item {\bf Crowdsourcing and research with human subjects}
    \item[] Question: For crowdsourcing experiments and research with human subjects, does the paper include the full text of instructions given to participants and screenshots, if applicable, as well as details about compensation (if any)? 
    \item[] Answer: \answerNA{} %
    \item[] Justification: No human experiments were conducted.

\item {\bf Institutional review board (IRB) approvals or equivalent for research with human subjects}
    \item[] Question: Does the paper describe potential risks incurred by study participants, whether such risks were disclosed to the subjects, and whether Institutional Review Board (IRB) approvals (or an equivalent approval/review based on the requirements of your country or institution) were obtained?
    \item[] Answer: \answerNA{} %
    \item[] Justification: No human experiments were conducted.

\item {\bf Declaration of LLM usage}
    \item[] Question: Does the paper describe the usage of LLMs if it is an important, original, or non-standard component of the core methods in this research? Note that if the LLM is used only for writing, editing, or formatting purposes and does \emph{not} impact the core methodology, scientific rigor, or originality of the research, declaration is not required.
    \item[] Answer: \answerNA{} %
    \item[] Justification: The core method did not involve LLMs as an important feature.

\end{enumerate}

\end{document}